\documentclass[journal]{IEEEtran}
\usepackage{blindtext}
\usepackage{tikz}
\usepackage{threeparttable}
\usepackage[ruled,linesnumbered]{algorithm2e}
\usepackage{diagbox}
\usepackage{multirow}
\usepackage{algpseudocode}
\usepackage{color}
\usepackage{subfigure}
\usepackage{amssymb}
\usepackage{cite}
\usepackage{caption}
\usepackage{xcolor,soul,framed} 

\usepackage[cmex10]{amsmath}
\usepackage{array}
\usepackage{mdwmath}
\usepackage{mdwtab}
\usepackage{eqparbox}
\usepackage{url}

\newtheorem{theorem}{\textbf{Theorem}}
\newtheorem{definition}{\textbf{Definition}}
\newtheorem{lemma}{\textbf{Lemma}}

\newtheorem{assumption}{\textbf{Assumption}}

\begin{document}
\bstctlcite{IEEEexample:BSTcontrol}
    \title{FedDef: Defense Against Gradient Leakage in Federated Learning-based Network Intrusion Detection Systems}
    \author{Jiahui~Chen,
    Yi~Zhao,~\IEEEmembership{Member,~IEEE,}
    Qi~Li,~\IEEEmembership{Senior Member,~IEEE,}
    Xuewei~Feng,\\
    and Ke~Xu,~\IEEEmembership{Senior Member,~IEEE}
    \thanks{Manuscript received 31 October 2022; revised 11 June 2023; accepted 3 July 2023. Date of publication x xxx xxxx; date of current version 14 July 2023. This work was supported in part by the National Science Foundation for Distinguished Young Scholars of China with No. 61825204, National Natural Science Foundation of China with No. 61932016, No. 62132011 and No. 62202258, Beijing Outstanding Young Scientist Program with No. BJJWZYJH01201910003011, China Postdoctoral Science Foundation with No. 2021M701894, China National Postdoctoral Program for Innovative Talents, and Shuimu Tsinghua Scholar Program. The associate editor coordinating the review of this manuscript and approving it for publication was Dr. George Theodorakopoulos. \textit{(Corresponding authors: Yi Zhao; Ke Xu.)}}
    \thanks{Jiahui Chen, Yi Zhao, Xuewei Feng, and  Ke Xu are with the Department of Computer Science and Technology, Tsinghua University, Beijing 100084, China (e-mail: chenjiah22@mails.tsinghua.edu.cn;
    zhao\_yi@tsinghua.edu.cn;
    \text{fengxw06@126.com}; xuke@tsinghua.edu.cn).}
    \thanks{Qi Li is with the Institute for Network Sciences and Cyberspace, Tsinghua University, Beijing 100084, China (e-mail: qli01@tsinghua.edu.cn).}
    \thanks{Digital Object Identifier xxx/TIFS.xx.xx}
}

\IEEEpubid{
\begin{minipage}{\textwidth}\ \\[5pt]
\centering
1556--6021~\copyright~2023 IEEE. Personal use is permitted, but republication/redistribution requires IEEE permission.\\
See https://www.ieee.org/publications/rights/index.html for more information.
\end{minipage}
}


\markboth{IEEE TRANSACTIONS ON INFORMATION FORENSICS AND SECURITY}
{CHEN \MakeLowercase{et al.}: FedDef: Defense Against Gradient Leakage in Federated Learning-based Network Intrusion Detection Systems}

\maketitle

\begin{abstract}
Deep learning (DL) methods have been widely applied to anomaly-based network intrusion detection system (NIDS) to detect malicious traffic.
To expand the usage scenarios of DL-based methods, federated learning (FL) allows multiple users to train a global model on the basis of respecting individual data privacy.
However, it has not yet been systematically evaluated how robust FL-based NIDSs are against existing privacy attacks under existing defenses.
To address this issue, we propose two privacy evaluation metrics designed for FL-based NIDSs, including (1) privacy score that evaluates the similarity between the original and recovered traffic features using reconstruction attacks, and (2) evasion rate against NIDSs using adversarial attack with the recovered traffic.
We conduct experiments to illustrate that existing defenses provide little protection and the corresponding adversarial traffic can even evade the SOTA NIDS Kitsune.
To defend against such attacks and build a more robust FL-based NIDS, we further propose FedDef, a novel optimization-based input perturbation defense strategy with theoretical guarantee.
It achieves both high utility by minimizing the gradient distance and strong privacy protection by maximizing the input distance.
We experimentally evaluate four existing defenses on four datasets and show that our defense outperforms all the baselines in terms of privacy protection with up to 7 times higher privacy score, while maintaining model accuracy loss within 3\% under optimal parameter combination.
\end{abstract}


\IEEEpeerreviewmaketitle

\maketitle
\section{Introduction}
\IEEEPARstart{D}{eep} learning-based  network intrusion detection system (NIDS) has been widely used to detect malicious traffic and intrusion attacks, they are expected to raise alarms whenever the incoming traffic carries malicious properties (e.g., scan user's ports) or induces attacks (e.g., DDoS).
Recently, researchers have been trying to adopt federated learning (FL), where multiple users collaboratively exchange information and train a global model with publicly shared gradients, to derive more accurate detection of cyberattacks without privacy leakage.
For example, \cite{zhao2022collaboration} proposes the collaboration-enabled intelligent Internet architecture for malicious traffic detection.
\cite{wagner2021united} proposes a decentralized federated architecture that collaborates with eleven Internet Exchange Points (IXPs) operating in three different regions.
It allows participant mitigation platforms to exchange information about ongoing amplification DDoS attacks.
Experimental results illustrate that such collaboration can detect and mitigate up to 90\% more DDoS attacks locally, which proves great success for faster and more effective detection and neutralization of attack traffic among collaborative networks.
\cite{nguyen2019diot} also proposes an autonomous FL-based anomaly detection system to locate compromised Internet of Things (IoT) devices with high detection rate (95.6\%) and fast inference ($257ms$), while \cite{li2020deepfed} further proposes a FL-based NIDS for industrial cyber–physical systems in real world.
These new researches indicate a promising trend that combination with FL can improve the overall detection performance of NIDS and is receiving extensive attention.

\IEEEpubidadjcol

However, sharing model updates or gradients also makes FL vulnerable to inference attack in computer vision (CV) domain, e.g., property inference attack that infers sensitive properties of training data \cite{melis2019exploiting} and model inversion attack that reconstructs image training data \cite{fredrikson2015model,geiping2020inverting,wang2019beyond,ferrag2020deep}.
Accordingly, there have also been some defense strategies such as differential privacy for deep learning \cite{abadi2016deep} and Soteria \cite{sun2021soteria} that perturbs the data representation features for images.
However, unlike CV domain where slight noises added to the images can induce totally different visual perception and thus have higher tolerance for privacy, the reconstruction of traffic data may be more intimidating, because adversaries can evade the target model via (1) black-box attack that trains Generative Adversarial Network (GAN) model with the reconstructed benign data to generate new malicious traffic from random noises, or (2) white-box attack that directly perturbs the reconstructed malicious data.
Unfortunately, few researches have systematically investigated to what extent current defenses combined with NIDS can protect user privacy, and whether there exists defense strategy that achieves both high utility and strong privacy guarantee to build a more robust FL-based NIDS.

To derive a more accurate evaluation of privacy for FL-based NIDSs, we propose two privacy metrics specifically designed for NIDS domain, i.e., privacy score and evasion rate.
For the first one, we leverage reconstruction attacks, i.e., inversion and extraction attack, to recover the original training data from model gradients.
Then we can calculate the similarity between raw and reconstructed data to evaluate privacy leakage, where we use different distance metrics for continuous (e.g., traffic duration) and discrete features (e.g., protocol type).
With enough reconstructed traffic, we can evaluate evasion rate by training GAN model or applying perturbation to generate adversarial traffic to attack other NIDSs, which presents the practical threats in real world.
With the above privacy metrics, we evaluate existing defenses combined with FL-based NIDS and demonstrate that all baselines fail to provide sufficient privacy protection and the adversarial traffic can even evade the SOTA NIDS Kitsune \cite{mirsky2018kitsune}, which urges for new effective defense strategy.

To bridge this gap and build a more robust FL-based NIDS defending against gradient leakage, we further propose a novel input perturbation-based defense approach with theoretical guarantee for model convergence and data privacy named FedDef, which optimizes an objective function to transform the original input such that: 1) the distance between the new input and the raw input is as far as possible to prevent privacy leakage and 2) the corresponding gradients are as similar as possible to maintain model performance.
Experimental results on four datasets illustrate that our defense can mitigate both reconstruction attacks and achieve at most 7 times higher privacy score compared to the second best defense, and the following adversarial attack fails to evade other NIDSs, which significantly outperforms other baselines.
Regarding model performance, the FL model can still converge under our defense with iid and non-iid data distribution, and that model accuracy can be guaranteed within at most 3\% loss with the optimal parameter combination.

\textbf{Contributions}.
We summarize our contributions as follows:

$\bullet$
We propose two privacy evaluation metrics designed for FL-based NIDS, including (1) privacy score that evaluates the similarity between raw traffic feature and the recovered feature using reconstruction attacks, and (2) evasion rate against NIDSs using black-box and white-box adversarial attacks.

$\bullet$
To the best of our knowledge, we are the first to propose FedDef, an optimization-based input perturbation defense scheme for FL-based NIDS to prevent privacy leakage while maintaining model performance.
To enhance our method, we also provide theoretical analysis for model convergence under non-iid data distribution and privacy guarantee for our defense.

$\bullet$
Experimental results on four datasets illustrate that our proposed FedDef outperforms existing defense approaches in terms of privacy protection. With FedDef, the privacy score is 1.5-7 times higher than the second best baseline during early training stage, and evasion rate on Kitsune model is always 0, while other baselines induce successful evasion more or less.
In the meantime, FedDef still maintains high model utility with up to 3\% accuracy loss compared to no-defense baseline under optimal parameter setting.

The rest of the paper is organized as follows:
In Section~\ref{preliminary}, we introduce some background and motivation.
Section~\ref{threat_model} provides our threat model and the corresponding privacy metrics with evaluation example.
In Section~\ref{method}, we introduce our defense design and detailed implementation with a representative FL framework.
We provide theoretical analysis on convergence and privacy guarantee for our defense in Section~\ref{theoretical_analysis}.
In Section~\ref{evaluation}, we evaluate our defense with four baselines on four datasets with respect to privacy preserving and model performance.
We discuss some future work in Section~\ref{future_work} and conclude this study in Section~\ref{conclusion}.

\section{BACKGROUND AND MOTIVATION}
\label{preliminary}

In this section, we introduce some background of FL-based NIDS, two SOTA reconstruction attacks and current defenses accordingly, and then we present the motivation of our study.

\subsection{FL-based NIDSs}
\label{federated_learning}

FL-based NIDS is a new promising topic that combines FL~\cite{konevcny2016federated} like FedAvg~\cite{mcmahan2017communication} with intrusion detection.
Local users first extract features from their private traffic data and then update the global model with the derived gradients, which are aggregated at the trusted server for later distribution.
For example, \cite{mothukuri2021federated} proposes a FL-based anomaly detection approach to proactively recognize intrusion in IoT networks using decentralized on-device data.
Some work also consider hierarchical and interpretable FL to further build a practical NIDS \cite{dong2022interpretable,wang2021towards}.
However, few have investigated the robustness of FL-based NIDSs against privacy attacks, which hinders the deployment of secure NIDSs.

\newcommand*\emptycirc[1][1ex]{\tikz\draw (0,0) circle (#1);} 
\newcommand*\halfcirc[1][1ex]{%
\begin{tikzpicture}
\draw[fill] (0,0)-- (90:#1) arc (90:270:#1) -- cycle ;
\draw (0,0) circle (#1);
\end{tikzpicture}}
\newcommand*\fullcirc[1][1ex]{\tikz\fill (0,0) circle (#1);} 

\begin{table*}[htbp]
\caption{A summary of existing defense approaches.}
\label{sum_defense}
\resizebox{\linewidth}{!}
{
\begin{tabular}{c|c|c|c|c|c|c|c}
\hline
Approach & Category              & Setting & Privacy & Utility & Theoretical Guarantee & Scalability & Feasibility \\ \hline
 MPC\cite{bohler2021secure}        & MPC                   &   \fullcirc      &   \fullcirc       &   \fullcirc       &  \fullcirc                      &    \halfcirc         &    \emptycirc         \\ 
 DP\cite{dwork2006calibrating,abadi2016deep,dong2022interpretable}        & Gradient Perturbation &   \fullcirc      &   \halfcirc      &  \halfcirc       &  \fullcirc                     &    \fullcirc         &    \fullcirc         \\ 
   GP\cite{zhu2019deep}      & Gradient Perturbation &    \halfcirc     &  \halfcirc       &  \fullcirc                     &  \emptycirc           &    \fullcirc   &\fullcirc      \\ 
  Lossless\cite{yang2022accuracy}       & Gradient Perturbation &   \halfcirc      &  \fullcirc       &   \fullcirc      &     \fullcirc                  &  \halfcirc           &   \fullcirc          \\ 
  Instahide\cite{huang2020instahide}       & Input Perturbation    &   \halfcirc      &  \halfcirc       &   \halfcirc      &        \emptycirc               &     \fullcirc        &   \halfcirc          \\ 
 Soteria\cite{sun2021soteria}        & Input Perturbation    &  \halfcirc       &   \emptycirc      &  \fullcirc       &    \fullcirc                   &   \halfcirc          &   \halfcirc          \\ 
  ATS \cite{gao2021privacy}       & Input Perturbation    &    \halfcirc     &   \halfcirc      &   \fullcirc      &      \emptycirc                 &    \emptycirc         &    \halfcirc         \\
  DCS \cite{wu2022defense}       & Input Perturbation    &    \fullcirc     &   \halfcirc      &   \fullcirc      &      \emptycirc                 &    \halfcirc         &    \halfcirc         \\
  FedDef (Ours)       & Input Perturbation    &  \fullcirc       &    \fullcirc     &     \fullcirc    &    \fullcirc                   &     \fullcirc        &    \fullcirc         \\ \hline
\end{tabular}
}
\begin{tablenotes}
\item \textbf{Notes} \fullcirc, \halfcirc, \emptycirc \ mean that the approach greatly, partly, barely considers setting metric or provides certain ability (other 5 metrics), respectively.
\end{tablenotes}
\end{table*}

\begin{figure}
\setlength{\abovecaptionskip}{0pt}
\setlength{\belowcaptionskip}{0pt}
\centering
\subfigure{
\begin{minipage}{0.6\linewidth}
\includegraphics[width=1\linewidth]{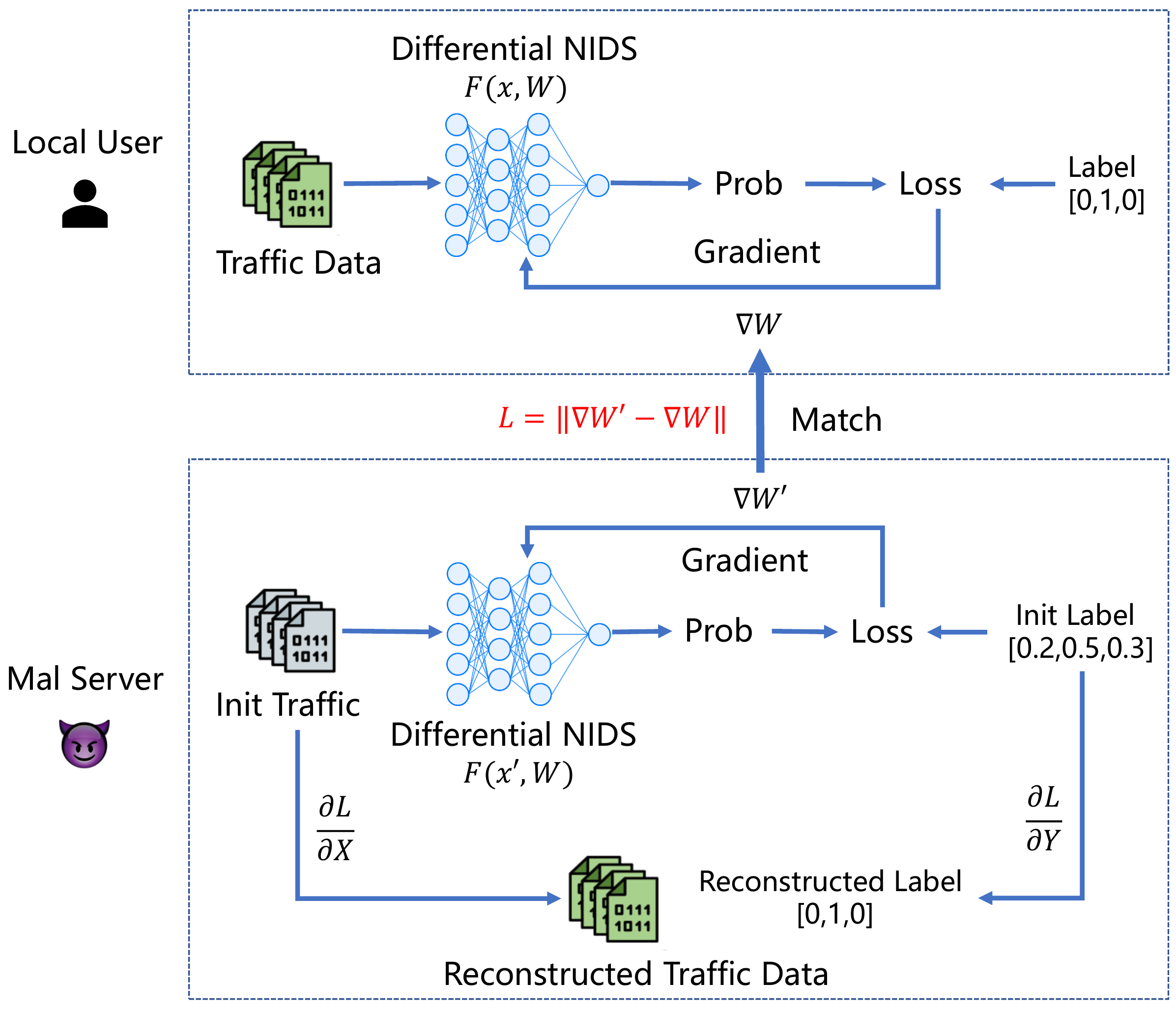}
\end{minipage}}
\caption{Overview of inversion reconstruction attack.}
\label{overview_attack}
\end{figure}

\subsection{Reconstruction Attack}
\label{reconstrcution_attack}

FL also faces privacy leakage issues, such as data reconstruction \cite{zhu2019deep}, membership inference \cite{carlini2022membership}, and attribute inference \cite{mehnaz2022your}.
We focus on the severest reconstruction attack that aims to recover training samples.
It can be categorized into two types, i.e., optimization-based inversion attack and accurate extraction attack.

\subsubsection{Inversion attack}
DLG (i.e., Deep Leakage from Gradients)~\cite{zhu2019deep} solves an optimization problem to obtain the raw features and labels, which can be formulated as:
\begin{equation}
\label{L2_optimization}
    argmin_{x^*,y^*} \ \ \Vert\nabla \theta(x^*,y^*)-\nabla \theta(x,y)\Vert_2
\end{equation}
where $\nabla \theta (x,y)=\frac{\partial \ell(\theta, x, y)}{\partial \theta}$, $\ell$ is the local overall loss function on differentiable deep learning model (e.g., DNN), $\theta$ is the overall model parameter, and $(x,y)$ is the original feature and label.
Fig.~\ref{overview_attack} illustrates the process of inversion attack in FL, where malicious server (or user in decentralized mode) tries to match the publicly shared gradient with the dummy gradient.
Specifically, they first randomly initialize a dummy input $x^*$ and $y^*$. Then, they iteratively optimize the dummy gradients $\nabla \theta(x^*,y^*)$ close as to original by Euclidean distance, which also makes the dummy data close to the real training data.

Follow-up works improve on these results by using different distance metrics such as cosine similarity \cite{geiping2020inverting}:
\begin{equation}
\label{cos_optimization}
    argmin_{x^*,y^*} \ \ 1-\frac{\langle \nabla \theta(x^*,y^*),\nabla \theta(x,y) \rangle }{\Vert\nabla \theta(x^*,y^*)\Vert\Vert\nabla \theta(x,y)\Vert}
\end{equation}
where $\langle \cdot, \cdot \rangle$ means inner product.
\cite{zhao2020idlg} also proposes an improved DLG (i.e., iDLG) by analytically and accurately extracting ground-truth label from the last fully-connected (FC) layer, and thus it performs better as objective function (\ref{L2_optimization}) only has to optimize $x^*$.

\subsubsection{Extraction attack}
To further improve the reconstruction accuracy, researches \cite{geiping2020inverting,boenisch2021curious} propose an accurate extraction attack which can almost perfectly reconstruct a single training sample without any costs.
This kind of attack assumes the model contains at least one layer, which consists of a weight and a bias.
Specifically, we assume the first layer contains both weight and bias (\cite{geiping2020inverting} gives proof for more generalized scenario), then the direct output of the first layer $y$ is computed as $W^Tx+b$, where $x$ is the raw input data and $(W,b)$ is the corresponding weight and bias pair with compatible dimensionality.
We further let $i$ be the row of the first layer such that $\frac{\partial \ell }{\partial b_i}\neq 0$, then by chain rule we can obtain $x$ as $x^T=(\frac{\partial \ell }{\partial b_i} )^{-1}\frac{\partial  \ell }{\partial W_i^T}$.
In this way, we can extract the exact data as long as there exists $\frac{\partial \ell }{\partial b_i}\neq 0$.
In Section~\ref{threat_model}, we present a specific solution of $x^*$ using PyTorch and give proof for its failure when batch size of input data $x$ is more than 1.

\subsection{Defenses}
\label{defenses}

There exist several defenses to mitigate such reconstruction attacks in FL and can be categorized into three types: secure multi-party computation (MPC), gradient perturbation, and input perturbation.
Existing representative defense approaches are summarized in TABLE~\ref{sum_defense} in terms of 6 metrics, where setting describes iid or non-iid training mode, utility evaluates the model performance, scalability denotes whether such defense can be applied in different domains, and feasibility denotes whether it requires extra setup or data to perform such defense.

\subsubsection{Secure multi-party computation}
MPC~\cite{danner2015fully, bohler2021secure} allows a group of parties to synchronously compute a function and obtain accurate representations of the final value while protecting their private data from privacy leakage.
For example, \cite{bonawitz2017practical} proposes to let users encrypt their local updates such that the central server can only recover the aggregation of the updates.
However, MPC requires special setup and can be costly to implement, and \cite{melis2019exploiting} shows that adversaries can still launch inference attacks against MPC for certain scenarios.

\subsubsection{Gradients perturbation}
These defenses modify the gradients before updating them to the server.
For example, \cite{zhu2019deep} proposes gradient pruning (GP) that sets certain percent of the derived gradients of small absolute magnitudes to zero, which serves as a gradient mask to hide data information.
Differential privacy (DP)~\cite{dwork2006calibrating}
provides theoretical privacy guarantee and bounds the change in output distribution caused by a small input noise like Gaussian noise \cite{abadi2016deep,dong2022interpretable}.
\cite{yang2022accuracy} also proposes to add and eliminate random noises to the global model from the server to prevent information leakage from malicious clients, yet it works only when the server is trusted.

\subsubsection{Input perturbation}
These defenses tend to perturb or mix the original data from the source.
For example, \cite{huang2020instahide} proposes Instahide to encode the training data with user's private dataset instead of blind noises.
\cite{sun2021soteria} proposes Soteria to perturb the image feature representation before the defended layer to confuse the reconstruction attack.
In the meantime, \cite{gao2021privacy} proposes to search for optimal image transformation combination such as image rotation and shift to preserve privacy.
However, there is no one-size-fits-all traffic transformation pattern and may not apply to NIDS domain.
\cite{wu2022defense} proposes to obfuscate the gradients of sensitive data with concealing data to preserve privacy, yet it limits the percentage of sensitive data within a batch and does not consider label protection.

\subsection{Motivation: Why privacy protection for FL-NIDS}
\label{motivation}

In FL, if the training data are images or contexts that carry high-level sensitive information, the adversary can directly reconstruct the data and manipulate such privacy. 
While for FL-NIDS, the training data are tabular features with low-level information, which can be less straightforward to understand.
However,
traffic data reconstruction can also pose great threats because those features contain unique patterns tailored to the clients, which can be manipulated for adversarial attacks.
\subsubsection{Privacy in traffic}
To reveal the necessity for privacy protection, we first present some important features of traffic data in intrusion detection scenario.

$\bullet$
\textbf{Static info} includes source/destination IP and port, network protocol, service, and specific flags.

$\bullet$
\textbf{Timestamp} represents the arrival time of a packet.
Aggregation of enough packets can carry information about traffic behavior pattern.
For example, TANTRA \cite{sharon2022tantra} utilizes DNN to learn benign traffic's timestamp differences and apply them on malicious data for adversarial attacks.

$\bullet$
\textbf{Packet size} represents the content length of a packet.
Traffic flows usually follow specific packet size distribution under certain network protocol.
For instance, \cite{nasr2021defeating} leverages this feature to generate blind adversarial perturbations to evade DNN-based traffic classifiers.

$\bullet$
\textbf{Flow throughput} represents the packet sending or receiving rate during a traffic flow.
Adversary can construct malicious traffic with moderate flow rate that mimics the target traffic tailored to the target model like DDoS attack in \cite{han2021evaluating}.

$\bullet$
\textbf{Others} include the combination of basic features like max/min packet size during a TCP flow between specific IPs.
These features are usually perturbed to help maximize the probability of a successful evasion.

\subsubsection{Why privacy protection}
In FL-NIDS, as introduced in Section~\ref{reconstrcution_attack}, the adversary can recover raw traffic data via reconstruction attacks.
Once he/she obtains enough samples with the features mentioned above, instead of targeting the data themselves like images, the adversary can make use of the data to perform the following actions:

$\bullet$
Learning benign traffic's behavior to 
generate malicious traffic with black-box attacks, e.g., training GAN with benign data and generate with random noises, or just learning the timestamp feature for more subtle attacks as TANTRA \cite{sharon2022tantra}.

$\bullet$
Directly perturbing the recovered malicious data with white-box attacks targeting specific model, e.g., PGD \cite{madry2017towards}, DeepFool \cite{moosavi2016deepfool}, and AutoPGD \cite{croce2020reliable}.

Note that we will conduct extensive experiments to validate the efficiency of such black-box and white-box attacks in Section~\ref{privacy_results}, as a straightforward demonstration for potential threats of reconstruction attacks even with existing defenses, and therefore the necessity for stronger privacy protection for FL-NIDS.
As a result, in this paper we are motivated to propose a new effective defense strategy FedDef to enhance the data security for FL in Section~\ref{method}.
\section{Threat Model and Privacy Metric}
\label{threat_model}

In this section, we present our adversary and make several assumptions, then we introduce two privacy metrics with an evaluation example to enhance our motivation.

\subsection{Adversary}
\label{server}

We consider an honest-but-curious server that follows the exact FL protocol and is allowed to observe the updates from different users.
The goal of the attacker is to reconstruct each user's private data and launch adversarial attack against NIDSs.
We make several assumptions of the adversary's power as follows:
(1) The adversary is aware of the model architecture and loss function.
(2) The adversary knows about the property of the training data, including max/min values, and feature types (i.e., discrete or continuous).
(3) The adversary is aware of each user's local training batch size.
In this way, the adversary can perform suitable attacks depending on different scenarios and thus evaluate defenses' lower bound.

\subsection{Privacy Score via Reconstruction Attack}
\label{attack_nids}

\subsubsection{Inversion Attack in NIDS scenario}
We leverage optimization-based reconstruction attack using $L_2$ and cosine distance metrics as introduced in Section~\ref{reconstrcution_attack}.
However, there exist two challenges when it comes to traffic features:

Firstly, traffic features are usually normalized during pre-processing procedure, therefore, the reconstructed data $x^*$ should satisfy $x^*\in[0,1]^{dim}$, where $dim$ is the dimension of data features.
Secondly, there are discrete and continuous traffic features in $x^*$, which require different techniques.

To address the first challenge, we can project the final optimized $x^*$ into legal feature space such that $x^*=clamp(x^*, min=0, max=1)$.
Note that we don't leverage variable change since it induces additional computation and achieves similar performance.
To address the second challenge, we first project the normalized $x^*$ into original feature space, and then we force the discrete features to be integer and normalize $x^*$ again for later privacy evaluation.

Note that when the batch size of the training data is more than $1$, inversion attack may not converge because the reconstructed data can have different permutations according to different initialization.
Therefore, we reconstruct a single sample at a time while keeping the rest the same.
We incorporate label $y^*$ in the optimization process and output the index of the maximum value of $y^*$ as the reconstructed label.

\subsubsection{Extraction attack in NIDS scenario}
In Section~\ref{reconstrcution_attack}, we have illustrated that extraction attack is effective for a single training sample whenever there is a layer with both weight and bias.
According to the neural network in PyTorch, we have the following theorem:
\begin{theorem}
\label{theorem_1}
If the first layer has a bias, then $x^T=(\frac{\partial \ell}{\partial W})^T(\frac{\partial \ell}{\partial b})^T(\frac{\partial \ell}{\partial b}(\frac{\partial \ell}{\partial b})^T)^{-1}$ when and only when batch size=1.
\end{theorem}
\begin{proof}[\textbf{Proof of Theorem~\ref{theorem_1}}]
When batch size is 1,
the output $y$ of a FC layer is computed as $y=xW^T+b$, where $x$ is the input with size $1 \times in\_feature$, $(W, b)$ are the weight and bias with size $out\_feature \times in\_feature$ and $1 \times out\_feature$, respectively, and we have $(\frac{\partial \ell}{\partial W})^T=\frac{\partial \ell}{\partial W^T}=x^T\frac{\partial \ell}{\partial y}=x^T\frac{\partial \ell}{\partial b}$, then we have:
\begin{equation}
\label{torch_extraction}
    x^T=(\frac{\partial \ell}{\partial W})^T(\frac{\partial \ell}{\partial b})^T(\frac{\partial \ell}{\partial b}(\frac{\partial \ell}{\partial b})^T)^{-1}
\end{equation}

When batch size is more than 1,
since $b$ is broadcast into $b^*=[b,...,b]$ with size $batch \times out\_feature$, we have $\frac{\partial \ell}{\partial b^*}(\frac{\partial \ell}{\partial b^*})^T=\begin{bmatrix}
 bb^T &...& bb^T\\
  ...& bb^T & ...\\
 bb^T & ... &bb^T
\end{bmatrix}=bb^T\begin{bmatrix}
 1 &...& 1\\
  ...& 1 & ...\\
 1 & ... &1
\end{bmatrix}$, where the second matrix ($batch \times batch$) is not invertible and Eq.~(\ref{torch_extraction}) fails.
\end{proof}

\begin{definition}[\textbf{Privacy Score}]
Based on the reconstructed data, we introduce privacy score metric where we directly add up the absolute distance for continuous features, and then we project the data into their original feature space and set the distance of discrete features to 1 if they don't match the original value.
The overall score computation is :
\begin{equation}
    score(x,x^*)=\frac{\sum_{i\in S_c} |x_i-x^*_i|+\sum_{i\in S_d} equal(X_i,X^*_i)}{|S_c|+|S_d|}
\end{equation}
where $S_c$ and $S_d$ denote the continuous and discrete features, $X$ and $X^*$ are projected features with original value ranges, $equal(x_1,x_2)=1$ if $x_1=x_2$, and it gets $0$ otherwise.
In this way, we can quantitatively evaluate the privacy leakage for each defense for FL-based NIDS.
\end{definition}

\subsection{Evasion Rate via Adversarial Attack}
\label{GAN_attack}

As long as the adversary obtains enough reconstructed data using either inversion attack or extraction attack, he/she can launch black-box attack by filtering out benign traffic to train GAN to generate adversarial examples from random noises.
In addition, the adversary can also perform white-box attack like PGD \cite{madry2017towards} that directly perturbs the reconstructed malicious traffic to reach better evasion rate than black-box methods.

\begin{definition}[\textbf{Evasion Rate}]
We derive evasion rate (ER) using the generated adversarial traffic to attack the trained DNN-based FL model and Kitsune.
For more straight evasion results, we present model accuracy $ACC_{DNN}$ for DNN model where $ER=1-ACC_{DNN}$, and Root Mean Squared Error $RMSE$ for Kitsune model, where $ER=1$ when $RMSE$ is lower than threshold and $ER=0$ otherwise. 
\end{definition}

\subsection{Evaluation Example}
\label{evaluation_example}

\subsubsection{Example of privacy score}
To better understand the privacy metrics, we first evaluate model utility and privacy score on several existing defenses (optimal parameters) against inversion attack and list some reconstructed normalized features in TABLE~\ref{reconstruction_example}, where $x$ and $x^*$ denote the original and reconstructed traffic feature, respectively.
Smaller score generally means more privacy leakage.
Model without defense is the most vulnerable with the lowest privacy score (6.6e-4) and the adversary can almost perfectly recover continuous and discrete traffic features, while the other four defenses can achieve higher score (3.3e-3 to 2.8e-1) with less data leakage.
However, deploying defenses can degrade the FL model utility, where Instahide performs the worst with 4.6\% accuracy loss in exchange for such privacy guarantee.
\begin{table}[htbp]
\caption{Reconstruction example on KDD99 dataset, \textbf{ACC} denotes the FL model accuracy, \textbf{PS} denotes privacy score.
We present normalized and original value for discrete features*.}
\label{reconstruction_example}
\resizebox{\linewidth}{!}{
\begin{tabular}{|c|c|c|c|c|c|c|c|}
\hline
\diagbox[]{Data}{Feature} & \textbf{ACC} & \textbf{PS} & Duration & Protocol* & Service* & Src\_bytes & Dst\_bytes \\ \hline
$x$ & \diagbox[]{}{} & \diagbox[]{}{} & 0.00 & 1.00/2 & 0.16/12 & 2.01e-4 & 0.00 \\ \hline
$x^*$ w/o. defense  & 0.996 & 6.6e-4 & 0.00 & 1.00/2 & 0.16/12 & 4.60e-4 & 0.00 \\ \hline
$x^*$ w/. Soteria \cite{sun2021soteria} & 0.994 & 3.3e-3 & 0.00 & 1.00/2 & 0.16/12 & 0.00 & 1.38e-3 \\ \hline
$x^*$ w/. GP \cite{zhu2019deep} & 0.989 & 1.5e-1 & 0.11 & 0.50/1 & 0.16/12 & 0.00 & 0.00 \\ \hline
$x^*$ w/. DP \cite{abadi2016deep} & 0.985 & 2.8e-1 & 0.10 & 1.00/2 & 0.00/1 & 0.00 & 0.00 \\ \hline
$x^*$ w/. Instahide \cite{huang2020instahide} & 0.950 & 2.0e-1 & 0.00 & 0.50/1 & 0.37/26 & 0.00 & 0.00 \\ \hline
\end{tabular}
}
\end{table}

\subsubsection{Example of evasion rate}
To further demonstrate the consequences of reconstruction attack, we leverage black-box attack as an example and evaluate evasion rate of adversarial examples (AEs) on Kitsune \cite{mirsky2018kitsune} and trained DNN-based FL-NIDS (refer to Section~\ref{privacy_results} for setup).
Fig. \ref{test_gan} illustrates the RMSE and accuracy change during the training process.
We can find that even DP with such strong privacy guarantee can sometimes induce successful evasion.
For example, RMSE for DP is under the threshold for KDD99 dataset against Kitsune and thus $ER=1$, while for the DNN model, the accuracy also drops to $0$ for Mirai dataset.
The other defenses can only perform worse, which poses great potential threats to users.
We also notice that AEs under DP perform relatively worse than other baselines, which corresponds to TABLE \ref{reconstruction_example} that DP achieves higher privacy score (2.8e-1) and induces less data privacy leakage.
\begin{figure}[!h]
\setlength{\abovecaptionskip}{0pt}
\setlength{\belowcaptionskip}{0pt}
\centering
\subfigure[KDD99+Kitsune]{
\includegraphics[width=0.43\linewidth]{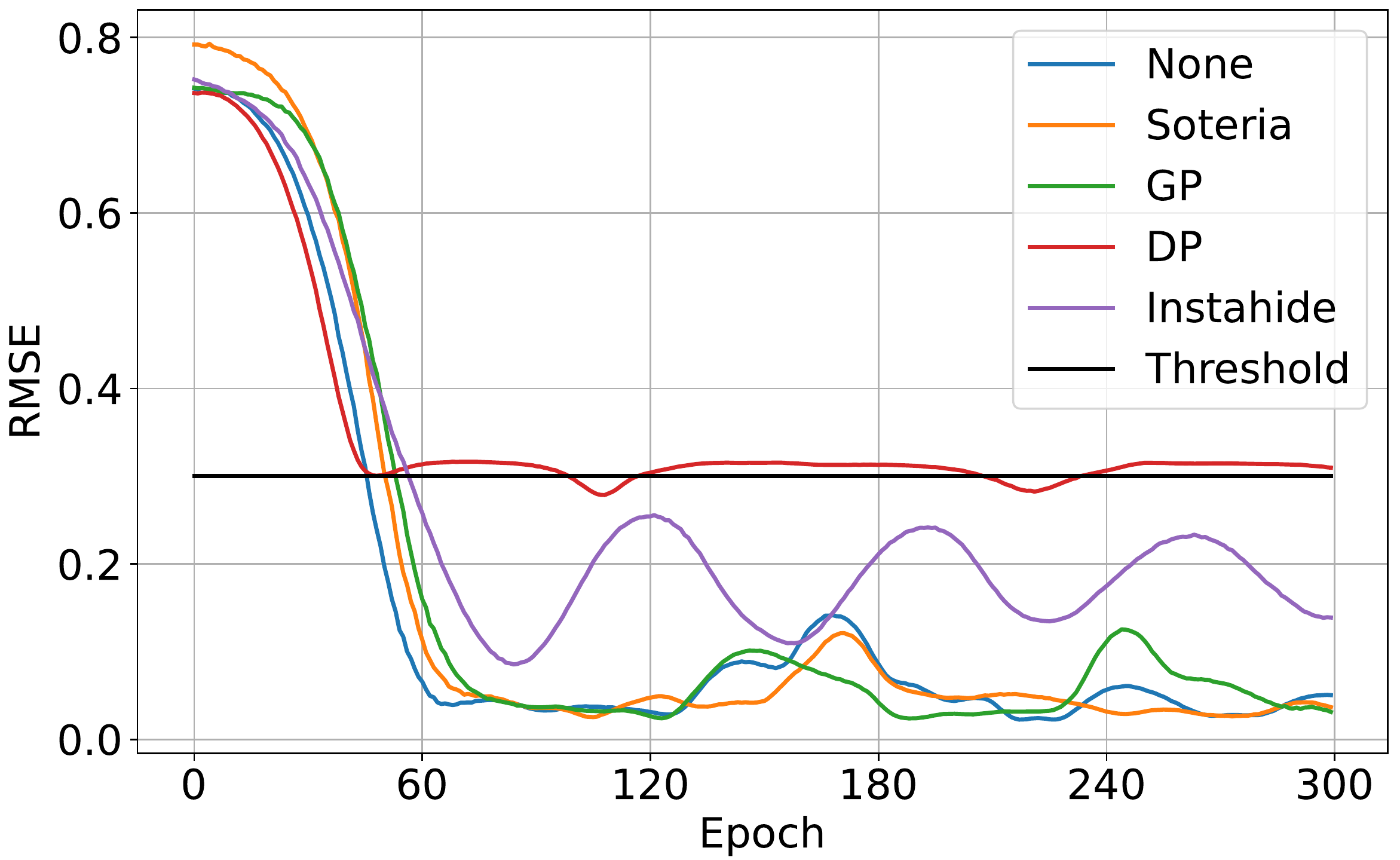}
}
\subfigure[Mirai+DNN]{
\includegraphics[width=0.43\linewidth]{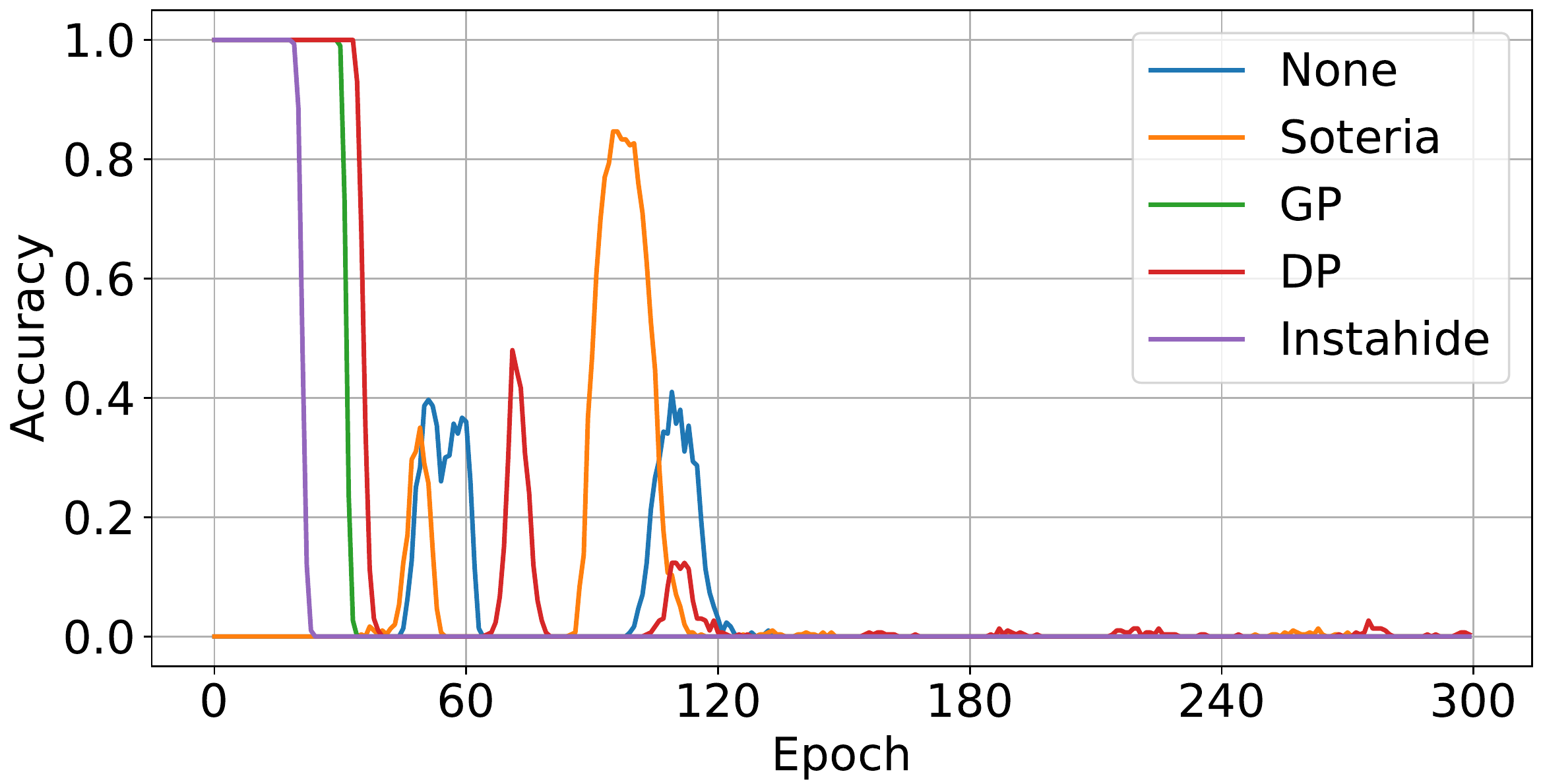}
}
\caption{Evasion rate on Kitsune and FL-NIDS models for KDD99 and Mirai dataset.}
\label{test_gan}
\end{figure}

In summary, our proposed metrics are practical to evaluate the robustness of FL-based NIDSs combined with different defenses, and the example results demonstrate the insufficiency of existing defenses.
Therefore, it is critical to design a new defense approach to build a more robust NIDS that preserves traffic data privacy without damaging model performance.

\section{FedDef: Optimization-based input perturbation defense}
\label{method}

In this section, we present our defense FedDef.
We introduce detailed optimization design and also implementation with a representative FL framework, i.e., FedAvg.

\subsection{Overview of FedDef}
\label{overview}

\begin{figure*}
\setlength{\abovecaptionskip}{0pt}
\setlength{\belowcaptionskip}{0pt}
\centering
\includegraphics[width=0.7\linewidth]{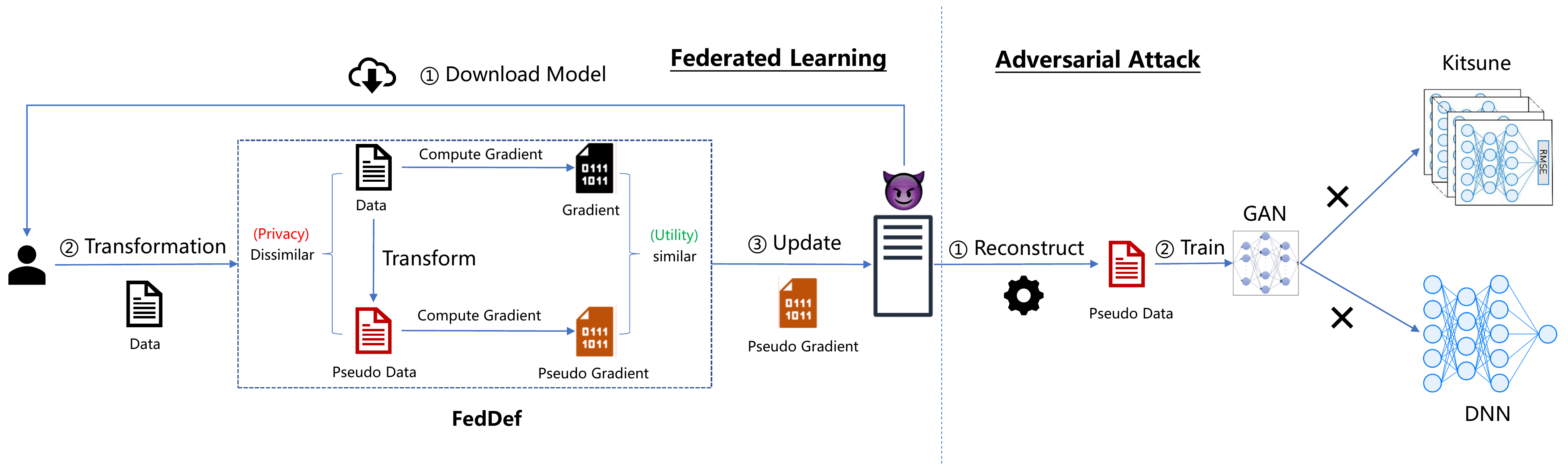}
\caption{Overview of FedDef.}\vspace{-0.5cm}
\label{FedDef}
\end{figure*}

As illustrated in Fig.~\ref{FedDef},
local users first download the latest global model from the server.
During local training, users first leverage FedDef to transform their own private data into pseudo data such that the data are dissimilar to preserve privacy and the corresponding gradients are similar to preserve FL model performance.
Then the users update the corresponding pseudo gradient instead of the original gradient.

The right side in Fig.~\ref{FedDef} illustrates the adversary manipulating the reconstructed data to attack the target models.
The adversary first leverages reconstruction attack to recover user's private training data and labels from the gradients, then black-box adversarial attack is available using GAN to evade the FL model or other NIDSs such as Kitsune.
While our defense ensures that the adversary can reconstruct only the pseudo data, which cannot help GAN model converge and therefore the adversarial attack fails.

\subsection{Optimization Design of FedDef}
\label{optimization design}

Privacy preserving and model performance are two key indicators for FL, users collaboratively train a global robust model to acquire high detection results on intrusion attacks while also protecting their private data from information leakage.
Based on these two aspects, we design a novel defense scheme for NIDS, i.e., FedDef, aiming to (1) preserve data privacy against reconstruction attack and the following adversarial attack, and (2) maintain FL model performance.

\subsubsection{Privacy preserving}
To address the privacy concern, we prevent privacy leakage from the source. The input training data should be dissimilar from the original data, and we have the following optimization:
\begin{equation}
\label{max_data}
    argmax_{x',y'} \ \ D( x', y'; x, y)
\end{equation}
\noindent where $x$ and $y$ are the original input data and labels, $x'$ and $y'$ are the pseudo pair for later pseudo gradient computation. $D(\cdot,\cdot)$ denotes the distance metric that evaluates how dissimilar $(x', y')$ is from $(x, y)$, where features and labels require different techniques.

\textbf{Data Transformation}.
According to the inversion attack and extraction attack introduced in Eq.~(\ref{L2_optimization}) and Eq.~(\ref{torch_extraction}), the reconstructed features $x^*$ depend on the absolute numerical values of $x'$, therefore, we optimize $x'$ to be as far from $x$ as possible by $L_2$ distance metrics as $argmax_{x'} \ \ \Vert x'-x\Vert_2$.

However, the reconstructed data $x^*$ is always projected into $[0,1]^{dim}$, where $dim$ is the feature dimension of $x$ and $x^*$.
Therefore, it makes no difference to the reconstructed data $x^*$ when $x^* \le 0$ or $x^* \geq 1$.
Due to the consideration above, we can constrain the upper bound of the $L_2$ distance and rewrite the optimization with respect to $x'$ as
\begin{equation}
    argmin_{x'} \ \ ReLU(\delta - \Vert x'-x\Vert_2)
\end{equation}
where $\delta$ is the upper bound of the distance between $x$ and $x'$, and $ReLU(x)=max(0, x)$, which ensures the distance is as large as possible while also guaranteeing the upper bound.

\textbf{Label Transformation}.
On the contrary, the reconstructed labels $y^*$ do not rely on the absolute numerical values of $y'$, instead the adversary takes the index $i$ of the maximum label probability $y^*_i$ as the final label, therefore increasing the overall distance between $y'$ and $y$ does not guarantee label privacy.
For example, let $i$ be the index of $y$, $y_i=1$ and $y_j=0,\forall j\in [1,n],j\ne i$, where $n$ is the total label classes.
By leveraging $L_2$ distance metric to optimize $y'$, we may obtain $y'$ where $y'_i=2$ and $y'_j=1,\forall j\in [1,n],j\ne i$.
In this way, though the overall distance between $y'$ and $y$ is quite large, the adversary may still acquire the correct label $i$ since $y'_i$ is the maximum one in $y'$ and the reconstructed label $y^*_i$ may also share the same property, therefore, label information is extracted by equation $i=argmax(y^*)$.

To address this problem, we choose to minimize $y'_i$ so that it stays the minimum one in $y'$ and the adversary will probably get any indexes other than the ground-truth label $i$.
Specifically, we first locate the label index $i$ in $y$, and we minimize $y'$ such that the minimum one $y'_{min}$ is as close to $y'_i$ as possible which can be formulated as $argmin_{y'} \ \ |y'_{min}-y'_i|$, where $|\cdot|$ denotes the absolute distance.
In this way, we optimize $y'_i$ to be the minimum one in $y'$ so that the adversary will most probably reconstruct any indexes except $i$. 

Finally, we formulate the optimization for privacy preserving as follows:
\begin{equation}
\label{privacy_optimization}
    argmin_{x',y'} \ \ ReLU(\delta - \Vert x'-x\Vert_2)+|y'_{min}-y'_i|
\end{equation}
where $i$ denotes the ground-truth label index such that $y_i=1$, and $y'_{min}$ is the minimum one in $y'$.

\subsubsection{Model performance}
To maintain FL model performance, we choose to add a constraint for the pseudo gradient $\nabla \theta(x',y')$ generated by the pseudo data and label $(x',y')$ to be as close to the original gradient $\nabla \theta(x,y)$ as possible.
In particular, we expect the $L_2$ distance between the pseudo gradient and the original gradient to be within $\epsilon$ boundary, which can be written as 
\begin{equation}
\label{model_constraint}
    s.t. \ \Vert \nabla \theta(x',y')-\nabla \theta(x,y)\Vert_2 \le \epsilon 
\end{equation}
However, such constraint is highly non-linear and thus we transform the constraint into an optimization problem using $ReLU$ function as 
\begin{equation}
\label{model_optimization}
    argmin_{x',y'} \ \ ReLU(\Vert \nabla \theta(x',y')-\nabla \theta(x,y)\Vert_2 -\epsilon )
\end{equation}
Finally, we can combine the two objective functions to optimize $(x',y')$ such that the pseudo data and the original data are dissimilar to preserve data privacy and the corresponding gradients are similar to maintain model performance.
In particular, we can derive the combined optimization as follows:
\begin{equation}
\label{combined_optimization}
    \begin{split}
    &argmin_{x',y'} \ \ L_{Pri}+\alpha L_{Per}\\
    &L_{Pri}=ReLU(\delta - \Vert x'-x\Vert_2)+|y'_{min}-y'_i|\\
    &L_{Per}=ReLU(\Vert \nabla \theta(x',y')-\nabla \theta(x,y)\Vert_2 -\epsilon)\\
    \end{split}
\end{equation}
where $L_{Pri}$ and $L_{Per}$ stand for optimization loss for privacy preserving and FL model performance, respectively.
$\alpha$ is a parameter that evaluates the tradeoff between privacy and model performance.
The overall algorithm is illustrated in \textbf{Algorithm~\ref{algorithm_feddef}}.
\begin{algorithm}
    \caption{Transformation of private data and label}
    \label{algorithm_feddef}
    \KwIn{original data pair $(x,y)$, global differentiable model parameters $\theta$, loss function $F$, defense epochs $def\_ep$, privacy tradeoff $\alpha$, learning rate $def\_lr$, constant $g\_value$, $\epsilon$, $\delta$}
    \KwOut{pseudo data pair $(x',y')$} 
    $x'\leftarrow U(0, 1), y'\leftarrow U(0, 1)$; \Comment{random initialization}\\
    $\nabla \theta(x,y) \leftarrow \frac{\partial F(\theta,x,y)}{\partial \theta}$; \Comment{original gradient}\\
    $gt \leftarrow argmax(y)$; \Comment{ground-truth label index of $y$}\\
    \For {$i\leftarrow1$ to $def\_ep$} 
    {$\nabla \theta(x',y') \leftarrow \frac{\partial F(\theta,x',y')}{\partial \theta}$;\Comment{pseudo gradient}\\
    \lIf{$max(|\nabla \theta(x',y')|) \le g\_value$}{\textbf{return} $x', y'$}\Comment{early stop}\\
    $loss \leftarrow ReLU(\Vert \nabla \theta(x',y')-\nabla \theta(x,y)\Vert_2 -\epsilon)$\Comment{$L_2$ distance}\\
    $y'_{min} \leftarrow min(y')$\;
    $loss \leftarrow \alpha*loss+ReLU(\delta - \Vert x'-x\Vert_2)+|y'_{min}-y'_{gt}|$\;
    $x' \leftarrow Adam(x', \frac{\partial loss}{\partial x'}, def\_lr)$;\Comment{Adam optimizer}\\
    $y' \leftarrow Adam(y', \frac{\partial loss}{\partial y'}, def\_lr)$\;
    }
    \textbf{return} $x', y'$
\end{algorithm}

Note that we leverage early stop to terminate the optimization as long as the pseudo gradient tends to vanish.
In particular, we first set a predetermined constant $g\_value$, then we check the maximum absolute value of all gradients of the FL model and stop the optimization whenever the maximum value is lower than $g\_value$ (line 6).
In this way, we prevent extreme case of vanishing gradient and find it effective against extraction attack, because the pseudo gradients are small enough that $\frac{\partial F}{\partial b}(\frac{\partial F}{\partial b})^T$ approaches $0$ due to calculation accuracy and thus matrix inversion in Eq.~(\ref{torch_extraction}) fails, and global model can still get updated in the meantime.

\subsection{Implementation With FedAvg}
\label{fl_problem}
We introduce how to combine FedDef with the classic FL framework, i.e., FedAvg~\cite{mcmahan2017communication}, to build a robust FL-based NIDS.
We first formulate our FL training objective in FedAvg:
\begin{equation}
\label{fl_optimization}
    argmin_\theta \ \ \{F(\theta ) = \sum_{k=1}^{N} p_k F_k(\theta ) \}
\end{equation}
\noindent where $N$ is the total number of local users, $p_k$ is the corresponding weight of the $k$-th user participating in the global model training, and $p_k \ge 0$, $ {\textstyle \sum_{k=1}^{N}}p_k =1$.
$\theta$ is the parameter of the global differentiable model, $F_k(\cdot)$ is the overall loss function for the $k$-th user (e.g., cross entropy loss).
$N$ local users collaboratively train a global NIDS model by solving optimization (\ref{fl_optimization}) iteratively.
Here we present one training round (e.g., $t$-th) with our defense as follows:

(1) The server has learnt a global model $\theta_t$ at $t$-th round and distributes the model to all the participants.

(2) Every local user (e.g., $k$-th) first initializes their model with $\theta_t$, i.e., $\theta_t^k=\theta_t$ and then performs $E\ge1$ local updates.
Specifically, during each update, user first leverages FedDef to transform their private data and label with the latest model $\theta_t^k$ to generate pseudo data pair $(x',y')$ for pseudo gradient computation.
Therefore, the local model can be updated for the $i$-th iteration as:
\begin{equation}
\label{feddef_transformation}
    \xi^{'k}_{t+i}=FedDef(\xi^{k}_{t+i})
\end{equation}
\begin{small}
\begin{equation}
\label{fl_update}
    \theta ^k_{t+i+1} \leftarrow \theta ^k_{t+i}-\eta_{t+i}\nabla F_k(\theta ^k_{t+i},\xi^{'k}_{t+i}) \ ,i=0,1,...,E-1
\end{equation}
\end{small}

\noindent where $\xi^{k}_{t+i}$ is a sample pair (i.e., $(x^{k}_{t+i},y^{k}_{t+i})$) uniformly chosen from the $k$-th user's private dataset. $FedDef$ is our defense algorithm that outputs the pseudo data pair $\xi^{'k}_{t+i}$, $\eta_{t+i}$ is the learning rate for the $i$-th iteration.
In this way, user $k$ performs Eq.~(\ref{feddef_transformation}) and Eq.~(\ref{fl_update}) for $E$ rounds and finally obtains the local trained model $\theta ^k_{t+E}$.

Note that at each training round, the pseudo data pair $\xi^{'k}_{t+i}$ can be different due to different input data pair $\xi^{k}_{t+i}$ and the newly updated FL model.

(3) After local training is complete, the central server determines a set $\mathcal{S}_t$, which contains a subset of $K$ indices randomly selected with
replacement according to the sampling probabilities $p1,...,p_N$ from the $N$ users.
Then the global model is aggregated by simply averaging as:
\begin{equation}
\label{fl_aggregation}
    \theta _{t+E} \leftarrow \frac{1}{K} \sum_{k \in \mathcal{S}_t} \theta ^k_{t+E}
\end{equation}
In this way, the global model $\theta_{t+E}$ can be updated with only partial users, which can mitigate serious ``straggler’s effect", where the server has to wait for the slowest (even offline) user's update and faces great time delay.

To evaluate the model performance of such FL-NIDS, we directly feed the testing dataset to the DNN model without $FedDef$ transformation to derive clean accuracy.
\section{Theoretical Analysis}
\label{theoretical_analysis}

In this section, we derive theoretical analysis of model convergence and privacy guarantee for FedDef.

\subsection{Convergence Analysis}
\label{convergence_analysis}

We provide theoretical analysis for FL model convergence using FedAvg combined with our defense.
Our analysis follows \cite{Li2020On} on the convergence of FedAvg on non-iid data with partial users' update.

We first make the five following assumptions same as \cite{Li2020On}:

\begin{assumption}
\label{ass1}
$F_1,...,F_N$ are all $L$-smooth, that is, for all $V$ and $W$, and any $k \in [1,N]$, $F_k(V) \le F_k(W)+(V-W)^T\nabla F_k(W)+\frac{L}{2} \Vert V-W \Vert_2^2$.
\end{assumption}

\begin{assumption}
\label{ass2}
$F_1,...,F_N$ are all $\mu$-strongly convex, that is, for all $V$ and $W$, and any $k \in [1,N]$, $F_k(V) \ge F_k(W)+(V-W)^T\nabla F_k(W)+\frac{\mu}{2} \Vert V-W \Vert_2^2$.
\end{assumption}

\begin{assumption}
\label{ass3}
Let $\xi_t^k$ be sampled from the $k$-th user’s local data uniformly at random.
The variance of stochastic gradients for each user is bounded: $\mathbb{E} \Vert \nabla F_k(\theta_t^k,\xi_t^k) - \nabla F_k(\theta_t^k) \Vert^2 \le \sigma_k^2$, for all $k=1,...,N$.
\end{assumption}

\begin{assumption}
\label{ass4}
The expected squared norm of stochastic gradients is uniformly bounded: $\mathbb{E} \Vert \nabla F_k(\theta_t^k,\xi_t^k)\Vert^2 \le G^2$, for all $k=1,...,N$ and $t=0,...,T-1$, $T$ is the total iterations of local users' updates.
\end{assumption}

\begin{assumption}
\label{ass5}
Assume $\mathcal{S}_t$ contains a subset of $K$ indices randomly selected with
replacement according to the sampling probabilities $p1,...,p_N$.
The aggregation step of FedAvg performs $\theta _{t} \leftarrow \frac{1}{K} \sum_{k \in \mathcal{S}_t} \theta ^k_{t}$.
\end{assumption}

\begin{theorem}
\label{theorem_2}
Let $F^*$ and $F^*_k$ be the minimum values of $F$ and $F_k$, respectively.
We use the term $\Gamma =F^*-\sum\nolimits_{k=1}^{N}p_kF^*_k$ to quantify the degree of non-iid.
Recall that $E$ denotes the local updates for each user and $T$ denotes the total iterations.
Let Assumptions \ref{ass1}-\ref{ass5} hold and $L$, $\mu$, $\sigma_k$, $G$ be defined therein.
Choose $\kappa=\frac{L}{\mu}$, $\gamma=max\{8\kappa,E \}$, and learning rate $\eta_t=\frac{2}{\mu(\gamma+t)}$.
Then FedAvg with our defense with partial users participation for non-iid data satisfies:
\begin{equation}
\label{convergence_fl}
    \mathbb{E}[F(\theta _T)] -F^* \le \frac{2\kappa}{\mu+T} (\frac{B+C}{\mu} +2L\Vert \theta _0-\theta^* \Vert^2)
\end{equation}
where
\begin{equation}
    \begin{split}
        B&=\sum_{k=1}^{N}p^2_k(\epsilon+\sigma)^2 _k +6L\Gamma+8(E-1)^2(\epsilon+G)^2\\
        C&=\frac{4}{K}E^2(\epsilon+G)^2 
    \end{split}
\end{equation}
\end{theorem}
\begin{proof}[\textbf{Proof of Theorem~\ref{theorem_2}}]
We replace Assumptions \ref{ass3}-\ref{ass4} with our pseudo gradients and then derive the Theorem \ref{theorem_2}.
Without the loss of generality, we consider the $k$-th user's pseudo gradients' property. 
We first present the following lemmas:

\begin{lemma}
\label{lemma1}
Let $\Vert \cdot \Vert_2$ be a sub-multiplicative norm, for all $A$ and $B$, we have $\Vert A\Vert_2 -\Vert B\Vert_2 \le \Vert A+B\Vert_2 \le \Vert A\Vert_2 +\Vert B\Vert_2$ and $\Vert AB\Vert_2 \le \Vert A\Vert_2 \Vert B\Vert_2$.
\end{lemma}
\begin{lemma}
\label{lemma2}
Let $\Vert \cdot \Vert_2$ be a sub-multiplicative norm, for all $A$ and $B$, we have $\mathbb{E}^2\Vert A\Vert_2 \Vert B\Vert_2 \le \mathbb{E}\Vert A\Vert_2^2 \mathbb{E}\Vert B\Vert_2^2$.
\end{lemma}

The proof of Lemma \ref{lemma1} and \ref{lemma2} can be naturally obtained with norm triangle inequality and Cauchy–Schwarz inequality.

\begin{lemma}
\label{lemma_3}
Let $\Vert \cdot \Vert_2$ be a sub-multiplicative norm, for all $A$ and $B$, we have $\mathbb{E} \Vert A+B\Vert_2^2 \le (\sqrt{\mathbb{E} \Vert A \Vert_2^2} +\sqrt{\mathbb{E} \Vert B \Vert_2^2} )^2$.
\end{lemma}
\begin{proof}[\textbf{Proof of Lemma~\ref{lemma_3}}]
\begin{equation}
\label{lemma_3_proof}
    \begin{split}
    \mathbb{E} \Vert A+B\Vert_2^2 &\le \mathbb{E} (\Vert A \Vert_2+\Vert B\Vert_2)^2\\
    &=\mathbb{E} \Vert A \Vert_2^2+\mathbb{E} \Vert B\Vert_2^2+2\mathbb{E}\Vert A\Vert_2\Vert B\Vert_2\\
    &\le \mathbb{E} \Vert A \Vert_2^2+\mathbb{E} \Vert B\Vert_2^2+2\sqrt{\mathbb{E}\Vert A\Vert_2^2 \mathbb{E}\Vert B\Vert_2^2} \\
    &=(\sqrt{\mathbb{E} \Vert A \Vert_2^2} +\sqrt{\mathbb{E} \Vert B \Vert_2^2} )^2
    \end{split}
\end{equation}
\end{proof}

\textbf{Replace Assumption \ref{ass3}}.
According to our optimization (\ref{combined_optimization}), the pseudo gradients' $L_2$ distance from the original are constrained within $\epsilon$.
Therefore, the following equation holds:
\begin{equation}
\label{constrain_expected_gradient}
    \mathbb{E} \Vert \nabla F_k(\theta_t^k,\xi_t^{'k}) - \nabla F_k(\theta_t^k,\xi_t^{k}) \Vert_2^2 \le \epsilon^2
\end{equation}
The variance of the pseudo stochastic gradients for each user is bounded with Lemma \ref{lemma_3}:
\begin{small}
\begin{equation}
\label{assumption_3}
    \begin{split}
    &\mathbb{E} \Vert \nabla F_k(\theta_t^k,\xi_t^{'k}) - \nabla F_k(\theta_t^k) \Vert^2\\ 
    =&\mathbb{E} \Vert \nabla F_k(\theta_t^k,\xi_t^{'k}) - \nabla F_k(\theta_t^k,\xi_t^{k}) +\nabla F_k(\theta_t^k,\xi_t^{k}) - \nabla F_k(\theta_t^k) \Vert_2^2\\
    \le& (\sqrt{\mathbb{E} \Vert \nabla F_k(\theta_t^k,\xi_t^{'k}) - \nabla F_k(\theta_t^k,\xi_t^{k}) \Vert_2^2} \\ &+\sqrt{\mathbb{E}\Vert \nabla F_k(\theta_t^k,\xi_t^{k}) - \nabla F_k(\theta_t^k) \Vert_2^2})^2\\
    \le& (\epsilon+\sigma_k)^2
    \end{split}
\end{equation}
\end{small}

\textbf{Replace Assumption \ref{ass4}}.
We leverage similar method to replace Assumption \ref{ass4} with our pseudo gradient, the expected squared norm of pseudo stochastic gradients is uniformly bounded: 
\begin{small}
\begin{equation}
\label{assumption_4}
    \begin{split}
    &\mathbb{E} \Vert \nabla F_k(\theta_t^k,\xi_t^{'k}) \Vert^2\\
    =&\mathbb{E} \Vert \nabla F_k(\theta_t^k,\xi_t^{'k})-\nabla F_k(\theta_t^k,\xi_t^{k})+\nabla F_k(\theta_t^k,\xi_t^{k}) \Vert^2_2\\
    \le& (\sqrt{\mathbb{E} \Vert \nabla F_k(\theta_t^k,\xi_t^{'k})-\nabla F_k(\theta_t^k,\xi_t^{k})\Vert_2^2}+\sqrt{\mathbb{E} \Vert \nabla F_k(\theta_t^k,\xi_t^{k}) \Vert_2^2})^2\\
    \le& (\epsilon+G)^2
    \end{split}
\end{equation}
\end{small}

In this way, we can replace the bound of the original gradients in Assumptions \ref{ass3}-\ref{ass4} with our pseudo gradients.
Therefore, Assumptions \ref{ass1}-\ref{ass5} hold with our defense.
By applying our new bounds in \cite{Li2020On}, we can derive Theorem 2
aforementioned.

\end{proof}

\subsection{Privacy Analysis}
\label{privacy_analysis}

We provide theoretical analysis for data privacy with our defense.
Specifically, we derive a possible lower bound for the deviation of pseudo data $x'$ from the original $x$ using the derivatives of the first layer.
We first make the following assumption:

\begin{assumption}
Assume that there exist both weight $W$ and bias $b$ in the first layer of the model, and that the norm of the gradient with respect to the bias is bounded for some $M>0$ : $\Vert \frac{\partial F}{\partial b} \Vert_2 \le M$.
\end{assumption}

\begin{theorem}
\label{theorem_3}
Let Assumption 6 hold, $(x,y)$ be the original data pair after normalization, $(x',y')$ be the transformed data pair with our defense, $\nabla b =\frac{\partial F(\theta,x,y)}{\partial b}$, $\nabla W =\frac{\partial F(\theta,x,y)}{\partial W}$,$\nabla b' =\frac{\partial F(\theta,x',y')}{\partial b}$, $\nabla W' =\frac{\partial F(\theta,x',y')}{\partial W}$ and we have the following theorem to guarantee data privacy.
\begin{equation}
\label{privacy_fl}
    \Vert x'-x\Vert_2 \ge \frac{2(\Vert \nabla W'-\nabla W \Vert_2-\Vert \nabla b'-\nabla b \Vert_2)}{2M+\Vert \nabla b'-\nabla b \Vert_2} 
\end{equation}
\end{theorem}

\begin{proof}[\textbf{Proof of Theorem~\ref{theorem_3}}]
We can derive that $\nabla W^T =x^T \nabla b$ and $\nabla W'^T=x'^T \nabla b'$, therefore, we have the following transformation:
\begin{small}
\begin{equation}
    \begin{split}
    2(\nabla W-\nabla W')^T=(x-x')^T(\nabla b+\nabla b')+(x+x')^T(\nabla b-\nabla b')
    \end{split}
\end{equation}
\end{small}

With Lemma 1, we can derive the inequality on the one hand:
\begin{equation}
    \begin{split}
    \Vert (x-x')^T(\nabla b+\nabla b')\Vert_2 \le& \Vert x-x'\Vert_2 \Vert \nabla b+\nabla b' \Vert_2\\
    \le& \Vert x-x'\Vert_2 (\Vert \nabla b \Vert_2+\Vert \nabla b' \Vert_2)\\
    \le& 2M\Vert x-x'\Vert_2
    \end{split}
\end{equation}
Note that $\Vert x\Vert_2 \le 1$ since $x$ is normalized.
Therefore, on the other hand, we also have:

\begin{equation}
    \begin{split}
    &\Vert (x-x')^T(\nabla b+\nabla b')\Vert_2\\
    =&\Vert 2(\nabla W-\nabla W')^T -(x+x')^T(\nabla b-\nabla b')^T \Vert_2 \\
    \ge &2\Vert(\nabla W-\nabla W')\Vert_2 -\Vert (x+x')^T(\nabla b-\nabla b')^T \Vert_2\\
    \ge &2\Vert(\nabla W-\nabla W')\Vert_2 -\Vert (x+x')\Vert_2 \Vert(\nabla b-\nabla b')\Vert_2\\
    =&2\Vert(\nabla W-\nabla W')\Vert_2 -\Vert (x'-x+2x)\Vert_2 \Vert(\nabla b-\nabla b')\Vert_2\\
    \ge &2\Vert(\nabla W-\nabla W')\Vert_2 -(\Vert x'-x\Vert_2+2\Vert x\Vert_2) \Vert(\nabla b-\nabla b')\Vert_2\\
    \ge &2\Vert(\nabla W-\nabla W')\Vert_2 -(\Vert x'-x\Vert_2+2) \Vert(\nabla b-\nabla b')\Vert_2
    \end{split}
\end{equation}
Therefore, by combining the above inequalities, we have :
\begin{small}
\begin{equation}
\label{final_inequality}
    2M\Vert x'-x\Vert_2 \ge 2\Vert(\nabla W-\nabla W')\Vert_2 -(\Vert x'-x\Vert_2+2) \Vert(\nabla b-\nabla b')\Vert_2
\end{equation}
\end{small}
By extracting $\Vert x'-x\Vert_2$, we have Theorem 3.
\end{proof}

\section{Evaluation}
\label{evaluation}

In this section, we conduct several experiments to evaluate our defense from two perspectives:
(1) Model utility under iid and non-iid data distribution.
(2) Privacy protection using the two proposed metrics, i.e., privacy score (plus label reconstruction accuracy) and evasion rate.
We also conduct ablation study to derive optimal parameters combination and test computation overhead.

\subsection{Experimental Setup}
\label{experiment_setup}

We conduct our experiments using PyTorch with 8-core, 64GB CPU and one Nvidia-P100 (16GB) GPU.
The basic experiments setting is as follows:

\subsubsection{Datasets}
We choose four network intrusion attack datasets to evaluate our defense, including KDDCUP'99 \cite{tavallaee2009detailed}, Mirai botnet \cite{mirsky2018kitsune}, CIC-IDS2017 \cite{sharafaldin2018toward}, and UNSW-NB15 \cite{moustafa2015unsw}.
Mirai botnet contains only one type of malicious traffic named Botnet Malware and the other three datasets contain multiple attack types.
To further balance the datasets, we select only DDoS attack traffic from CIC-IDS2017 so that we evaluate on two 2-class dataset (i.e., Mirai Botnet and CIC-IDS2017) and two multi-class datasets (i.e., KDDCUP'99 and UNSW-NB15).
More details can be found in TABLE~\ref{datasets_table}.
\begin{table}[htbp]
\caption{Datasets used in our work}
\label{datasets_table}
\resizebox{\linewidth}{!}{
\begin{tabular}{|c|c|c|c|c|}
\hline
\diagbox[]{Dataset}{Property}     & Feature Dimension & Attack Type & Training Samples & Testing Samples \\ \hline
KDD99       & 41                & 23          & 345815           & 148206          \\ \hline
Mirai       & 100               & 2           & 210000           & 90000           \\ \hline
CIC-IDS2017 & 76                & 2           & 157998           & 67713           \\ \hline
UNSW-NB15   & 42                & 10          & 180372           & 77301           \\ \hline
\end{tabular}
}
\end{table}

\subsubsection{FL model}
We derive our global NIDS model under FedAvg framework with 10 local users collaboratively training a 3-layer fully connected DNN with layer sizes [$dim$, $2*dim$, $3*dim$, $n$], where $dim$ denotes the feature size and $n$ denotes attack types.
The first two layers both consist of a linear network with ReLU activation function, while the third layer also has a linear network with softmax that projects the output vector into label class probability.

\subsubsection{Baselines and parameters configuration}
We choose five baselines mentioned in Section~\ref{defenses} to compare with our work, including no defense, Soteria, gradient pruning (GP), differential privacy (DP), and Instahide.

For our defense, we set the default parameter as $\epsilon=0$, $\delta=1$,  $g\_value=1e-15$, learning rate $def\_lr=2e-1$, and steps $def\_ep=40$ unless otherwise specified, we also set $\alpha=0.25, 0.5, 1$ accordingly to observe its impact.

For other defenses, we optimize the parameters to achieve the best tradeoff between performance and privacy.
Specifically, we set the gradient compression rate to 99\% for GP, and we leverage Laplace noise and set the mean and variance of the noise distribution as $0$ and $0.1$ for DP, while following the default setting for Soteria and Instahide.

\begin{table*}[htbp]
\caption{Parameter setting for model performance evaluation presented as learning rate/decay.}
\label{parameter_setting_model}
\resizebox{\linewidth}{!}{
\begin{tabular}{|c|c|ccc|c|c|c|c|}
\hline
\multirow{2}{*}{\diagbox[]{Dataset}{Defense}} & \multirow{2}{*}{No Defense} & \multicolumn{3}{c|}{Our Defense}& \multirow{2}{*}{Soteria} & \multirow{2}{*}{GP} & \multirow{2}{*}{DP} & \multirow{2}{*}{Instahide} \\ \cline{3-5}
& & \multicolumn{1}{c|}{$\alpha=1$} & \multicolumn{1}{c|}{$\alpha=0.5$} & $\alpha=0.25$ & & & & \\ \hline 
KDD99 & 1e-2/0.9 & \multicolumn{1}{c|}{1.5e-2/0.9} & \multicolumn{1}{c|}{1.5e-2/0.9} & 1.5e-2/0.9 & 1e-2/0.9 & 3e-2/0.9 & 3e-2/0.9 & 3e-2/0.9 \\ 
Mirai & 3e-3/0.9 & \multicolumn{1}{c|}{3e-3/0.9}  & \multicolumn{1}{c|}{3e-3/0.9} & 3e-3/0.9 & 3e-3/0.9 & 3e-3/0.9 & 3e-2/0.9 & 3e-3/0.9 \\ 
CIC-IDS2017 & 1e-3/0.8 & \multicolumn{1}{c|}{1.5e-2/0.8} & \multicolumn{1}{c|}{1.5e-2/0.8} & 1.5e-2/0.8 & 1e-3/0.8 & 3e-3/0.9 & 3e-2/0.9 & 3e-3/0.9 \\ 
UNSW & 3e-3/0.8 & \multicolumn{1}{c|}{1.5e-2/0.9} & \multicolumn{1}{c|}{1.5e-2/0.9} & 1.5e-2/0.9 & 2e-3/0.9 & 3e-2/0.9  & 3e-2/0.9 & 3e-2/0.8 \\ \hline
\end{tabular}
}
\end{table*}

\begin{table*}[htbp]
\caption{Accuracy comparison on four datasets, higher is better.}
\label{fl_acc}
\resizebox{\linewidth}{!}
{\begin{tabular}{|c|c|ccc|c|c|c|c|}
\hline
\multirow{2}{*}{\diagbox[]{Dataset}{Defense}} & \multirow{2}{*}{No Defense} & \multicolumn{3}{c|}{Our Defense}& \multirow{2}{*}{Soteria} & \multirow{2}{*}{GP} & \multirow{2}{*}{DP} & \multirow{2}{*}{Instahide} \\ \cline{3-5}
& & \multicolumn{1}{c|}{$\alpha=1$} & \multicolumn{1}{c|}{$\alpha=0.5$} & $\alpha=0.25$ & & & & \\ \hline 
KDD99 & \textcolor{green}{0.996} & \multicolumn{1}{c|}{0.987$\pm$0.001} & \multicolumn{1}{c|}{0.986$\pm$0.001} & 0.985$\pm$0.002 & 0.994 & 0.989 & 0.985 & \textcolor{red}{0.950} \\ 
Mirai & \textcolor{green}{0.923} & \multicolumn{1}{c|}{0.922$\pm$0.000}  & \multicolumn{1}{c|}{0.922$\pm$0.000} & 0.921$\pm$0.000 & 0.923 & 0.922 & 0.922 & \textcolor{red}{0.920} \\ 
CIC-IDS2017 & \textcolor{green}{0.982} & \multicolumn{1}{c|}{0.971$\pm$0.002} & \multicolumn{1}{c|}{0.970$\pm$0.004} & 0.968$\pm$0.005 & 0.978 & 0.970 & 0.971 & \textcolor{red}{0.964} \\ 
UNSW-NB15 & \textcolor{green}{0.746} & \multicolumn{1}{c|}{0.720$\pm$0.008} & \multicolumn{1}{c|}{0.710$\pm$0.009} & 0.706$\pm$0.010 & 0.743 & 0.710  & 0.705 & \textcolor{red}{0.689} \\ \hline
\end{tabular}}
\end{table*}

\subsubsection{Evaluation metrics}
We leverage FL model accuracy to evaluate defense's utility.
For privacy guarantee, we derive privacy score (plus label reconstruction accuracy) and evasion rate as introduced in Section~\ref{threat_model}.

\subsection{Model Performance Results}
\label{convergence_results}

\textbf{Setup 1}.
We evaluate FL model accuracy as introduced in Section~\ref{experiment_setup} with Adam optimizer with learning decay.
We first divide the original datasets to derive training dataset and testing dataset by $7:3$, and each user can get certain amount of the training dataset under iid or non-iid setting.

We design a new non-iid data generation algorithm.
Specifically, each user can randomly get the same percent of the benign traffic, and then we randomly select $p>0$ attack types and distribute the same percent ($\frac{1}{10}$ in our case) of that specific malicious traffic.
To balance the result, we consider non-iid data distribution for multi-class KDD99 dataset, and iid data distribution for the other three datasets, where each user gets the same percent of the whole training dataset.
After dataset distribution, local users normalize their own datasets with maximum/minimum values.

We set the overall training rounds $T=300$, and each user updates the global model only once $local\_ep=1$ with batch size $local\_bs=1000$.
We follow the same parameter setting introduced in Section~\ref{experiment_setup} and we optimize the local training learning rate for different defenses and datasets.
Detailed parameters are illustrated in TABLE~\ref{parameter_setting_model}, where no defense is around 1e-2, ours is around 1.5e-2, and strong defense like DP requires higher learning rate as 3e-2, otherwise such great noises to the gradients can hardly help optimize the model.
Note that we also leverage learning rate decay to accelerate the model convergence.
Specifically, the learning rate $lr$ is updated every 20 rounds as $lr=lr \times decay$, therefore, $lr$ is getting smaller as the model is approaching convergence to avoid performance oscillation.

\textbf{Accuracy Analysis}.
We run the training for five times and report the best results for baselines and the average accuracy for our defense in TABLE \ref{fl_acc}.
It can be found that our defense with $\alpha=1$ achieves similar performance from the baseline with at most 2.6\% accuracy loss (from UNSW dataset) and the least deviation under iid or non-iid setting.
As $\alpha$ tends to get smaller, Eq.~(\ref{combined_optimization}) tends to optimize pseudo data's privacy more than pseudo gradients' utility, which is why FedDef with $\alpha=0.25$ achieves lower accuracy with greater deviation (0.01 at most).
Meanwhile, we also compare our defense with other baselines and find that Soteria achieves the highest accuracy with only 1\% loss at most while Instahide generally performs the worst.
We suspect that the reason for Instahide's poor performance is that it randomly flips certain signs of the mixed data to provide additional security protection in exchange of model performance.

In summary, the FL model can still converge and reach high accuracy combined with our defense, thanks to our constraints on gradient deviation.

\begin{figure*}
\setlength{\abovecaptionskip}{0pt}
\setlength{\belowcaptionskip}{0pt}
\centering
\subfigure{
\begin{minipage}{1\linewidth}
\includegraphics[width=1\linewidth]{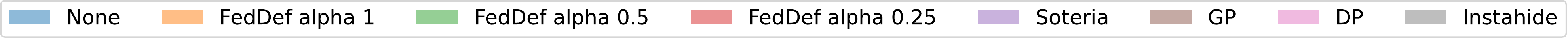}
\end{minipage}}
\setcounter{subfigure}{0}
\subfigure[KDD99]{
\begin{minipage}{0.45\linewidth}
\includegraphics[width=1\linewidth]{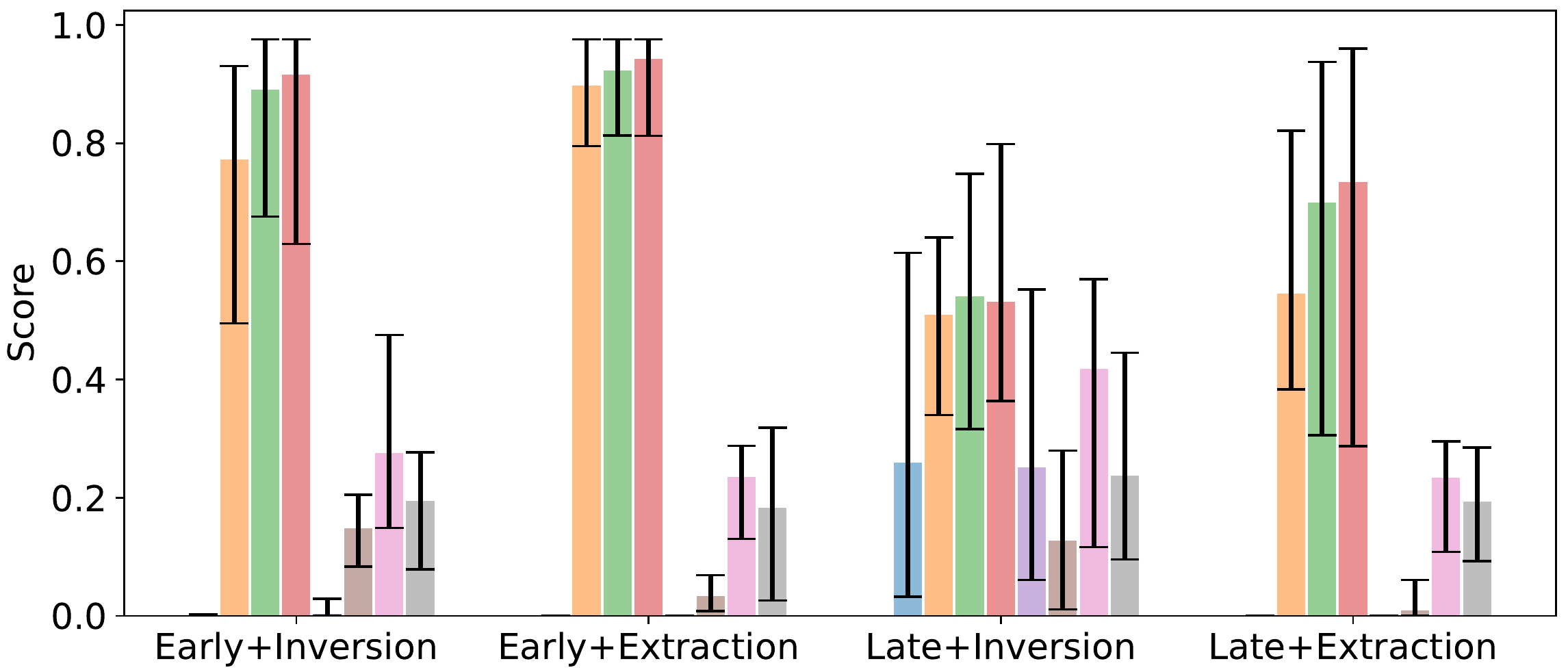}
\end{minipage}}
\subfigure[Mirai]{
\begin{minipage}{0.45\linewidth}
\includegraphics[width=1\linewidth]{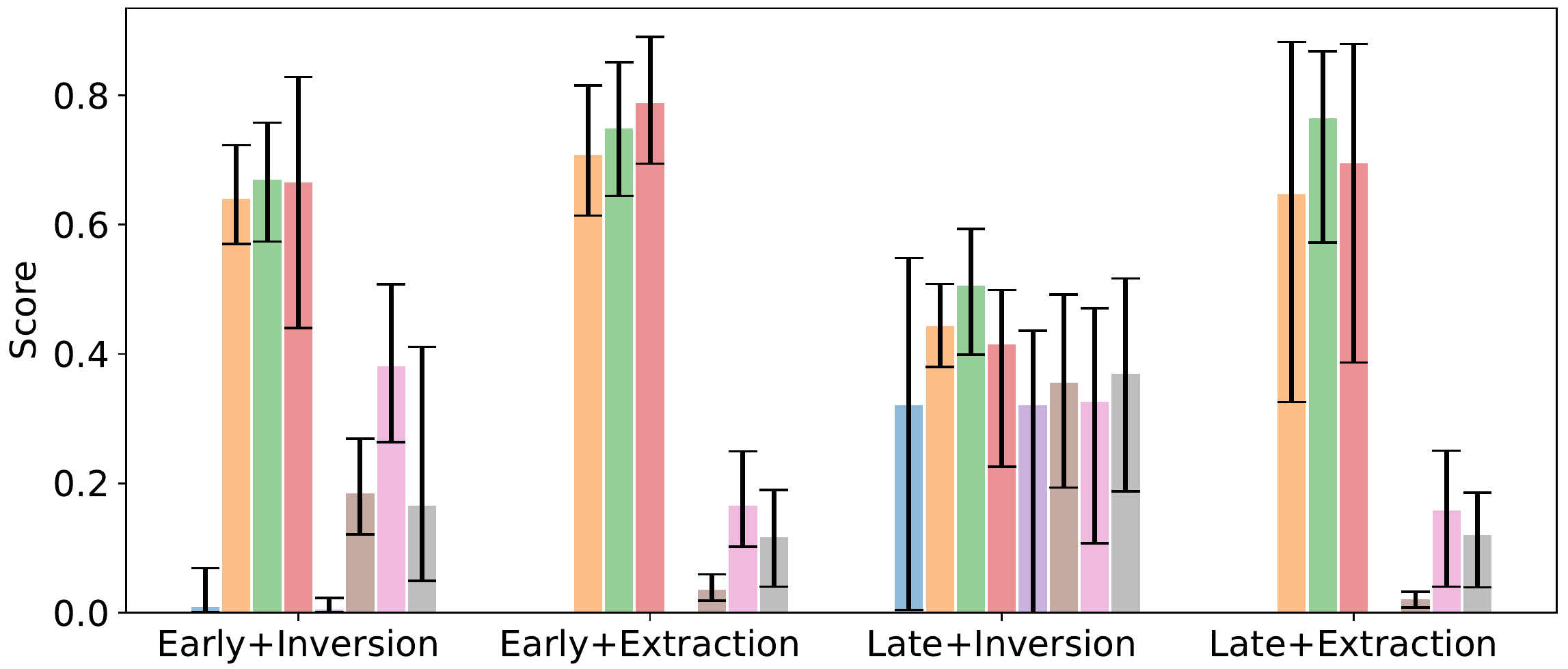}
\end{minipage}}
\subfigure[CIC-IDS2017]{
\begin{minipage}{0.45\linewidth}
\includegraphics[width=1\linewidth]{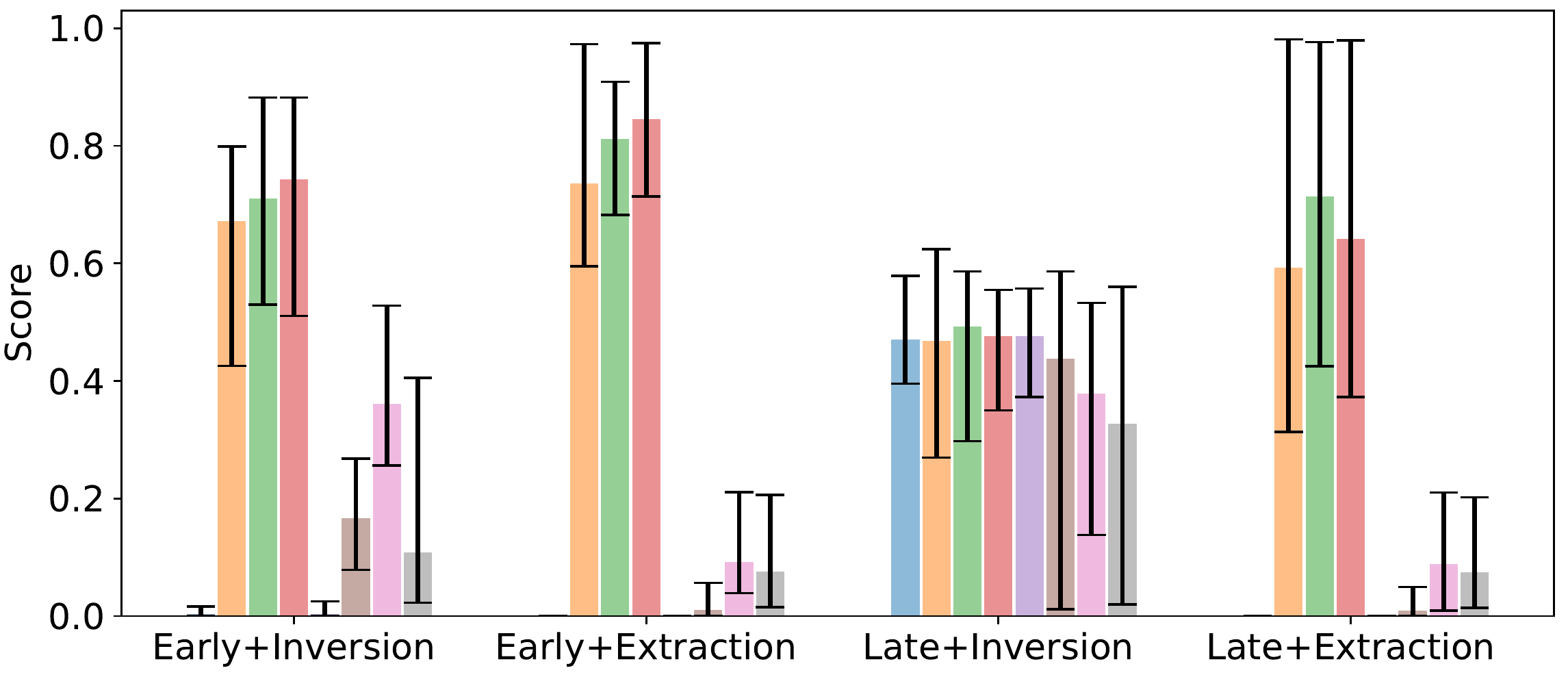}
\end{minipage}}
\subfigure[UNSW-NB15]{
\begin{minipage}{0.45\linewidth}
\includegraphics[width=1\linewidth]{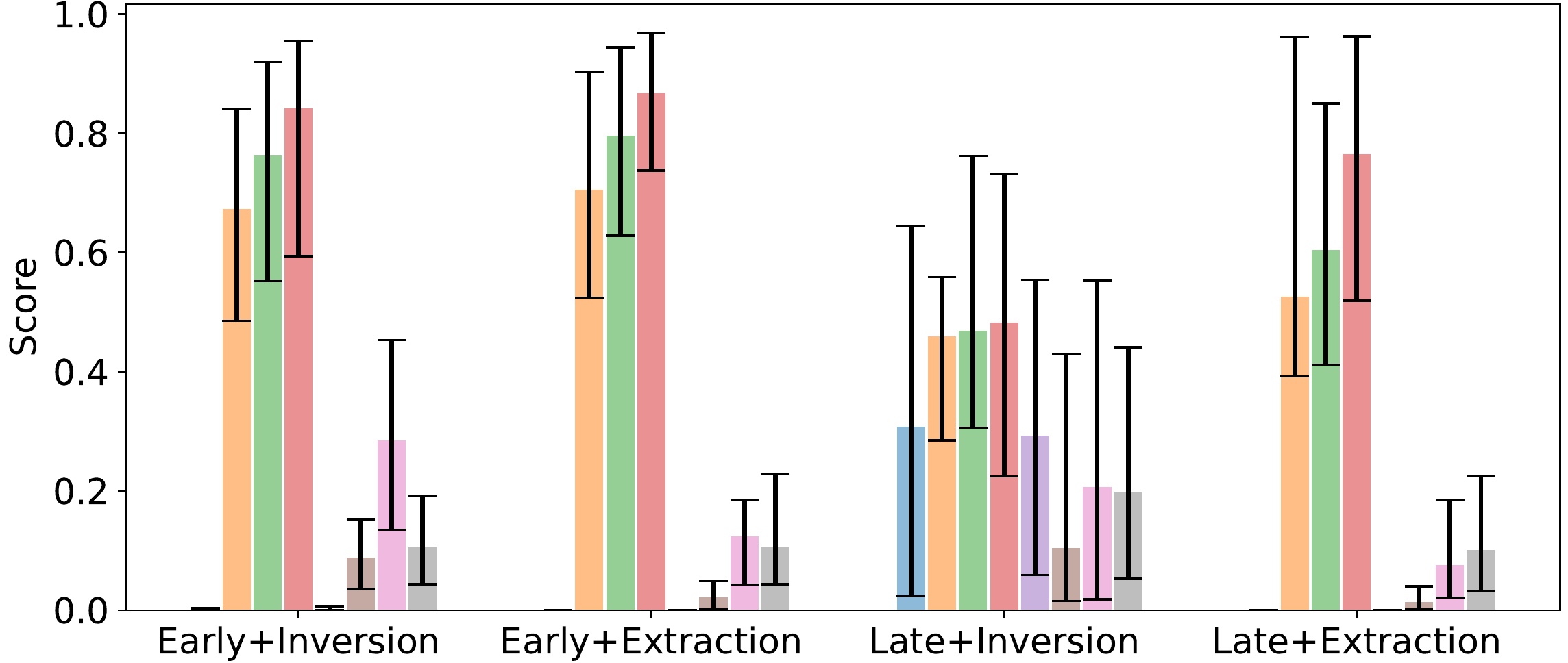}
\end{minipage}}
\caption{Reconstruction privacy score comparison on KDD99 and Mirai datasets, higher is better.}
\label{reconstruction_score_comparison}
\end{figure*}

\subsection{Model Privacy Results}
\label{privacy_results}

We evaluate model privacy with privacy score and evasion rate for different local training batches and training stages.

\subsubsection{Single sample reconstruction}
\label{aaaa}
We first consider local users updating model only once with only one training sample.
In this scenario, the adversary can launch both inversion attack and extraction attack.

\textbf{Setup 2}.
We evaluate privacy score with inversion attack using appropriate distance metric to optimize the reconstruction, i.e., we use cosine distance for KDD99 and CIC-IDS2017 datasets during late training stage and $L_2$ for the rest of the cases.
Note that extraction attack may fail due to calculation accuracy when gradients are small, and we will fall back to inversion attack when Eq.~(\ref{torch_extraction}) fails.
As introduced in Section~\ref{reconstrcution_attack}, we can only derive data from extraction attack, therefore, we leverage optimization-based inversion attack to acquire labels.
We follow the same parameter setting as in Section~\ref{experiment_setup}, and we set the overall iteration $T=100$, local batch size $local\_bs=1$, local step $local\_ep=1$ for both early and late stages.
Note that we use the randomly initialized model for early stage and the respective trained model in Section~\ref{convergence_results} for late training stage, we do not actually update the model to constantly evaluate privacy score.

\begin{table}[htbp]
\caption{Label reconstruction accuracy comparison on four datasets, and lower is better.}
\label{reconstruction_label_acc}
\resizebox{\linewidth}{!}
{\begin{tabular}{|c|c|c|ccc|c|c|c|c|}
\hline
\multirow{2}{*}{Data}        & \multirow{2}{*}{\diagbox[]{Stage}{Defense}} & \multirow{2}{*}{No Defense} & \multicolumn{3}{c|}{Our}                                       & \multirow{2}{*}{Soteria} & \multirow{2}{*}{GP} & \multirow{2}{*}{DP} & \multirow{2}{*}{Instahide} \\ \cline{4-6}
                             &                        &                             & \multicolumn{1}{c|}{$\alpha=1$} & \multicolumn{1}{c|}{$\alpha=0.5$} & $\alpha=0.25$ &                          &                     &                     &                            \\ \hline
\multirow{2}{*}{KDD99}       & Early                  & \textcolor{red}{1.00}                        & \multicolumn{1}{c|}{0.01}  & \multicolumn{1}{c|}{\textcolor{green}{0.00}} & 0.00  & 1.00                     & 1.00                & 0.99                & 0.37                       \\ 
                             & Late                   & 0.59                        & \multicolumn{1}{c|}{0.17}  & \multicolumn{1}{c|}{0.17} & \textcolor{green}{0.16}  & 0.75                     & \textcolor{red}{0.99}                & 0.29                & 0.22                       \\ \hline
\multirow{2}{*}{Mirai}       & Early                  & \textcolor{red}{1.00}                        & \multicolumn{1}{c|}{0.97}  & \multicolumn{1}{c|}{0.86} & 0.81  & 1.00                     & 1.00                & 1.00                & \textcolor{green}{0.52}                       \\  
                             & Late                   & 0.53                        & \multicolumn{1}{c|}{0.50}  & \multicolumn{1}{c|}{0.49} & \textcolor{green}{0.45}  & 0.58                     & \textcolor{red}{0.80}                & 0.68                & 0.48                       \\ \hline
\multirow{2}{*}{CIC-IDS2017} & Early                  & \textcolor{red}{1.00}                        & \multicolumn{1}{c|}{0.92}  & \multicolumn{1}{c|}{0.91} & 0.74  & 1.00                     & 1.00                & 1.00                & \textcolor{green}{0.53}                       \\ 
                             & Late                   & 0.52                        & \multicolumn{1}{c|}{0.48}  & \multicolumn{1}{c|}{0.45} & \textcolor{green}{0.44}  & \textcolor{red}{0.55}                     & 0.45                & 0.53                & 0.47                       \\ \hline
\multirow{2}{*}{UNSW-NB15}   & Early                  & \textcolor{red}{1.00}                        & \multicolumn{1}{c|}{0.06}  & \multicolumn{1}{c|}{0.01} & \textcolor{green}{0.00}  & 1.00                     & 1.00                & 1.00                & 0.26                       \\ 
                             & Late                   & 0.71                        & \multicolumn{1}{c|}{0.18}  & \multicolumn{1}{c|}{\textcolor{green}{0.06}} & 0.06  & 0.81                     & \textcolor{red}{0.98}                & 0.77                & 0.22                       \\ \hline
\end{tabular}}
\end{table}

\textbf{Privacy Score Analysis}.
The full results for average privacy score and reconstructed label accuracy over the $T$ samples are in Fig.~\ref{reconstruction_score_comparison} and TABLE~\ref{reconstruction_label_acc}.
During early training stage, the reconstruction attack proves excellent performance, where the privacy score can reach almost 0 without any defense and labels are also accurately extracted (ACC=1) among four datasets.
The baseline defenses provide some but limited protection, especially against more accurate extraction attack.

\begin{figure*}[htpb]
\setlength{\abovecaptionskip}{0pt}
\setlength{\belowcaptionskip}{0pt}
\centering
\subfigure{
\begin{minipage}{1\linewidth}
\includegraphics[width=1\linewidth]{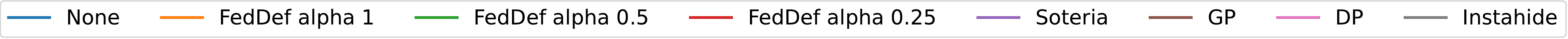}
\end{minipage}}
\setcounter{subfigure}{0}
\subfigure{
\begin{minipage}{0.23\linewidth}
\includegraphics[width=1\linewidth]{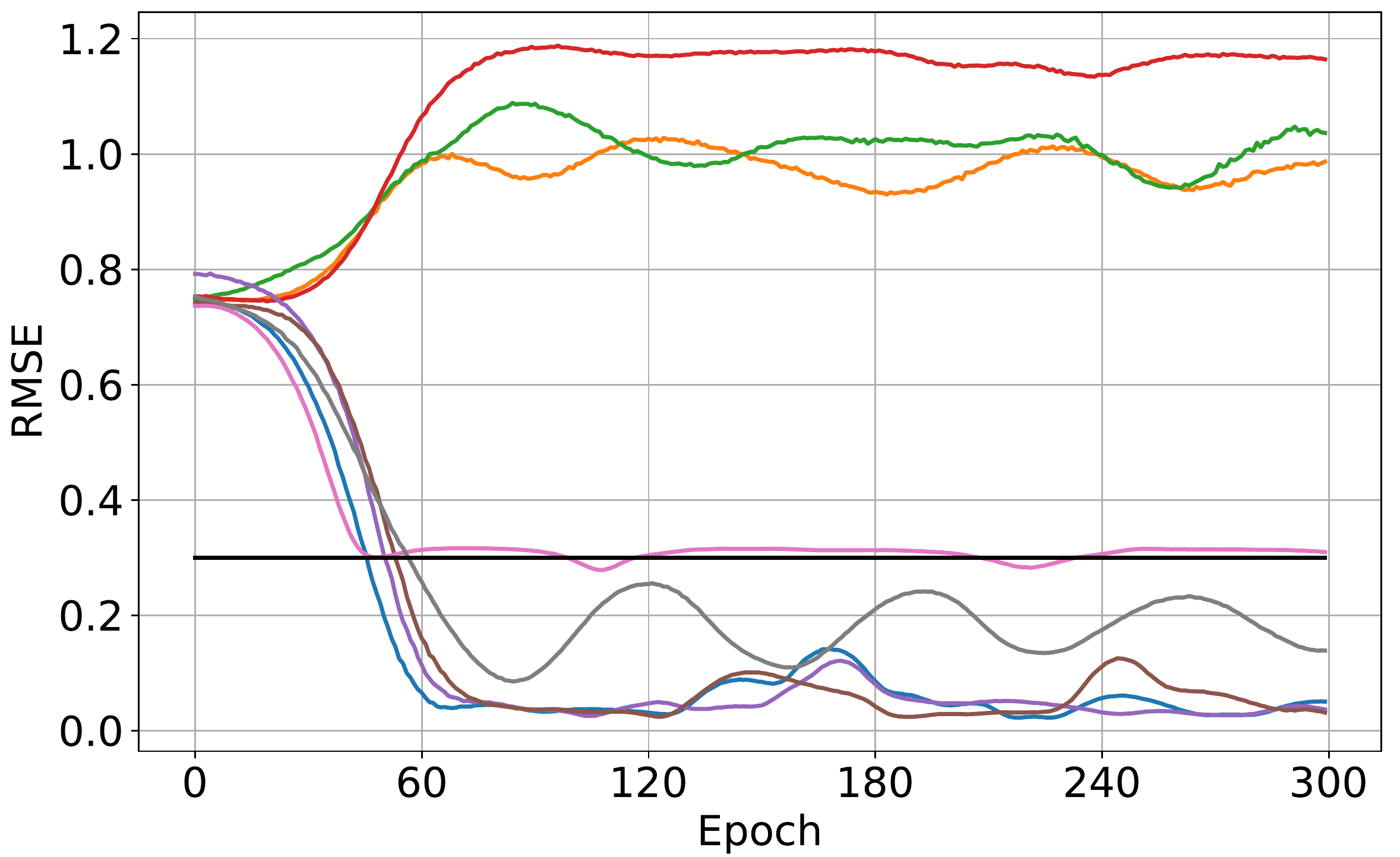}
\end{minipage}}
\subfigure{
\begin{minipage}{0.23\linewidth}
\includegraphics[width=1\linewidth]{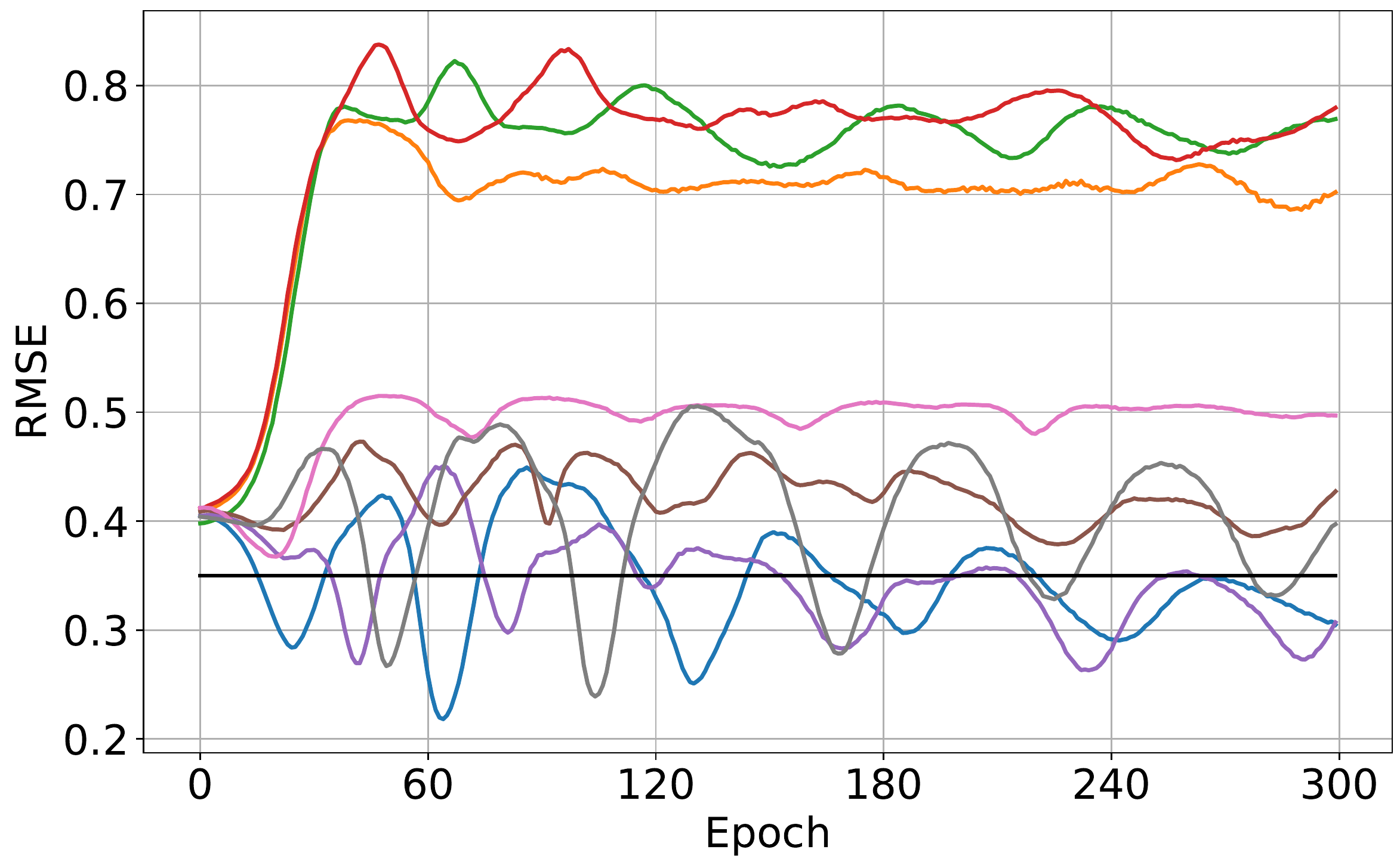}
\end{minipage}}
\subfigure{
\begin{minipage}{0.23\linewidth}
\includegraphics[width=1\linewidth]{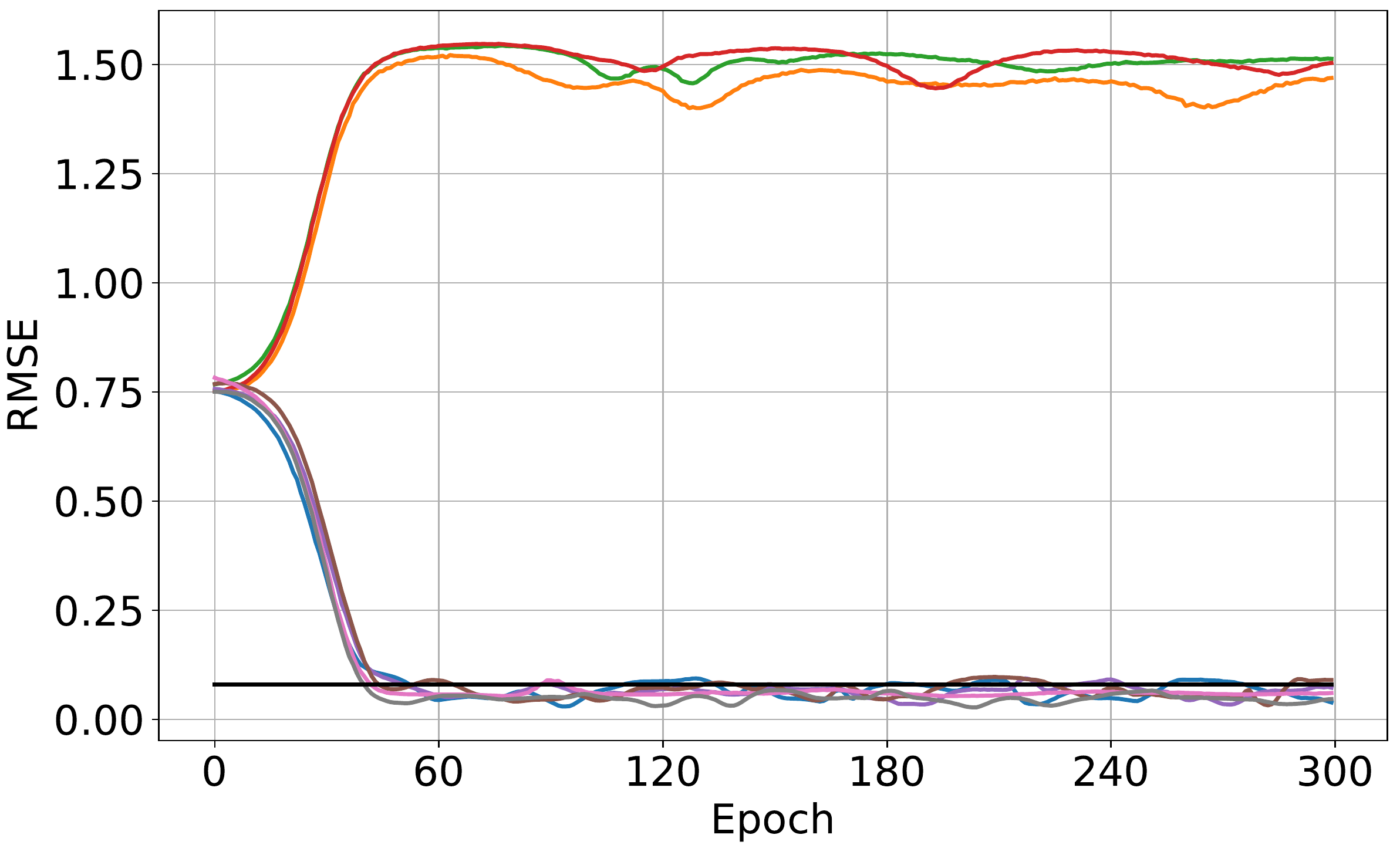}
\end{minipage}}
\subfigure{
\begin{minipage}{0.23\linewidth}
\includegraphics[width=1\linewidth]{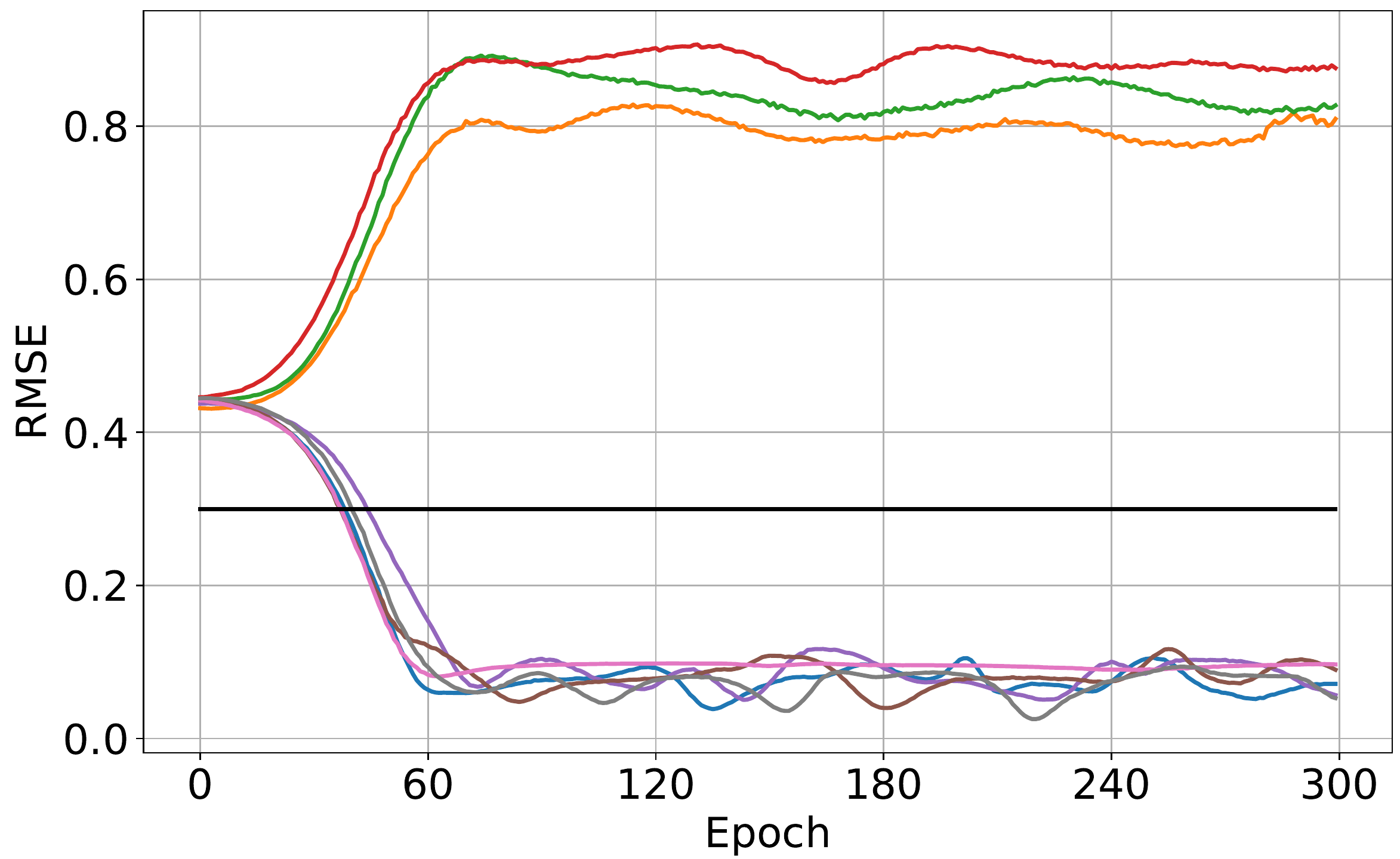}
\end{minipage}}
\setcounter{subfigure}{0}
\subfigure[KDD99]{
\begin{minipage}{0.23\linewidth}
\includegraphics[width=1\linewidth]{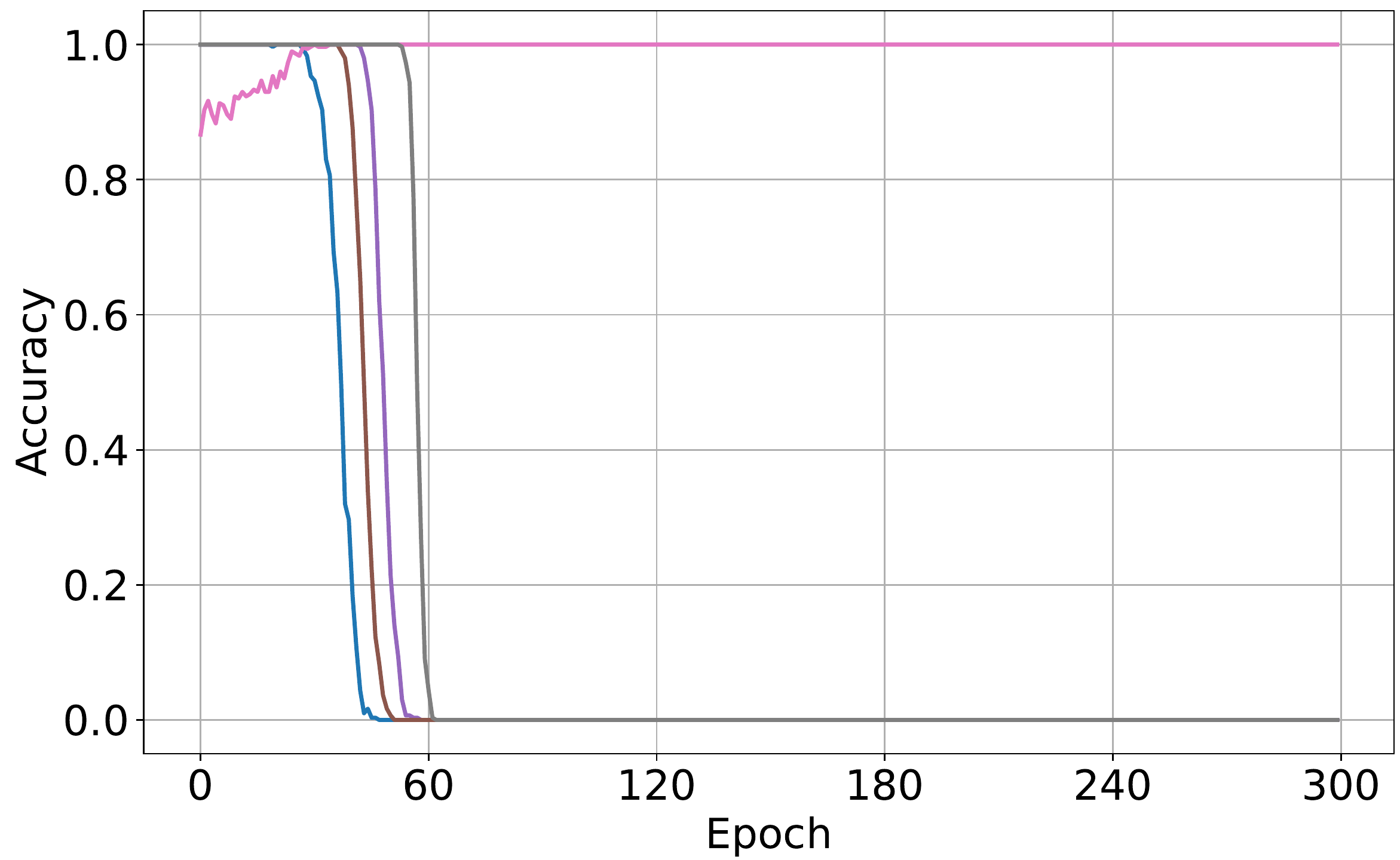}
\end{minipage}}
\subfigure[Mirai]{
\begin{minipage}{0.23\linewidth}
\includegraphics[width=1\linewidth]{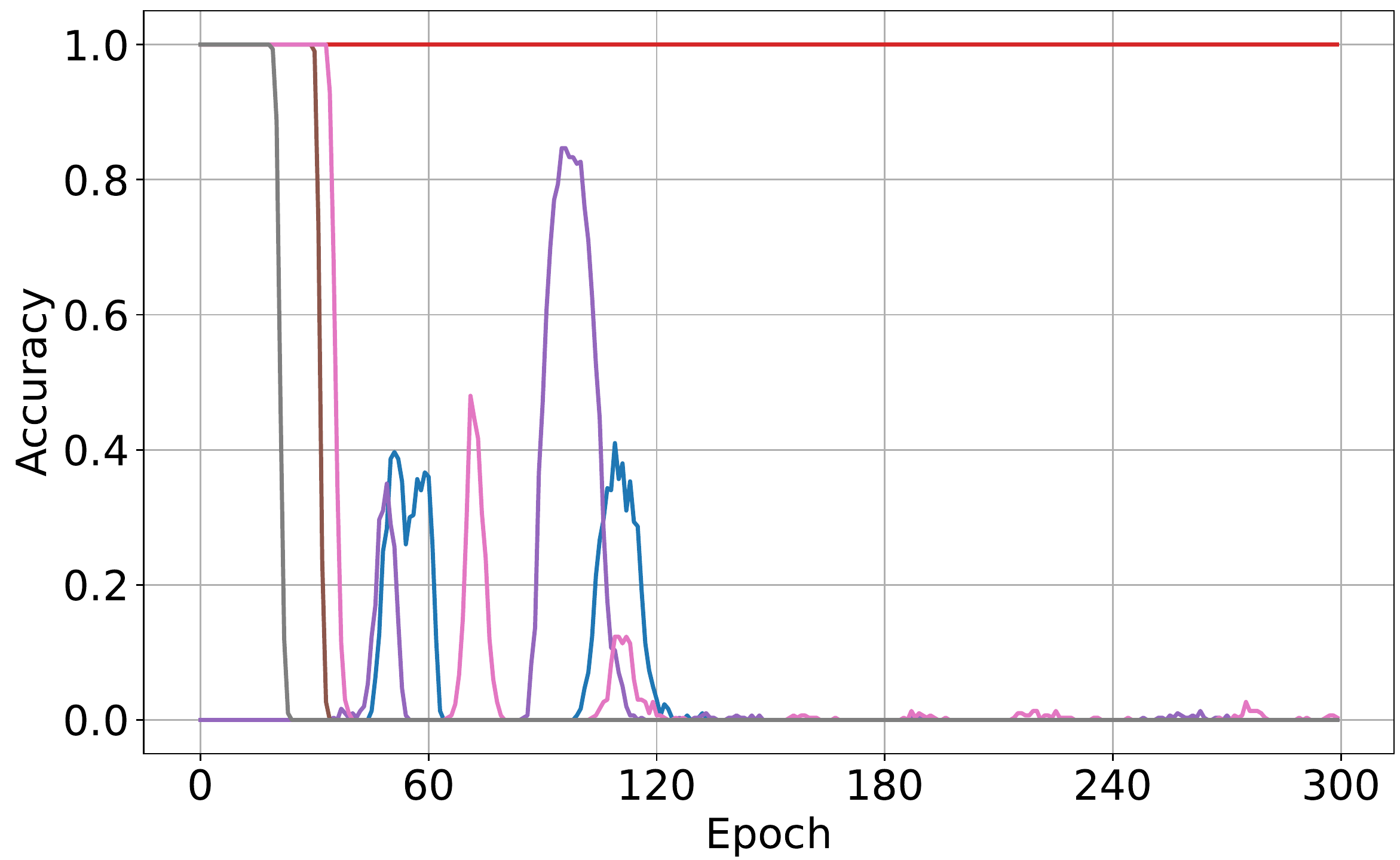}
\end{minipage}}
\subfigure[CIC-IDS2017]{
\begin{minipage}{0.23\linewidth}
\includegraphics[width=1\linewidth]{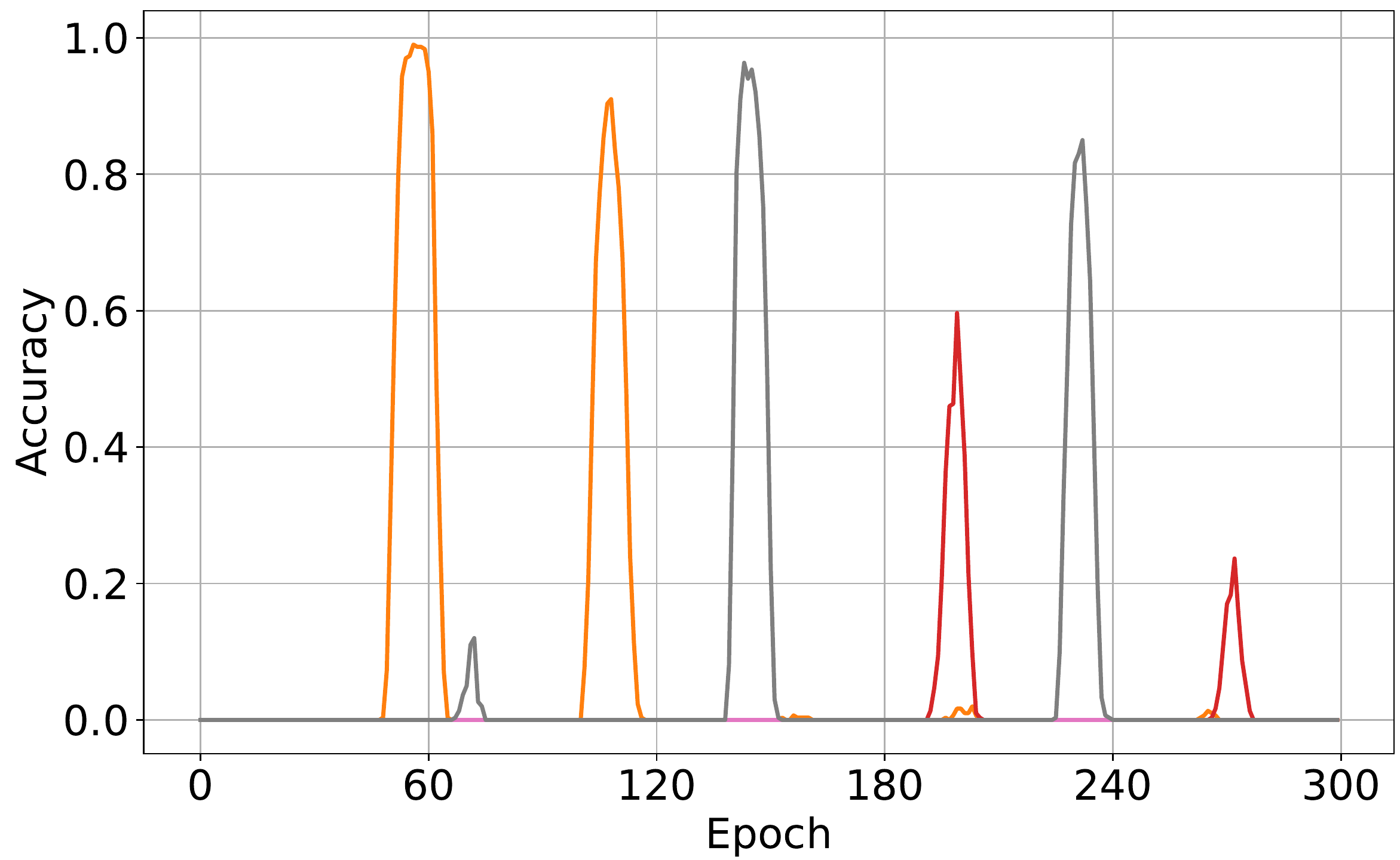}
\end{minipage}}
\subfigure[UNSW-NB15]{
\begin{minipage}{0.23\linewidth}
\includegraphics[width=1\linewidth]{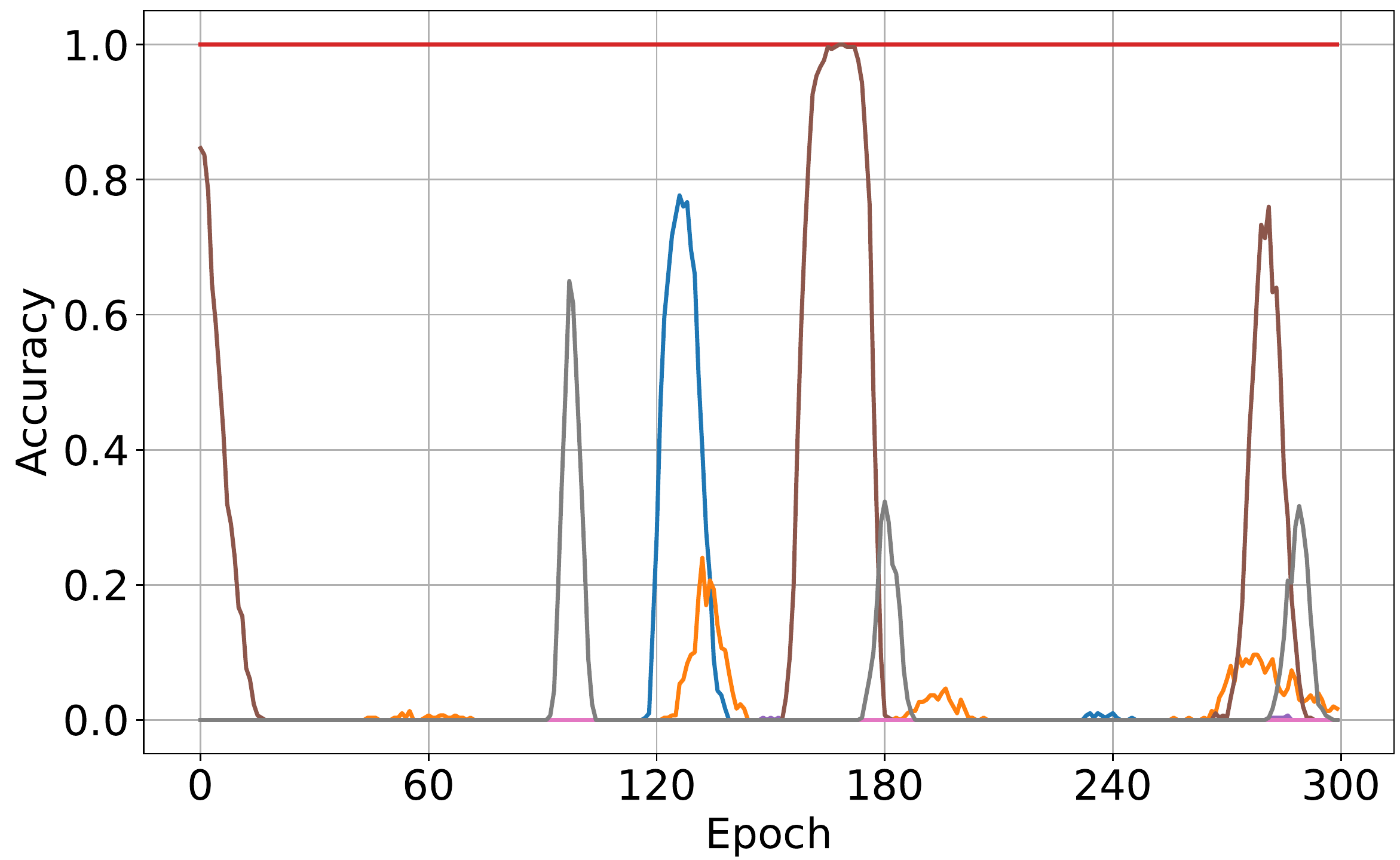}
\end{minipage}}
\caption{Black-box adversarial attack against two NIDSs with extraction attack for single sample reconstruction during early stage.
The first row is the average RMSE score change (higher is better) for Kitsune during the GAN training process, the black line represents the optimal threshold, and the second row represents the DNN accuracy change (lower is better).}
\label{gan_result}
\end{figure*}

\begin{figure}[htpb]
\setlength{\abovecaptionskip}{0pt}
\setlength{\belowcaptionskip}{0pt}
\setcounter{subfigure}{0}
\subfigure[Kitsune]{
\begin{minipage}{0.48\linewidth}
\includegraphics[width=1\linewidth]{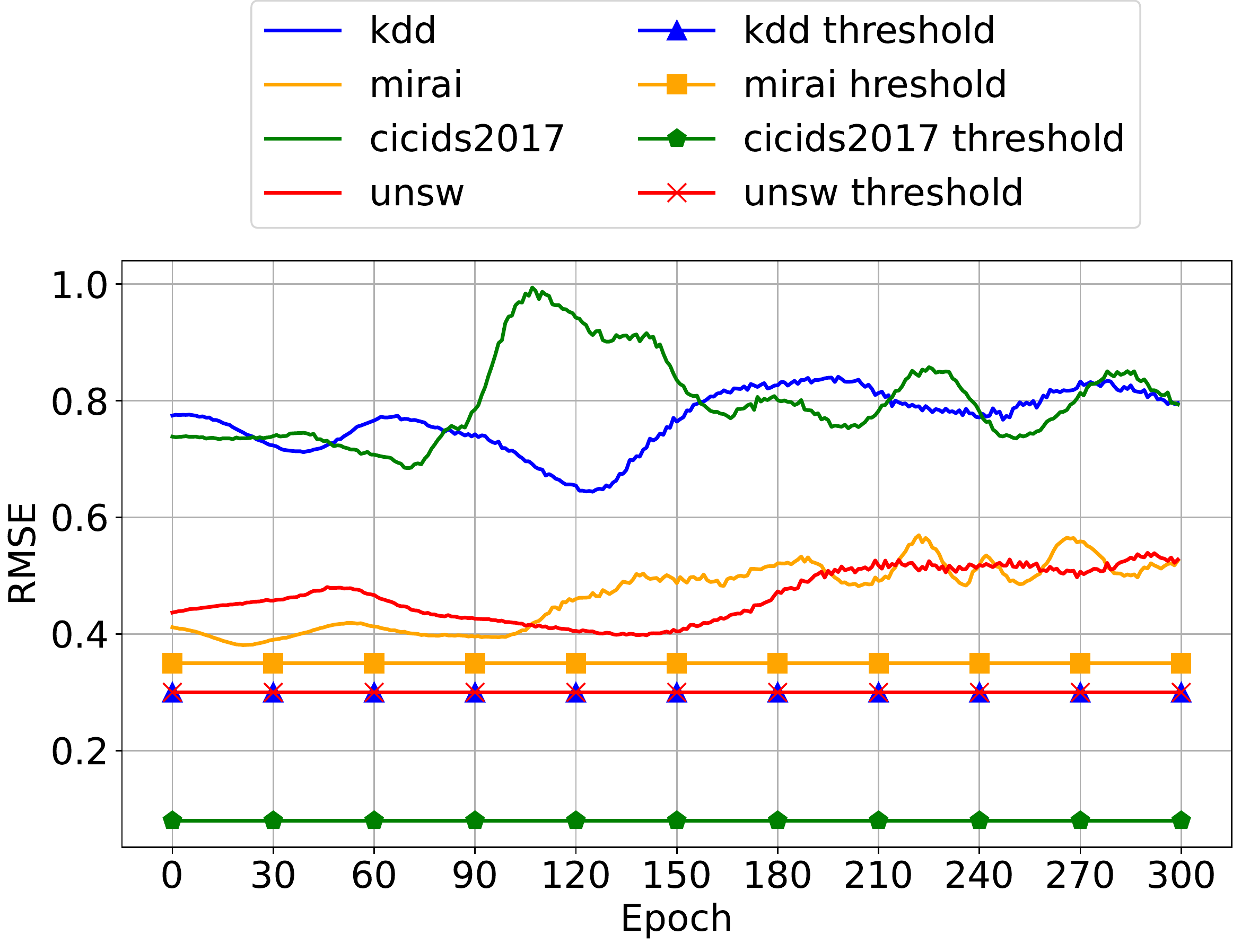}
\end{minipage}}
\subfigure[DNN]{
\begin{minipage}{0.48\linewidth}
\includegraphics[width=1\linewidth]{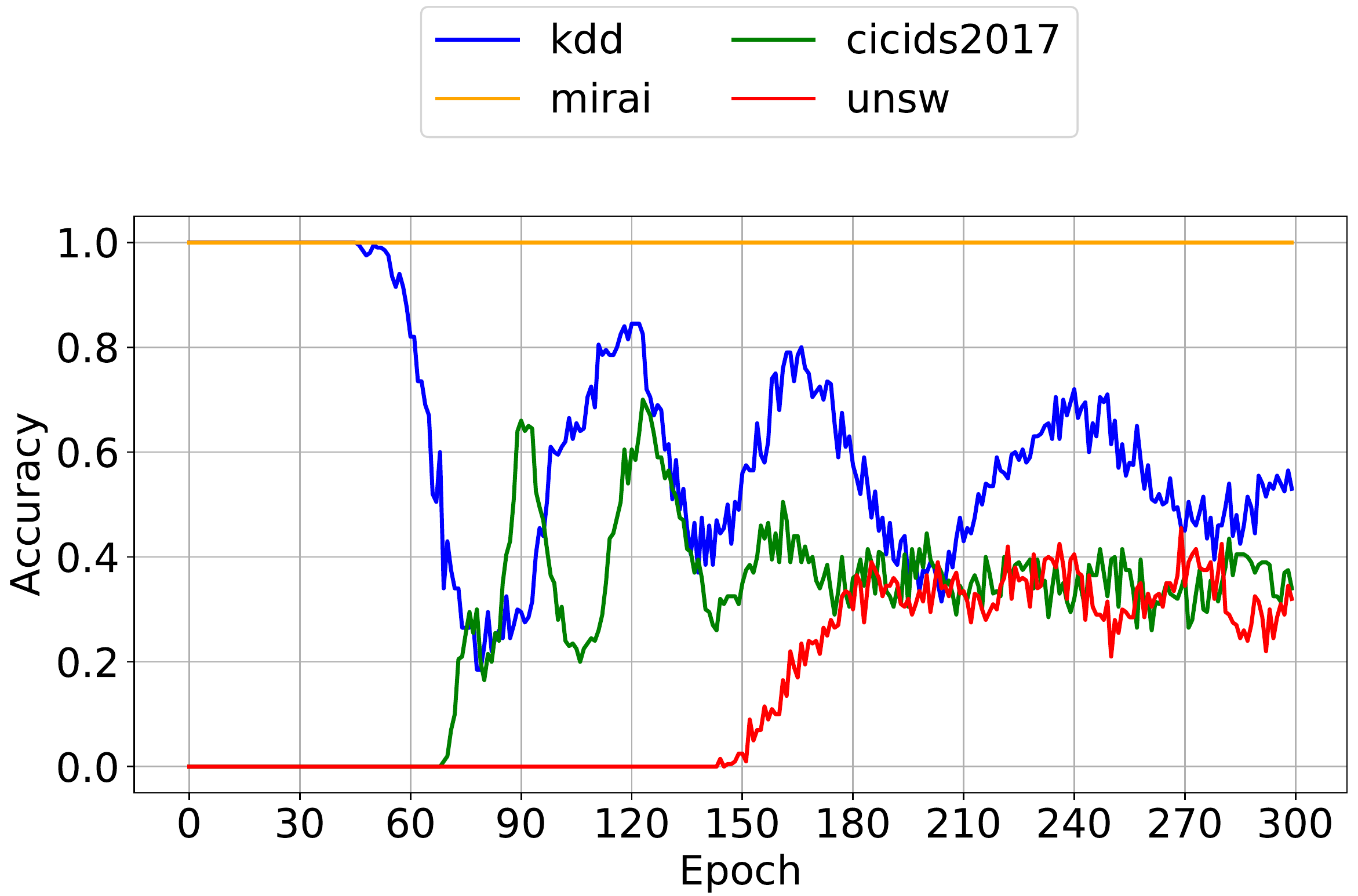}
\end{minipage}}
\caption{Black-box adversarial attack against two NIDSs with inversion attack for single sample against our defense with $\alpha=1$ during late stage.}
\label{gan_result_all_dataset}
\end{figure}

On the contrary, our defense outperforms other baselines and significantly mitigates such attacks in a way that the score is around 0.6-0.7 with $\alpha=1$, which is 1.5-7 times higher than the second best DP with privacy score around 0.1-0.4, and it tends to get even higher with smaller $\alpha$ and against stronger extraction attack. 
Our defense also prevents the label leakage especially for multi-class datasets with almost 0 label inversion accuracy compared to 0.26-1.00 for baselines, while it may perform worse for two-class datasets where labels are easier to obtain.
Overall, the privacy is well protected with little information leakage.
Interestingly, Soteria almost proves no privacy protection against either attack, we suspect it's because we don't leverage CNN as our feature extractor and that Soteria only perturbs the gradients of the defended layer while the rest still carry much private information.

During the late training stage, the shared gradients tend to get smaller, and such optimization-based attack performs worse because it's more difficult to optimize the dummy gradient to fit the original one.
As illustrated in Fig.~\ref{reconstruction_score_comparison} and TABLE~\ref{reconstruction_label_acc}, the privacy score using inversion attack varies from $0.2$ to $0.5$ even for gradients without defenses, and the label accuracy also drops from 1 to even 0.52 for CIC-IDS2017 dataset.
Nonetheless, our defense still generally outperforms other baselines where extraction attack performs with similar privacy score to that of early stage because matrix inversion is always accurate.
We also notice that smaller $\alpha$ may induce even lower score because the gradients are of small magnitude and thus may trigger falling back from extraction to inversion attack, which makes the reconstruction results more unstable.

\textbf{Setup 3}.
We first evaluate evasion rate ($ER$) using black-box attack as introduced in Section~\ref{GAN_attack}.
For early stage, we recover 100 reconstructed benign traffic using extraction attack to train GAN to generate 100 randomly initialized adversarial examples (AEs) against two NIDSs, i.e., Kitsune and the trained global DNN model with corresponding defenses.
While we repeat the process for late stage with inversion attack for our defense with $\alpha=1$ only since extraction attack presents similar privacy score.
For Kitsune setup, we train it with the four datasets and determine optimal thresholds respectively.
For the trained DNN-based FL model, we only consider either normal or malicious for evasion results.

\textbf{Black-box Adversarial Attack}.
The evasion results for early and late stages are in Fig.~\ref{gan_result} and Fig.~\ref{gan_result_all_dataset}.
For early stage, we can find that almost all baselines fail to prevent such adversarial attack ($ER=1$ for RMSEs lower than threshold) except for Mirai dataset, where threshold is more strict and DP may be sufficient with strong privacy guarantee, while our defense consistently outperforms baselines that the curve under FedDef doesn't converge with higher RMSE around 0.8-1.5, which is 2-15 times higher than threshold and thus AEs fail to evade Kitsune ($ER=0$).
However,
it's easier to evade the target DNN model ($ER=1$ when $ACC_{DNN}=0$) even with our defense for CIC-IDS2017 dataset.
It's because we also leverage DNN model to train GAN, therefore, evading discriminator also means likely evasion on DNN model.

While in late training stage, inversion attack can be quite unstable thus the recovered data approaches random guess, which is why RMSE is similar to that of randomly initialized data (RMSE is still higher than threshold and $ER=0$), and accuracy (also for $ER$) is around 0.5.

\subsubsection{Batched samples reconstruction}
We study a more practical scenario with multi-sample reconstruction.

\textbf{Setup 4}.
We evaluate $ER$ using inversion attack to reconstruct data and labels because we proved in Section \ref{threat_model} that extraction attack is only effective when batch size is 1.
Note that we don't consider the specific privacy score because the reconstructed data may have different permutations and it's hard to correspond them to the ground-truth data and label.
Instead, we directly apply the data with benign labels to train the GAN model for black-box attack scenario.

\begin{figure}[!h]
\setlength{\abovecaptionskip}{0pt}
\setlength{\belowcaptionskip}{0pt}
\centering
\subfigure{
\includegraphics[width=1\linewidth]{figures/sample2.pdf}
}
\setcounter{subfigure}{0}
\subfigure[Batch Size=5]{
\includegraphics[width=0.45\linewidth]{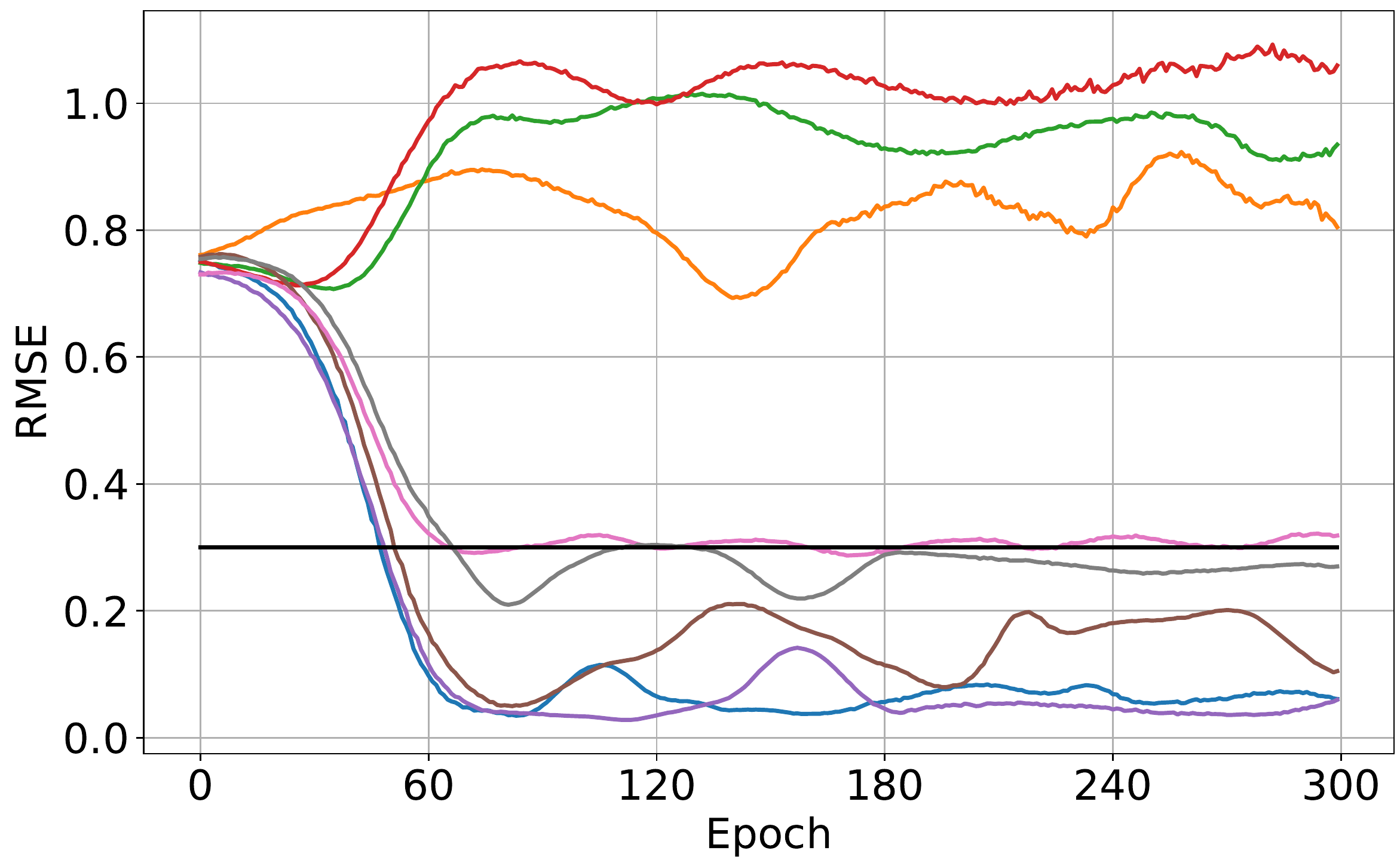}
}
\subfigure[Batch Size=10]{
\includegraphics[width=0.45\linewidth]{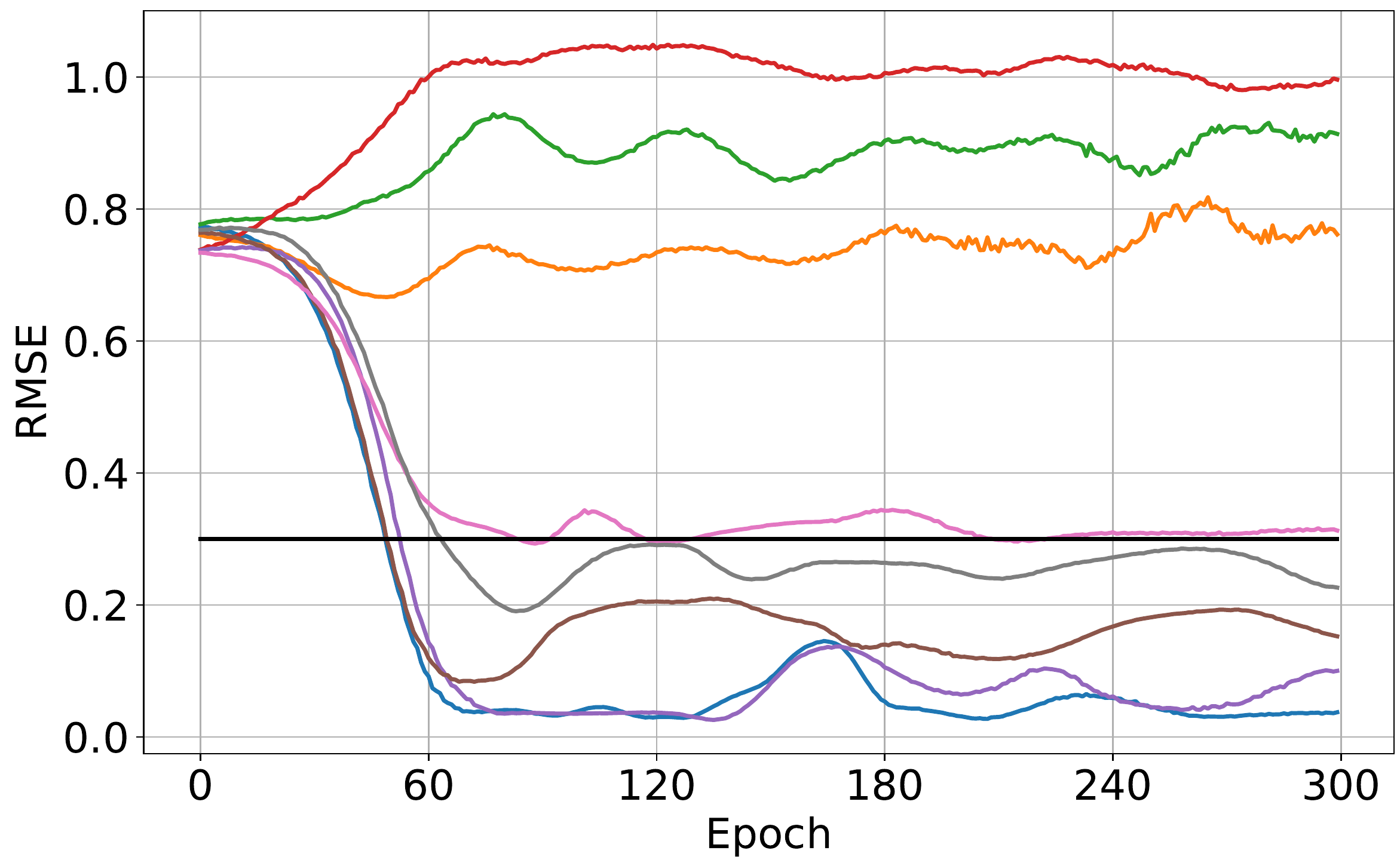}
}
\caption{Black-box adversarial attack for KDD99 dataset with inversion attack against Kitsune during early stage when batch size$=5/10$.}
\label{gan_result_bs}
\end{figure}

\textbf{Black-box Adversarial Attack}.
The representative results for KDD99 dataset with batch size 5 and 10 are illustrated in Fig.~\ref{gan_result_bs}.
We can observe that with larger batch size, the inversion attack may perform worse and therefore may slightly degrade our defense (RMSE is around 0.8-1.0 compared to 1.0-1.2 when batch size=1), yet the overall results are similar to that of single sample reconstruction ($ER=0$ for FedDef and $ER=1$ for baselines mostly), which further demonstrates the threats in practical training.

\subsubsection{\textbf{White-box adversarial attack}}
To further evaluate the privacy leakage in reconstructed data, we also conduct several white-box adversarial attacks against the two NIDSs.

\textbf{Setup 5}.
Recall that for GAN-based black-box attack, we recover benign traffic so as to generate malicious traffic from random examples.
While in white-box attack scenario, we directly reconstruct malicious training data and perform different adversarial attacks on the samples themselves to evade the target model.

Specifically, we recover 100 malicious traffic (based on the reconstructed label) using inversion attack per sample in early stage for different defenses.
Then we perform the following attacks to generate adversarial examples:

$\bullet$
FGSM \cite{goodfellow2014explaining}
utilizes the sign of the gradient of the cross-entropy (CE) loss associated with the target label to obtain perturbation.

$\bullet$
CW \cite{carlini2017towards}
generates adversarial perturbations by solving a norm-restricted constrained optimization problem, where the loss has in contrast to the CE loss a direct interpretation in terms of the decision of the classifier.

$\bullet$
PGD \cite{madry2017towards} is an extension of the FGSM.
It initializes the attack at a random point in the $L_p$ ball constraint and projects the perturbation back onto the $L_p$ ball after every iteration.

$\bullet$
DeepFool \cite{moosavi2016deepfool} is an untargeted attack under the assumption that deep learning models are linear, with a decision boundary (i.e., hyperplanes) separating each class.
On every iteration, DeepFool linearizes the classifier around the current point $x$ and computes the perturbations as an orthogonal projection vector that projects $x$ onto the closest hyperplane.

$\bullet$
AutoPGD \cite{croce2020reliable} is a budget-aware step size-free variant of PGD that induces failures due to suboptimal step size and the objective function.
Instead it automatically adjusts the step size and chooses Difference of Logits Ratio (DLR) loss, which is both shift and rescaling invariant, and thus has the same degrees of freedom as the decision of the classifier.

For implementaion, we leverage the open-source package torchattacks \cite{kim2020torchattacks}, which provides unified API for different adversarial attacks.
For attack initialization, we present some key parameters for each attack in TABLE \ref{white-box attack}.
Note that for CW attack, we leverage parameter $c$ to adjust the weights for distance optimization.
For DeepFool, the perturbation is usually too large to be valid, therefore, we manually add $\epsilon$ parameter to constrain the maximum perturbation on the final output to ensure practical attack.
For other parameters, we follow the default setting in torchattacks for each attack.

\begin{table}[htbp]
\caption{White-box attack parameter setting.}
\label{white-box attack}
\resizebox{\linewidth}{!}
{\begin{tabular}{|c|c|c|}
\hline
Attack   & Key Parameter  & Description \\ \hline
FGSM     & $\epsilon$            & maximum perturbation           \\ \hline
CW       & $c$              & constraint for $L_2$ distance           \\ \hline
PGD      & $\epsilon$, $\alpha$, $step$ & maximum perturbation, single perturbation, attack step size           \\ \hline
DeepFool & $\epsilon$            & maximum perturbation          \\ \hline
AutoPGD  & $\epsilon$            & maximum perturbation           \\ \hline
\end{tabular}}
\end{table}

For the trained DNN model, we apply the 5 attacks mentioned above in targeted mode by labels where malicious traffic are projected into benign class, except for untargeted DeepFool that perturbs the traffic in the direction of the closest label class.
Specifically, we empirically set $\epsilon=40/255, \alpha=6/255$ for KDD99, Mirai, UNSW datasets, $\epsilon=4/255, \alpha=2/255$ for CIC2017 dataset, and $c=0.01, step=100$ for all scenarios.

For the unsupervised Kitsune model, the output is the anomaly score instead of possibility vector.
Therefore, DeepFool and AutoPGD are unavailable because DeepFool requires the label classification of the adversarial examples to compute the perturbation towards the closest class.
While AutoPGD also needs the output of DNN vector to obtain DLR loss and update the step size.
Therefore, we only adapt FGSM, CW, and PGD attacks, where we change the CE loss to the anomaly score output from Kitsune.
In this way, we can minimize the score to fool the model to classify the adversarial examples as benign as long as the score is lower than the threshold.

In this scenario, we can conduct ablation study on the key parameters since anomaly score can convey direct impact of the parameter on the adversarial attacks.
Specifically, we choose $c=1e-2, 1e-3, 1e-4$ for CW attack and $\epsilon=10/255, 40/255, 80/255$ for PGD attack, $\epsilon=40/255$ is constant for FGSM attack since it's a simple version of PGD.

\begin{table*}[htbp]
\caption{White-box adversarial attack against DNN model. We present as (model accuracy)/($L_2$ perturbation distance) as below, higher is better from defenders' view.}
\label{white-dnn}
\resizebox{\linewidth}{!}
{\begin{tabular}{|c|c|c|ccc|c|c|c|c|}
\hline
\multirow{2}{*}{Datasets}    & \multirow{2}{*}{\diagbox[]{Attack}{Defense}} & \multirow{2}{*}{No Defense} & \multicolumn{3}{c|}{FedDef}                              & \multirow{2}{*}{Soteria} & \multirow{2}{*}{GP} & \multirow{2}{*}{DP} & \multirow{2}{*}{Instanhide} \\ \cline{4-6}
                             &                         &                             & \multicolumn{1}{c|}{$\alpha=1$} & \multicolumn{1}{c|}{$\alpha=0.5$} & $\alpha=0.25$ &                          &                     &                     &                             \\ \hline
\multirow{5}{*}{KDD99}       & FGSM                    & \textcolor{red}{6\%}/0.12                           & \multicolumn{1}{c|}{62\%/0.12} & \multicolumn{1}{c|}{\textcolor{green}{100\%}/0.13}   & 100\%/0.12    & 25\%/0.11                        & 24\%/0.13                   & 9\%/0.14                   & 52\%/0.13                           \\ 
                             & CW                      & 24\%/0.10                           & \multicolumn{1}{c|}{63\%/0.08} & \multicolumn{1}{c|}{\textcolor{green}{100\%}/0.00}   & 100\%/0.00    & 42\%/0.08                        & \textcolor{red}{0\%}/0.09                   & 1\%/0.09                   & 50\%/0.06                           \\  
                             & PGD                     & \textcolor{red}{0\%}/0.11                           & \multicolumn{1}{c|}{57\%/0.11} & \multicolumn{1}{c|}{\textcolor{green}{100\%}/0.12}   & 100\%/0.09    & 22\%/0.11                        & 24\%/0.12                   & 6\%/0.13                   & 47\%/0.12                           \\  
                             & DeepFool                & \textcolor{red}{7\%}/0.11                           & \multicolumn{1}{c|}{64\%/0.10} & \multicolumn{1}{c|}{\textcolor{green}{100\%}/0.02}   & 100\%/0.02    & 27\%/0.11                        & 28\%/0.10                   & 23\%/0.06                   & 92\%/0.01                           \\  
                             & AutoPGD                 & \textcolor{red}{1\%}/0.12                           & \multicolumn{1}{c|}{61\%/0.11} & \multicolumn{1}{c|}{\textcolor{green}{100\%}/0.08}   & 100\%/0.09    & 22\%/0.10                        & 24\%/0.13                   & 7\%/0.14                   & 81\%/0.09                           \\ \hline
\multirow{5}{*}{Mirai}       & FGSM                    & \textcolor{red}{0\%}/0.14                           & \multicolumn{1}{c|}{75\%/0.11} & \multicolumn{1}{c|}{\textcolor{green}{100\%}/0.12}   & 100\%/0.11    & 27\%/0.14                        & 64\%/0.13                   & 56\%/0.13                   & 50\%/0.14                           \\   
                             & CW                      & 31\%/0.03                           & \multicolumn{1}{c|}{91\%/0.01} & \multicolumn{1}{c|}{\textcolor{green}{100\%}/0.00}   & 100\%/0.00    & 67\%/0.04                        & \textcolor{red}{0\%}/0.17                   & 2\%/0.16                   & 12\%/0.11                           \\  
                             & PGD                     & \textcolor{red}{0\%}/0.13                           & \multicolumn{1}{c|}{75\%/0.11} & \multicolumn{1}{c|}{\textcolor{green}{100\%}/0.12}   & 100\%/0.11    & 24\%/0.14                        & 64\%/0.13                   & 47\%/0.12                   & 44\%/0.13                           \\  
                             & DeepFool                & \textcolor{red}{0\%}/0.12                           & \multicolumn{1}{c|}{75\%/0.11} & \multicolumn{1}{c|}{\textcolor{green}{100\%}/0.12}   & 100\%/0.11    & 28\%/0.13                        & 90\%/0.11                   & 66\%/0.11                   & 95\%/0.06                           \\  
                             & AutoPGD                 & \textcolor{red}{0\%}/0.13                           & \multicolumn{1}{c|}{75\%/0.08} & \multicolumn{1}{c|}{\textcolor{green}{100\%}/0.07}   & 100\%/0.09    & 24\%/0.13                        & 63\%/0.10                   & 47\%/0.11                   & 65\%/0.11                           \\ \hline
\multirow{5}{*}{CIC-IDS2017} & FGSM                    & \textcolor{red}{0\%}/0.01                           & \multicolumn{1}{c|}{74\%/0.01} & \multicolumn{1}{c|}{\textcolor{green}{100\%}/0.01}   & 100\%/0.01    & 0\%/0.01                        & 0\%/0.01                   & 2\%/0.01                   & 20\%/0.02                           \\  
                             & CW                      & \textcolor{red}{0\%}/0.00                           & \multicolumn{1}{c|}{69\%/0.01} & \multicolumn{1}{c|}{\textcolor{green}{100\%}/0.00}   & 100\%/0.00    & 0\%/0.00                        & 0\%/0.01                   & 0\%/0.01                   & 0\%/0.03                           \\  
                             & PGD                     & \textcolor{red}{0\%}/0.01                           & \multicolumn{1}{c|}{74\%/0.01} & \multicolumn{1}{c|}{\textcolor{green}{100\%}/0.01}   & 100\%/0.01    & 0\%/0.01                        & 0\%/0.01                   & 2\%/0.01                   & 20\%/0.01                           \\  
                             & DeepFool                & \textcolor{red}{0\%}/0.00                           & \multicolumn{1}{c|}{74\%/0.01} & \multicolumn{1}{c|}{\textcolor{green}{100\%}/0.00}   & 100\%/0.00    & 0\%/0.00                        & 4\%/0.00                   & 2\%/0.00                   & 33\%/0.00                           \\  
                             & AutoPGD                 & \textcolor{red}{0\%}/0.00                           & \multicolumn{1}{c|}{74\%/0.01} & \multicolumn{1}{c|}{\textcolor{green}{100\%}/0.01}   & 100\%/0.01    & 0\%/0.00                        & 0\%/0.01                   & 2\%/0.01                   & 30\%/0.01                           \\ \hline
\multirow{5}{*}{UNSW-NB15}   & FGSM                    & 18\%/0.12                           & \multicolumn{1}{c|}{71\%/0.13} & \multicolumn{1}{c|}{91\%/0.12}   & \textcolor{green}{100\%}/0.10    & 10\%/0.12                        & \textcolor{red}{0\%}/0.13                   & 3\%/0.13                   & 2\%/0.14                           \\  
                             & CW                      & 1\%/0.09                           & \multicolumn{1}{c|}{75\%/0.02} & \multicolumn{1}{c|}{\textcolor{green}{100\%}/0.05}   & 100\%/0.00    & \textcolor{red}{0\%}/0.08                        & 0\%/0.03                   & 1\%/0.05                   & 0\%/0.04                           \\  
                             & PGD                     & \textcolor{red}{0\%}/0.12                           & \multicolumn{1}{c|}{65\%/0.13} & \multicolumn{1}{c|}{91\%/0.11}   & \textcolor{green}{100\%}/0.09    & 0\%/0.12                        & 0\%/0.12                   & 0\%/0.13                   & 2\%/0.13                           \\  
                             & DeepFool                & 46\%/0.07                           & \multicolumn{1}{c|}{90\%/0.03} & \multicolumn{1}{c|}{\textcolor{green}{100\%}/0.03}   & 100\%/0.02    & 44\%/0.07                        & 32\%/0.01                   & 37\%/0.02                   & \textcolor{red}{22\%}/0.00                           \\  
                             & AutoPGD                 & \textcolor{red}{0\%}/0.12                           & \multicolumn{1}{c|}{79\%/0.10} & \multicolumn{1}{c|}{91\%/0.10}   & \textcolor{green}{100\%}/0.10    & 0\%/0.12                        & 0\%/0.13                   & 1\%/0.14                   & 11\%/0.12                           \\ \hline
\end{tabular}}
\end{table*}

\textbf{DNN Analysis.}
TABLE~\ref{white-dnn} illustrate the results of white-box adversarial attacks against DNN model.
We present the classification accuracy and the perturbation distance of the adversarial examples under different attacks and defenses. 
Higher accuracy means better defense because the reconstructed traffic are further from the original distribution of the user training data, therefore, limited perturbation is not enough to evade the target model.
The evasion rate is consistent with the definition in Section \ref{GAN_attack} that $ER=1-ACC_{DNN}$.

From defenders' view, we can conclude that our defense FedDef achieves the highest accuracy for most of the cases when $\alpha=1$ (57\%-91\%) and all the cases when $\alpha=0.5, 0.25$ (91\%-100\%).
This corresponds with the fact that lower $\alpha$ means better privacy optimization, therefore, the reconstructed malicious traffic carry less information and can be harder to be perturbed.
While Instahide also performs well with high accuracy, which means that mixed training data can already achieve moderate privacy guarantee by incorporating different label classes.
Soteria and GP still perform worse since the reconstructed data are just similar to the baseline defense.

From attackers' view, PGD attack generally works better with lower accuracy (also higher evasion rate) with reasonable perturbation, followed by FGSM, CW, AutoPGD, and DeepFool.
This illustrates that DLR loss in AutoPGD may not be an optimal choice compared with traditional CE loss in this case.
While DeepFool only targets the nearest label class instead of benign class and thus cannot work as well as other targeted attacks.

In some scenarios, CW attack works better than PGD with higher accuracy and lower perturbation range (e.g., 12\% vs. 44\% for Instahide defense for Mirai dataset), thanks to the optimization of both perturbation distance and evasion rate.
During the experiments, some adversarial examples are not perturbed at all (distance=0), this is because CW only updates the examples when classification reaches the target benign label and the perturbation distance is lower.
Therefore, it won't perform perturbation if the reconstructed malicious data themselves are benign (0\%/0.00 for No defense for CIC2017 dataset) since distance=0 is already the lowest case, or if the adversarial examples can never reach benign class (100\%/0.00 for FedDef when $\alpha=0.5$ for Mirai and CIC2017 datasets).

\textbf{Kitsune Analysis.}
We present model accuracy and the average anomaly score for Kitsune model in TABLE~\ref{white-kitsune-acc}.
Among the 100 perturbed reconstructed traffic data, those with scores lower than the model threshold are considered to successfully evade Kitsune.
Therefore, higher anomaly score and accuracy means better defense performance.

\begin{table*}[htbp]
\caption{White-box adversarial attack against Kitsune model. We present as (model accuracy)/(anomaly score) as below, higher is better from defenders' view.}
\label{white-kitsune-acc}
\resizebox{\linewidth}{!}
{\begin{tabular}{|c|cl|c|ccc|c|c|c|c|}
\hline
\multirow{2}{*}{Datasets}    & \multicolumn{2}{c|}{\multirow{2}{*}{\diagbox[]{Attack}{Defense}}}       & \multirow{2}{*}{No Defense} & \multicolumn{3}{c|}{FedDef}                              & \multirow{2}{*}{Soteria} & \multirow{2}{*}{GP} & \multirow{2}{*}{DP} & \multirow{2}{*}{Instanhide} \\ \cline{5-7}
                             & \multicolumn{2}{c|}{}                              &                             & \multicolumn{1}{c|}{$\alpha=1$} & \multicolumn{1}{c|}{$\alpha=0.5$} & $\alpha=0.25$ &                          &                     &                     &                             \\ \hline
\multirow{7}{*}{KDD99}       & \multicolumn{2}{c|}{FGSM}                          & 19\%/0.21                           & \multicolumn{1}{c|}{\textcolor{green}{100\%}/0.88} & \multicolumn{1}{c|}{100\%/0.97}   & 100\%/0.94    & 17\%/0.21                        & \textcolor{red}{0\%}/0.17                   & 44\%/0.31                   & 97\%/0.61                           \\ \cline{2-11} 
                             & \multicolumn{1}{c|}{\multirow{3}{*}{CW}}  & $c=1e-2$   & 25\%/0.27                           & \multicolumn{1}{c|}{\textcolor{green}{100\%}/1.06} & \multicolumn{1}{c|}{100\%/1.16}   & 100\%/1.13    & 27\%/0.26                        & \textcolor{red}{24\%}/0.24                   & 81\%/0.42                   & 100\%/0.76                           \\ 
                             & \multicolumn{1}{c|}{}                     & $c=1e-3$   & 19\%/0.26                           & \multicolumn{1}{c|}{\textcolor{green}{100\%}/1.06} & \multicolumn{1}{c|}{100\%/1.15}   & 100\%/1.13    & 8\%/0.25                        & \textcolor{red}{6\%}/0.23                   & 42\%/0.38                   & 95\%/0.75                           \\  
                             & \multicolumn{1}{c|}{}                     & $c=1e-4$   & \textcolor{red}{0\%}/0.24                           & \multicolumn{1}{c|}{\textcolor{green}{100\%}/1.06} & \multicolumn{1}{c|}{100\%/1.15}   & 100\%/1.13    & 0\%/0.25                        & 4\%/0.23                   & 19\%/0.34                   & 85\%/0.73                           \\ \cline{2-11} 
                             & \multicolumn{1}{c|}{\multirow{3}{*}{PGD}} & $\epsilon=10/255$ & 25\%/0.25                           & \multicolumn{1}{c|}{\textcolor{green}{100\%}/1.01} & \multicolumn{1}{c|}{100\%/1.11}   & 100\%/1.08    & 26\%/0.25                        & \textcolor{red}{24\%}/0.22                   & 81\%/0.39                   & 100\%/0.74                           \\  
                             & \multicolumn{1}{c|}{}                     & $\epsilon=40/255$ & 16\%/0.21                           & \multicolumn{1}{c|}{\textcolor{green}{100\%}/0.87} & \multicolumn{1}{c|}{100\%/0.96}   & 100\%/0.94    & 14\%/0.21                        & \textcolor{red}{0\%}/0.16                   & 42\%/0.30                   & 97\%/0.61                           \\  
                             & \multicolumn{1}{c|}{}                     & $\epsilon=80/255$ & \textcolor{red}{0\%}/0.16                           & \multicolumn{1}{c|}{\textcolor{green}{100\%}/0.70} & \multicolumn{1}{c|}{100\%/0.77}   & 100\%/0.75    & 0\%/0.16                        & 0\%/0.10                   & 16\%/0.21                   & 85\%/0.48                           \\ \hline
\multirow{7}{*}{Mirai}       & \multicolumn{2}{c|}{FGSM}                          & 8\%/0.32                           & \multicolumn{1}{c|}{\textcolor{green}{100\%}/0.60} & \multicolumn{1}{c|}{100\%/0.65}   & 100\%/0.63    & \textcolor{red}{1\%}/0.31                        & 47\%/0.34                   & 89\%/0.41                   & 79\%/0.41                           \\ \cline{2-11} 
                             & \multicolumn{1}{c|}{\multirow{3}{*}{CW}}  & $c=1e-2$   & \textcolor{red}{67\%}/0.39                           & \multicolumn{1}{c|}{\textcolor{green}{100\%}/0.77} & \multicolumn{1}{c|}{100\%/0.82}   & 100\%/0.81    & 86\%/0.40                        & 97\%/0.47                   & 100\%/0.57                   & 100\%/0.57                           \\ 
                             & \multicolumn{1}{c|}{}                     & $c=1e-3$   & \textcolor{red}{29\%}/0.36                           & \multicolumn{1}{c|}{\textcolor{green}{100\%}/0.77} & \multicolumn{1}{c|}{100\%/0.82}   & 100\%/0.81    & 68\%/0.39                        & 59\%/0.43                   & 100\%/0.57                   & 98\%/0.57                           \\  
                             & \multicolumn{1}{c|}{}                     & $c=1e-4$   & \textcolor{red}{17\%}/0.36                           & \multicolumn{1}{c|}{\textcolor{green}{100\%}/0.77} & \multicolumn{1}{c|}{100\%/0.82}   & 100\%/0.81    & 34\%/0.37                        & 18\%/0.36                   & 75\%/0.52                   & 25\%/0.42                           \\ \cline{2-11} 
                             & \multicolumn{1}{c|}{\multirow{3}{*}{PGD}} & $\epsilon=10/255$ & \textcolor{red}{90\%}/0.38                           & \multicolumn{1}{c|}{\textcolor{green}{100\%}/0.72} & \multicolumn{1}{c|}{100\%/0.78}   & 100\%/0.76    & 93\%/0.37                       & 93\%/0.43                   & 100\%/0.52                   & 100\%/0.53                           \\  
                             & \multicolumn{1}{c|}{}                     & $\epsilon=40/255$ & 4\%/0.31                           & \multicolumn{1}{c|}{\textcolor{green}{100\%}/0.60} & \multicolumn{1}{c|}{100\%/0.65}   & 100\%/0.63    & \textcolor{red}{0\%}/0.31                        & 44\%/0.33                   & 84\%/0.41                   & 75\%/0.41                           \\  
                             & \multicolumn{1}{c|}{}                     & $\epsilon=80/255$ & \textcolor{red}{0\%}/0.23                           & \multicolumn{1}{c|}{\textcolor{green}{100\%}/0.44} & \multicolumn{1}{c|}{100\%/0.48}   & 100\%/0.46    & 0\%/0.23                        & 0\%/0.22                   & 9\%/0.28                   & 14\%/0.27                           \\ \hline
\multirow{7}{*}{CIC-IDS2017} & \multicolumn{2}{c|}{FGSM}                          & \textcolor{red}{100\%}/0.19                           & \multicolumn{1}{c|}{\textcolor{green}{100\%}/1.03} & \multicolumn{1}{c|}{100\%/1.06}   & 100\%/1.14    & 100\%/0.19                        & 100\%/0.26                   & 100\%/0.60                   & 100\%/0.58                           \\ \cline{2-11} 
                             & \multicolumn{1}{c|}{\multirow{3}{*}{CW}}  & $c=1e-2$   & \textcolor{red}{100\%}/0.26                           & \multicolumn{1}{c|}{\textcolor{green}{100\%}/1.24} & \multicolumn{1}{c|}{100\%/1.28}   & 100\%/1.40    & 100\%/0.28                        & 100\%/0.41                   & 100\%/0.77                   & 100\%/0.77                           \\ 
                             & \multicolumn{1}{c|}{}                     & $c=1e-3$   & \textcolor{red}{100\%}/0.26                           & \multicolumn{1}{c|}{\textcolor{green}{100\%}/1.24} & \multicolumn{1}{c|}{100\%/1.28}   & 100\%/1.37    & 100\%/0.28                        & 100\%/0.41                   & 100\%/0.77                   & 100\%/0.77                           \\  
                             & \multicolumn{1}{c|}{}                     & $c=1e-4$   & 78\%/0.23                           & \multicolumn{1}{c|}{\textcolor{green}{100\%}/1.24} & \multicolumn{1}{c|}{100\%/1.28}   & 100\%/1.37    & 82\%/0.26                        & \textcolor{red}{49\%}/0.27                   & 100\%/0.77                   & 100\%/0.77                           \\ \cline{2-11} 
                             & \multicolumn{1}{c|}{\multirow{3}{*}{PGD}} & $\epsilon=10/255$ & \textcolor{red}{100\%}/0.23                           & \multicolumn{1}{c|}{\textcolor{green}{100\%}/1.19} & \multicolumn{1}{c|}{100\%/1.22}   & 100\%/1.31    & 100\%/0.25                        & 100\%/0.36                   & 100\%/0.72                   & 100\%/0.72                           \\  
                             & \multicolumn{1}{c|}{}                     & $\epsilon=40/255$ & \textcolor{red}{100\%}/0.17                           & \multicolumn{1}{c|}{\textcolor{green}{100\%}/1.03} & \multicolumn{1}{c|}{100\%/1.05}   & 100\%/1.14    & 100\%/0.18                       & 100\%/0.25                   & 100\%/0.59                   & 100\%/0.57                           \\  
                             & \multicolumn{1}{c|}{}                     & $\epsilon=80/255$ & 91\%/0.12                           & \multicolumn{1}{c|}{\textcolor{green}{100\%}/0.82} & \multicolumn{1}{c|}{100\%/0.84}   & 100\%/0.91    & 90\%/0.11                        & \textcolor{red}{78\%}/0.14                   & 100\%/0.44                   & 100\%/0.40                           \\ \hline
\multirow{7}{*}{UNSW-NB15}   & \multicolumn{2}{c|}{FGSM}                          & \textcolor{red}{0\%}/0.08                           & \multicolumn{1}{c|}{\textcolor{green}{100\%}/0.64} & \multicolumn{1}{c|}{100\%/0.70}   & 100\%/0.73    & 0\%/0.08                        & 0\%/0.07                   & 10\%/0.22                   & 40\%/0.27                           \\ \cline{2-11} 
                             & \multicolumn{1}{c|}{\multirow{3}{*}{CW}}  & $c=1e-2$   & \textcolor{red}{0\%}/0.12                           & \multicolumn{1}{c|}{\textcolor{green}{100\%}/0.80} & \multicolumn{1}{c|}{100\%/0.87}   & 100\%/0.90    & 0\%/0.11                        & 0\%/0.08                   & 49\%/0.32                   & 82\%/0.40                           \\ 
                             & \multicolumn{1}{c|}{}                     & $c=1e-3$   & \textcolor{red}{0\%}/0.12                           & \multicolumn{1}{c|}{\textcolor{green}{100\%}/0.80} & \multicolumn{1}{c|}{100\%/0.87}   & 100\%/0.90    & 0\%/0.11                        & 0\%/0.08                   & 12\%/0.29                   & 33\%/0.34                           \\  
                             & \multicolumn{1}{c|}{}                     & $c=1e-4$   & \textcolor{red}{0\%}/0.12                           & \multicolumn{1}{c|}{\textcolor{green}{100\%}/0.80} & \multicolumn{1}{c|}{100\%/0.87}   & 100\%/0.90    & 0\%/0.11                        & 0\%/0.08                   & 2\%/0.28                   & 10\%/0.30                           \\ \cline{2-11} 
                             & \multicolumn{1}{c|}{\multirow{3}{*}{PGD}} & $\epsilon=10/255$ & \textcolor{red}{0\%}/0.10                           & \multicolumn{1}{c|}{\textcolor{green}{100\%}/0.76} & \multicolumn{1}{c|}{100\%/0.82}   & 100\%/0.86    & 0\%/0.09                        & 0\%/0.05                   & 48\%/0.29                   & 82\%/0.36                           \\  
                             & \multicolumn{1}{c|}{}                     & $\epsilon=40/255$ & \textcolor{red}{0\%}/0.06                           & \multicolumn{1}{c|}{\textcolor{green}{100\%}/0.64} & \multicolumn{1}{c|}{100\%/0.70}   & 100\%/0.73    & 0\%/0.06                        & 0\%/0.02                   & 8\%/0.21                   & 36\%/0.27                           \\  
                             & \multicolumn{1}{c|}{}                     & $\epsilon=80/255$ & \textcolor{red}{0\%}/0.03                           & \multicolumn{1}{c|}{\textcolor{green}{100\%}/0.49} & \multicolumn{1}{c|}{100\%/0.54}   & 100\%/0.57    & 0\%/0.03                        & 0\%/0.02                   & 1\%/0.13                   & 1\%/0.17                           \\ \hline
\end{tabular}}
\end{table*}

\begin{table*}[htbp]
\caption{Original anomaly score of the reconstructed malicious data against Kitsune model.}
\label{white-kitsune-score}
\resizebox{\linewidth}{!}
{\begin{tabular}{|c|c|c|ccc|c|c|c|c|}
\hline
\multirow{2}{*}{\diagbox[]{Dataset}{Defense}} & \multirow{2}{*}{Threshold} & \multirow{2}{*}{No Defense} & \multicolumn{3}{c|}{FedDef}& \multirow{2}{*}{Soteria} & \multirow{2}{*}{GP} & \multirow{2}{*}{DP} & \multirow{2}{*}{Instahide} \\ \cline{4-6}
& & & \multicolumn{1}{c|}{$\alpha=1$} & \multicolumn{1}{c|}{$\alpha=0.5$} & $\alpha=0.25$ & & & & \\ \hline 
KDD99 & 0.30 & 0.27 & \multicolumn{1}{c|}{1.06} & \multicolumn{1}{c|}{1.16} & 1.13 & 0.27 & 0.25 & 0.43 & 0.76 \\ 
Mirai & 0.35 & 0.41 & \multicolumn{1}{c|}{0.77} & \multicolumn{1}{c|}{0.83} & 0.81 & 0.40 & 0.47 & 0.57 & 0.57 \\ 
CIC-IDS2017 & 0.08 & 0.26 & \multicolumn{1}{c|}{1.24} & \multicolumn{1}{c|}{1.28} & 1.37 & 0.28 & 0.41 & 0.77 & 0.77 \\ 
UNSW-NB15 & 0.30 & 0.12 & \multicolumn{1}{c|}{0.80} & \multicolumn{1}{c|}{0.87} & 0.90 & 0.11 & 0.08 & 0.32 & 0.40 \\ \hline
\end{tabular}}
\end{table*}

For better comparison, we also present the anomaly score for each defense before the adversarial attack in TABLE~\ref{white-kitsune-score}.
From defenders' view, FedDef still outperforms all the baselines in terms of both accuracy and score.
We can find that the anomaly score for FedDef before the attack is already high enough (0.77-1.24 when $\alpha=1$) and thus model accuracy is always 100\%.
In contrast, for KDD99 and UNSW datasets, the original scores for no defense (0.27/0.12), Soteria (0.27/0.11) and GP (0.25/0.08) are low enough to evade Kitsune, let alone the adversarial score.
While DP and Instahide present relatively higher scores, yet adversarial perturbation can still render some scores lower than the threshold, resulting in 16\%-81\% and 85\%-100\% accuracy for KDD99 respectively.
This corresponds with the privacy score in Fig. \ref{reconstruction_score_comparison}. In other words, FedDef outperforms DP and Instahide, which again outperform GP and Soteria.

However,
the threshold for CIC2017 is so low (0.08). Even PGD attack can hardly render the adversarial examples successful.
PGD with $\epsilon=\frac{80}{255}$ perturbation range can only achieve 91\% for no defense.
Yet we can still come to similar conclusion as DNN model analysis for other three datasets for each defense.

From attackers' view, PGD attack still works better than FGSM and CW.
Specifically, given the same $\epsilon=\frac{40}{255}$, PGD can achieve lower score and accuracy (e.g., 4\%/0.31 and 8\%/0.32 for Mirai dataset for no defense).
While CW with lower distance constraint $c$ can also reach better evasion rate, PGD generally works better if we correspond $c$ and $\epsilon$ in order (e.g., 0\%/0.16-25\%/0.25 and 0\%/0.24-25\%/0.27 for KDD99 for no defense).

Note that CW attack also requires successful evasion (DNN's target classification or scores lower than the threshold) to actually update the adversarial example, therefore, it also cannot perturb FedDef's reconstructed data due to our better privacy protection (the same anomaly score as the original).
While PGD attack can always set higher $\epsilon$ to ensure successful evasion (at higher perturbation price), which is more feasible than CW since constraint $c$ has a limit 0.

\subsubsection{Conclusion}
We briefly conclude our privacy analysis.
In general, our defense outperforms all current defenses with high privacy score and low evasion rate even with strong model performance guarantee ($\epsilon=0$) for single or multiple samples, in both training stages, against either privacy attack, and the following black-box and white-box adversarial attacks.

\begin{table*}[htbp]
\caption{Accuracy comparison on KDD99 datasets for different parameters with baseline 0.996, higher is better.}
\label{ablation_acc}
\resizebox{\linewidth}{!}
{\begin{tabular}{|c|c|c|c|c|c|c|c|c|c|c|}
\hline
\diagbox[]{$def\_lr$}{$def\_ep$}    & 10    & 20    & 30    & 40    & 50    & 60    & 70    & 80    & 90    & 100   \\ \hline
3e-2 & 0.442$\pm$0.186 & 0.534$\pm$0.087 & 0.917$\pm$0.040 & 0.965$\pm$0.015 & 0.977$\pm$0.003 & 0.983$\pm$0.002 & 0.985$\pm$0.002 & 0.987$\pm$0.002 & 0.987$\pm$0.001 & 0.987$\pm$0.001 \\ 
8e-2 & 0.497$\pm$0.193 & 0.968$\pm$0.003 & 0.980$\pm$0.003 & 0.986$\pm$0.002 & 0.987$\pm$0.002 & 0.989$\pm$0.001 & 0.989$\pm$0.001 & 0.988$\pm$0.001 & 0.990$\pm$0.002 & 0.989$\pm$0.001 \\ 
2e-1  & 0.967$\pm$0.006 & 0.984$\pm$0.002 & 0.984$\pm$0.002 & 0.987$\pm$0.001 & 0.987$\pm$0.001 & 0.988$\pm$0.002 & 0.990$\pm$0.001 & 0.987$\pm$0.001 & 0.990$\pm$0.001 & 0.992$\pm$0.001 \\ \hline
\end{tabular}}
\end{table*}

\subsection{Ablation Study}
\label{ablation_study}
So far, we have demonstrated our defense's model performance and privacy preserving guarantee.
Next we will study some parameter impact on the overall performance of FedDef.

\textbf{Setup 6}.
We evaluate model accuracy and privacy score (plus label accuracy) on KDD99 dataset to study the optimal parameter combination with respect to $\alpha$, $def\_lr$ and $def\_ep$ with determined $\epsilon$ and $\delta$.
Firstly, we choose our optimal $\alpha=1$ because it induces the best accuracy performance and privacy score is high enough to mitigate the adversarial attack.
We follow the same experiment setting for model performance and privacy evaluation as introduced in Section \ref{convergence_results} and \ref{privacy_results}.

\begin{figure}[!h]
\setlength{\abovecaptionskip}{0pt}
\setlength{\belowcaptionskip}{0pt}
\subfigure{
\includegraphics[width=1\linewidth]{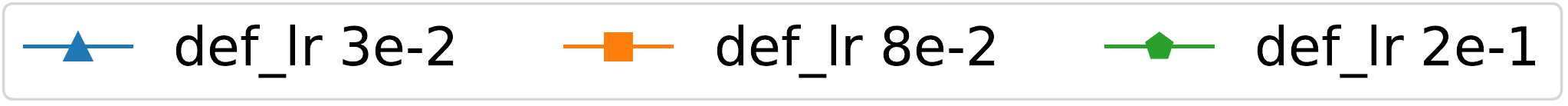}
}
\setcounter{subfigure}{0}
\subfigure[Early+Inversion]{
\includegraphics[width=0.45\linewidth]{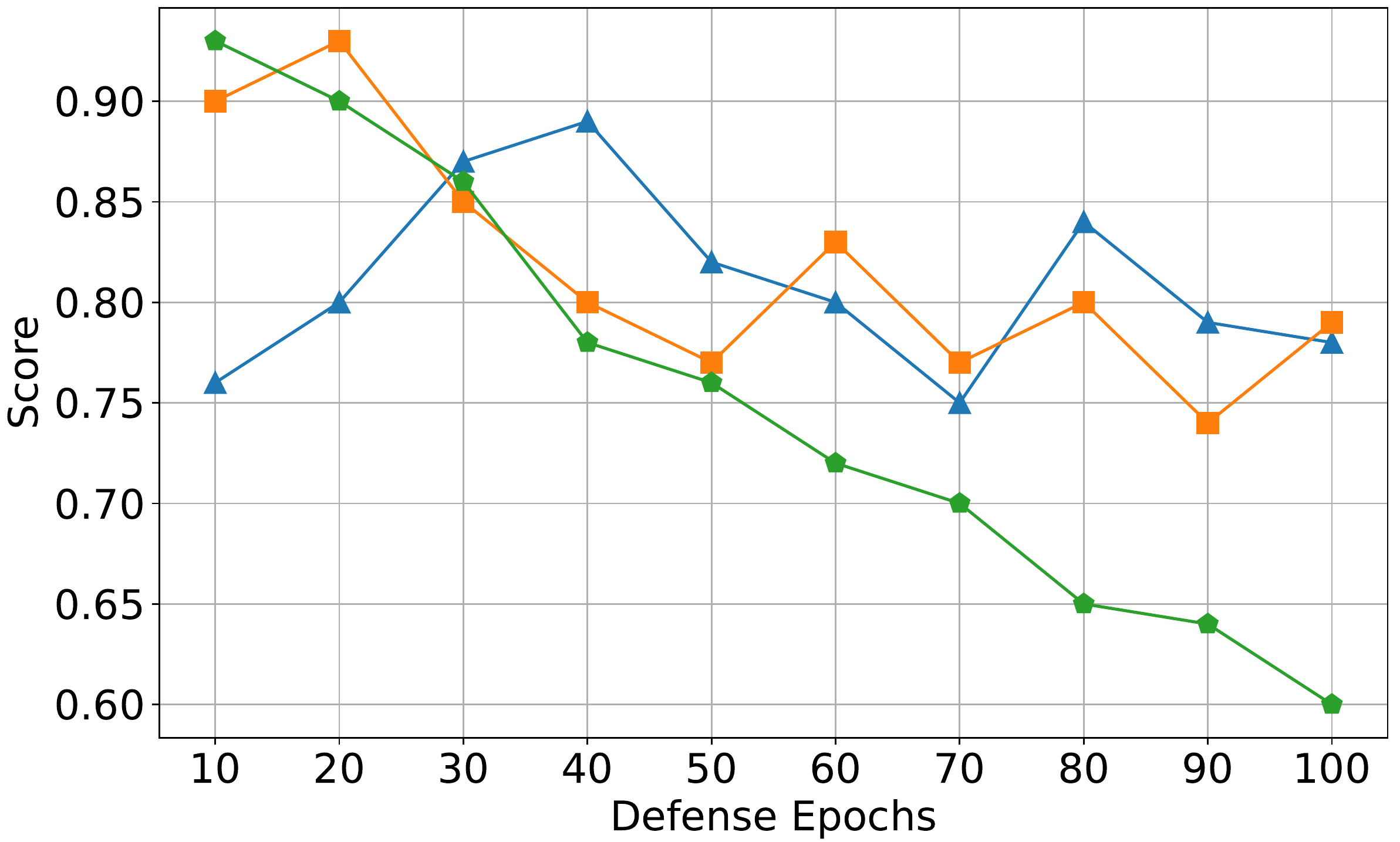}
}
\subfigure[Early+Label Accuracy]{
\includegraphics[width=0.45\linewidth]{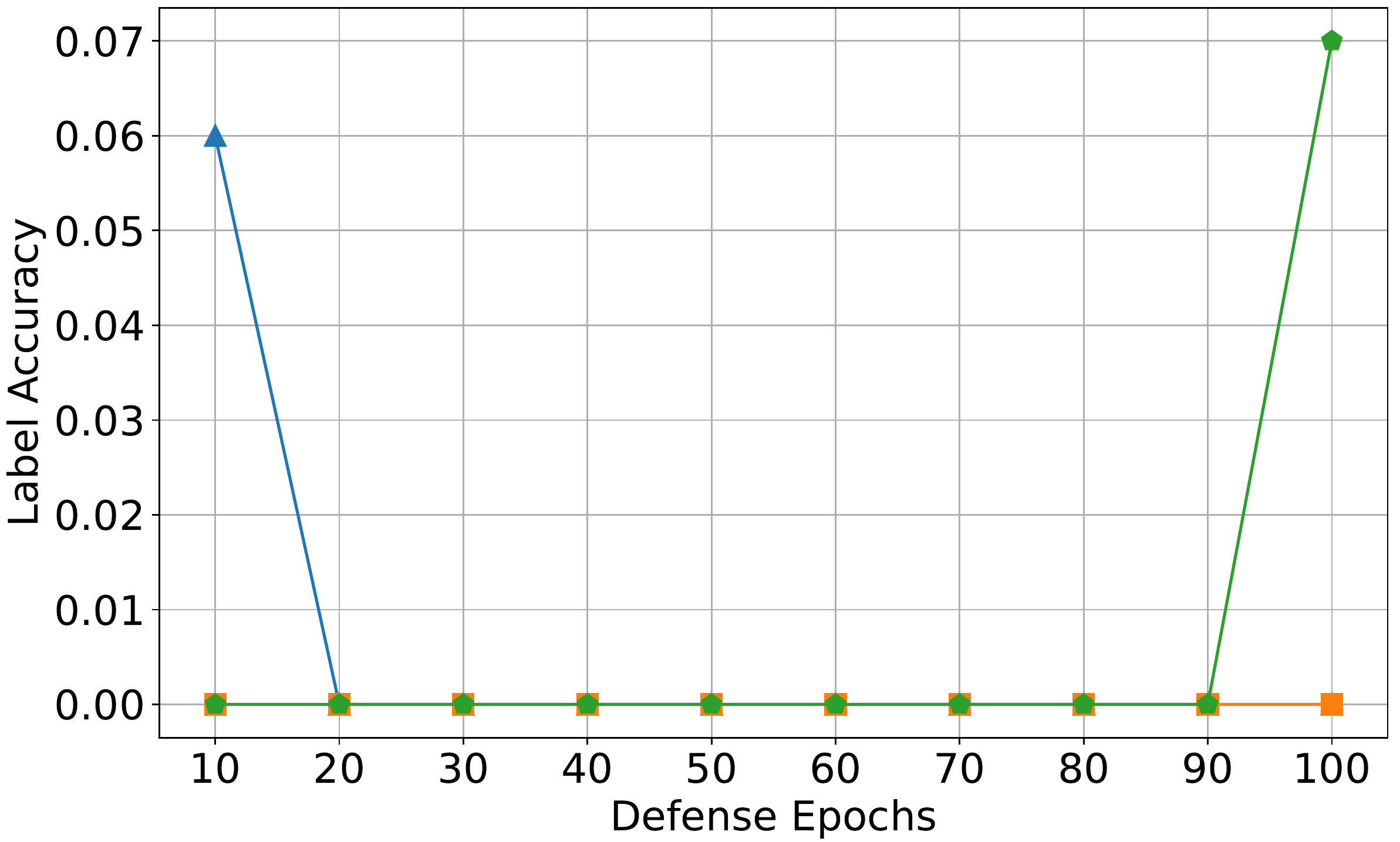}
}
\caption{Ablation study on $def\_lr$ during early training stage on KDD99 dataset with inversion attack for single sample reconstruction.}
\label{ablation_kdd}
\end{figure}

\textbf{Ablation Analysis}.
The average privacy score and accuracy results over five training times can be found in TABLE \ref{ablation_acc} and Fig. \ref{ablation_kdd}.
We can find that higher learning rate generally reduces the steps needed to fully optimize pseudo data.
Specifically, it takes $def\_lr=3e-2, 8e-2, 2e-1$ about 70, 50, 40 steps accordingly to achieve model accuracy as high as 0.987, while also achieving similar privacy score as high as 0.78.
However, when $def\_lr$ is 0.2, more optimization steps can induce privacy loss (only 0.6 score for $def\_ep=100$ and 0.07 reconstructed label accuracy) during early stage (see Fig.~\ref{ablation_kdd}) which means pseudo data is over-optimized.

Due to the above consideration, we recommend predetermined constants $\epsilon=0$, $\delta=1$, $g\_value=1e-15$, and key parameter $\alpha=1$, $def\_lr=0.2$ and $def\_ep=40$ as our optimal parameter combination, where the model accuracy is maintained and privacy score is high enough to mitigate any reconstruction attack and the following adversarial attack.

\subsection{Defense Overhead}
\label{test_time}
Finally, we conduct experiments to present defenses' computation and memory overhead.

\textbf{Setup 7}.
We first evaluate algorithm running time for all the defenses, where only one local user trains a global DNN model on KDD99 dataset with full training dataset for 100 epochs.
To further demonstrate FedDef's advantage, we also compare FedDef alone with traditional homomorphic encryption's (HE) performance in terms of computation and memory overhead.

Specifically, we adapt from an open-source github project \cite{PythonPaillier} named phe, which implements the Paillier Partially Homomorphic Encryption.
It first generates both public and private key, where local user encrypts their gradients with public key and decrypts the aggregated gradients with private key.
Such encryption and decryption can induce extra overhead during FL.
We also follow the same training setting as other defenses and record the absolute training time and memory.

For complete results, we also present the privacy evaluation overhead that includes GAN training (100 epochs) for black-box attack and PGD optimization (100 steps) for white-box attack.
For all the experiments, we repeat the process for five times and report the average results.

\begin{figure}
\setlength{\abovecaptionskip}{0pt}
\setlength{\belowcaptionskip}{0pt}
\centering
\subfigure{
\begin{minipage}{0.8\linewidth}
\includegraphics[width=1\linewidth]{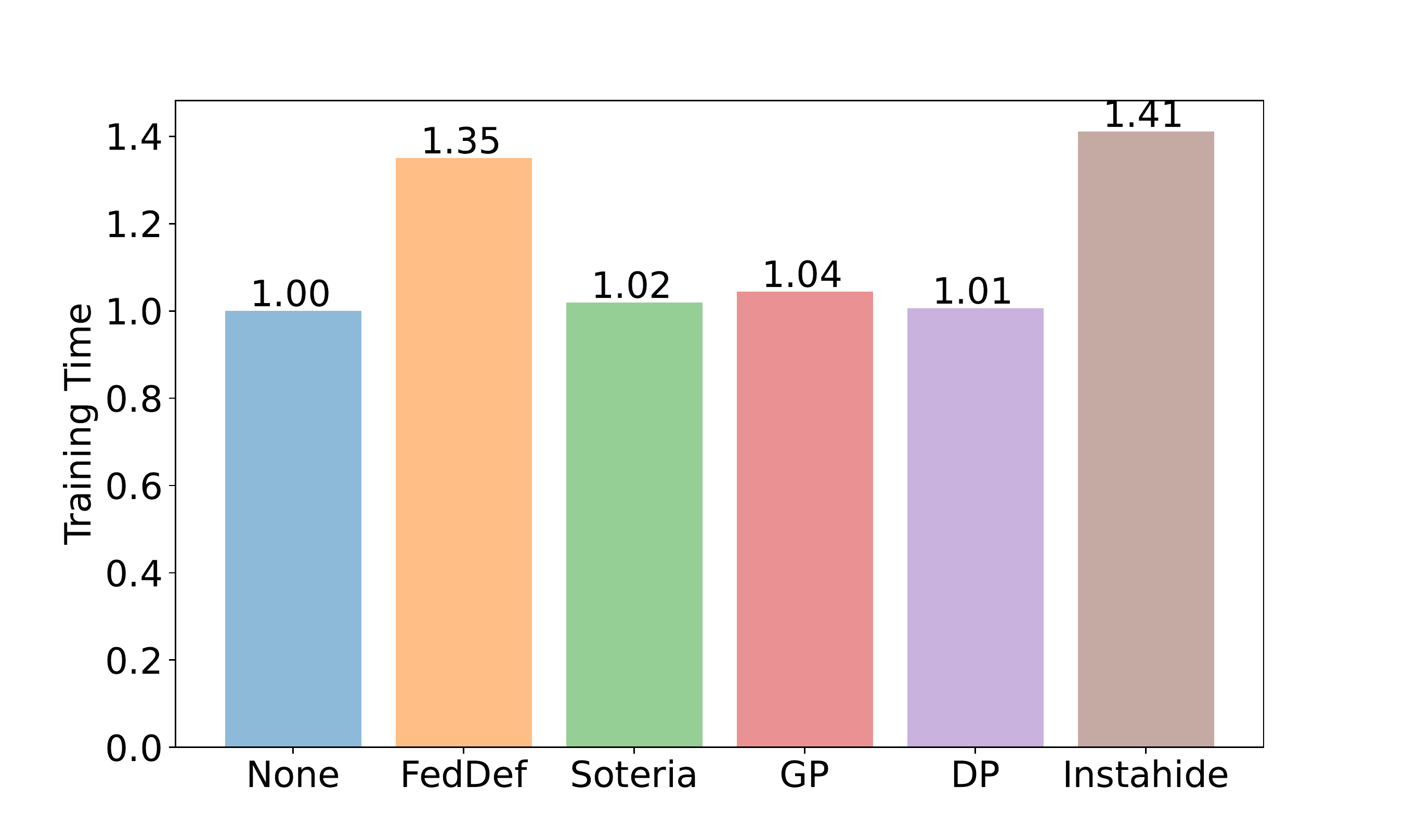}
\end{minipage}}
\caption{Relative training time comparison among baselines.}
\label{running_time}
\end{figure}

\begin{figure}
\setcounter{subfigure}{0}
\subfigure[Time Overhead]{
\begin{minipage}{0.48\linewidth}
\includegraphics[width=1\linewidth]{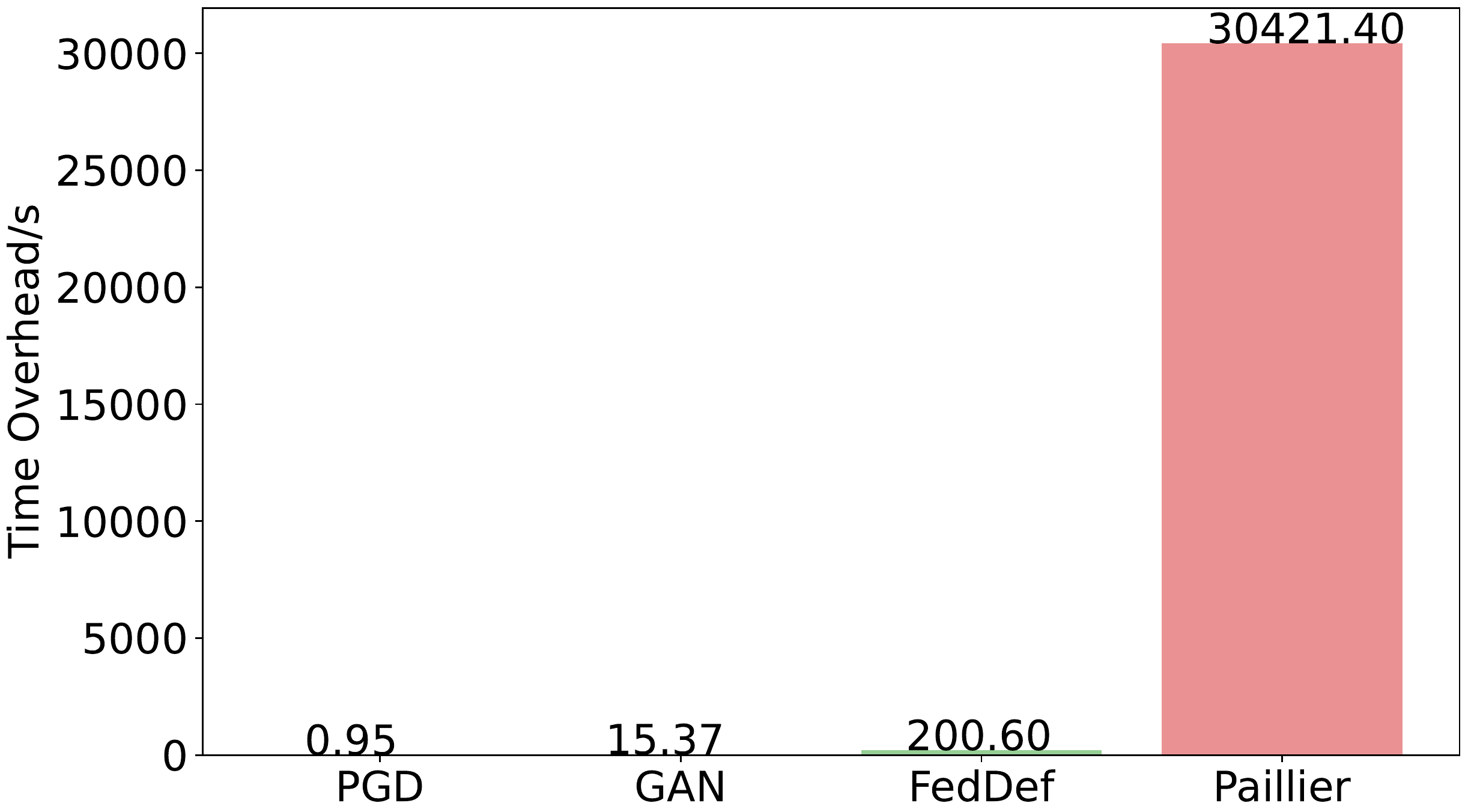}
\end{minipage}}
\subfigure[Memory Overhead]{
\begin{minipage}{0.48\linewidth}
\includegraphics[width=1\linewidth]{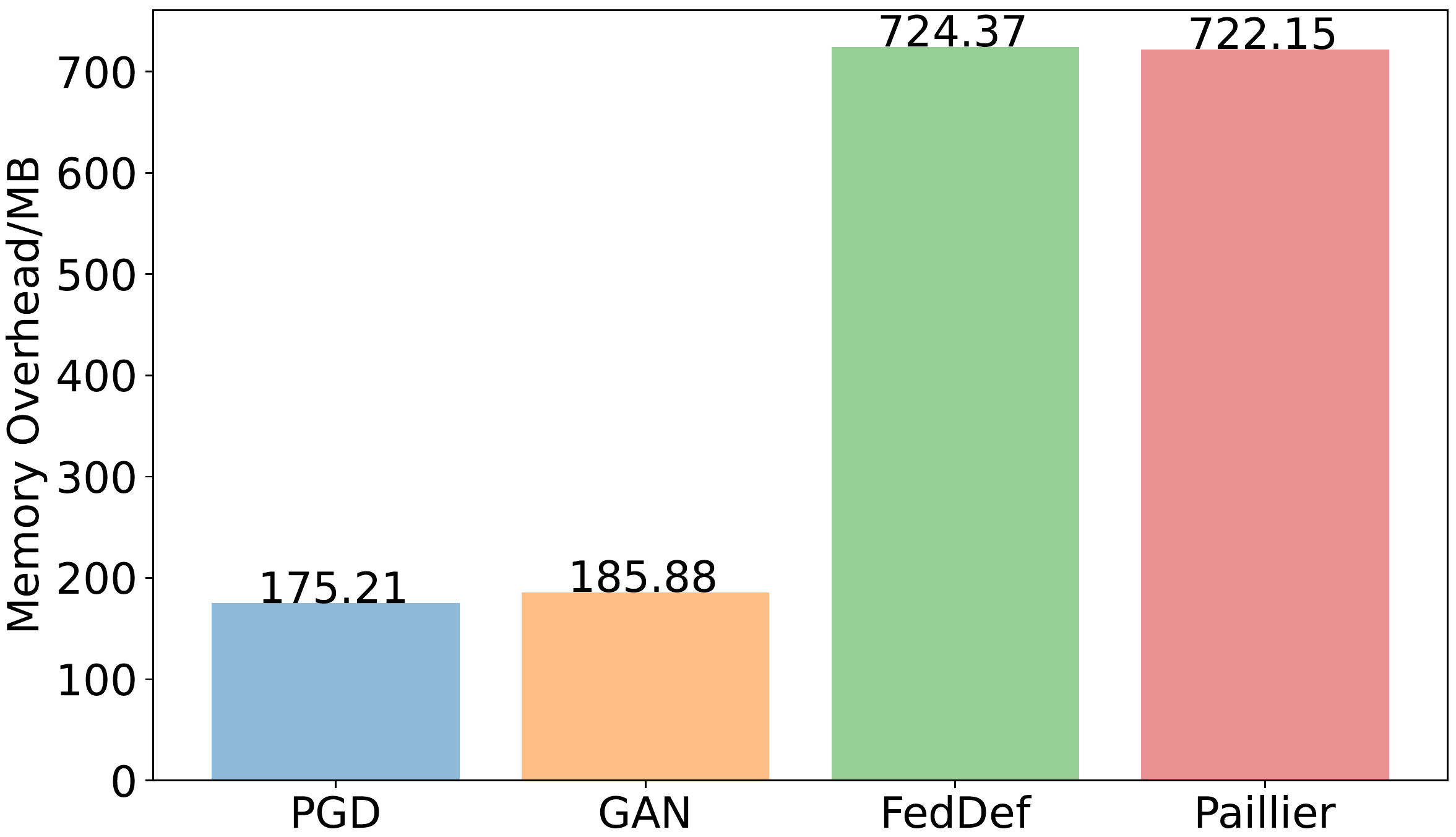}
\end{minipage}}
\caption{Absolute time and memory overhead comparison between FL (FedDef and Paillier) and adversarial attack (PGD and GAN).}
\label{overhead}
\end{figure}

\textbf{Overhead Analysis}.
Fig.~\ref{running_time} illustrates the relative training time under our optimal parameters.
We can find that our defense induces additional 0.35 time in exchange for model performance and privacy protection guarantee.
Such computation overhead is reasonable because optimization-based defense iteratively searches for optimal pseudo data while other defenses such as differential privacy directly injects noises and does not induce additional computation.

We also present Paillier's overhead along with our privacy evaluation process in Fig.~\ref{overhead}.
In general, FL process (FedDef and Paillier) can take up more time and memory than adversarial attack (PGD and GAN).
Specifically, Paillier is significantly worse and induces about 150 times extra training time ($3,0421s$) than FedDef ($200s$), while our privacy evaluation is also feasible to perform and only takes $<1s$ for white-box PGD attack and around $15s$ even for GAN-based attack.
As for memory occupation, FedDef and Paillier share similar results (722-724 MB), while PGD and GAN-based attack only take up 175-185 MB memory.

The results above demonstrate that our defense has great advantage over HE-based methods in terms of time overhead, and our privacy evaluation (white-box or GAN-based) is feasible to perform that only requires limited resources.

\section{Discussion}
\label{future_work}

\subsubsection{Additional computation overhead}
The main reason for additional training time is that our defense transforms the data for every local update with high gradient computation overhead in Eq.~(\ref{model_optimization}).
Our future work can try improving Algorithm~\ref{algorithm_feddef} by transforming private data only once or distilling the model parameters to reduce pseudo gradient optimization.

\subsubsection{Traffic space attack}
For black-box attack, we have only leveraged GAN to generate adversarial features instead of real traffic data to attack the NIDSs.
However, generating more practical traffic is beyond the scope of our work, and we can adapt current approaches like GAN+PSO \cite{han2021evaluating} to further improve the attacks and thus better evaluate the defenses.
\section{Conclusions}
\label{conclusion}
In this work, we first propose two privacy metrics specifically designed for FL-based NIDS, i.e., privacy score from reconstruction attacks and evasion rate from GAN-based adversarial attack, to derive an accurate evaluation of privacy protection and demonstrate the insufficiency of existing defenses.
To build a more robust FL-based NIDS, we further propose a novel input perturbation optimization defense strategy, i.e., FedDef, which aims to generate pseudo data to maintain model utility and constrain gradient deviation to provide privacy protection.
We also give theoretical analysis for utility and privacy guarantee with our defense.
The experimental results illustrate that our defense outperforms existing defenses and proves great privacy protection during both training stages against both attacks while also maintaining model accuracy loss within 3\%.

\ifCLASSOPTIONcaptionsoff
  \newpage
\fi

\bibliographystyle{IEEEtran}
\bibliography{Bibliography}
\vfill

\end{document}